\documentclass[12pt]{article}

\usepackage[utf8]{inputenc}
\usepackage[T1]{fontenc}
\usepackage{amsmath,amssymb,amsthm}
\usepackage{mathtools}
\usepackage{booktabs}
\usepackage{natbib}
\usepackage[margin=1in]{geometry}
\usepackage{hyperref}
\usepackage{algorithm}
\usepackage{algpseudocode}
\usepackage{enumitem}
\usepackage{caption}
\usepackage{threeparttable}

\newtheorem{theorem}{Theorem}[section]
\newtheorem{proposition}[theorem]{Proposition}
\newtheorem{lemma}[theorem]{Lemma}

\theoremstyle{definition}
\newtheorem{definition}[theorem]{Definition}

\theoremstyle{remark}
\newtheorem{remark}[theorem]{Remark}
\theoremstyle{definition}\newtheorem{principle}[theorem]{Principle}

\newcommand{\E}{\mathbb{E}}
\newcommand{\Var}{\mathrm{Var}}
\newcommand{\Cov}{\mathrm{Cov}}
\newcommand{\CV}{\mathrm{CV}}
\newcommand{\NB}{\mathrm{NegBin}}
\newcommand{\Mult}{\mathrm{Multinomial}}
\newcommand{\Dir}{\mathrm{Dirichlet}}
\newcommand{\Gam}{\mathrm{Gamma}}
\newcommand{\Bet}{\mathrm{Beta}}
\newcommand{\Pois}{\mathrm{Poisson}}
\newcommand{\Nobs}{N^{\mathrm{obs}}}
\newcommand{\Nibnr}{N^{\mathrm{IBNR}}}
\newcommand{\Xobs}{X^{\mathrm{obs}}}
\newcommand{\Sibnp}{S^{\mathrm{IBNP}}}
\newcommand{\Wobs}{W^{\mathrm{obs}}}
\newcommand{\bpi}{\boldsymbol{\pi}}

\title{Conditional Prediction for Macro-Level Claims Reserving:\\
What a Single Paid Triangle Identifies}
\author{Robin Van Oirbeek\thanks{University of Antwerp,
Department of Mathematics, Antwerp, Belgium.
Email: \texttt{robin.vanoirbeek@gmail.com}.}
\and
Tim Verdonck\thanks{University of Antwerp -- imec,
Department of Mathematics, Antwerp, Belgium;
KU Leuven, Department of Mathematics, Leuven, Belgium.
Email: \texttt{tim.verdonck@uantwerpen.be}.}}
\date{September 2026}

\begin{document}
\maketitle


\begin{abstract}

In macro-level claims reserving, reserve uncertainty is assessed from a single realised run-off triangle. The relevant quantity is therefore the conditional predictive distribution $p(R \mid D)$ of the outstanding amount given that triangle. We state the requirements for
targeting it and supply, in a Dirichlet--Gamma hierarchy, a
closed-form conditional predictive that attaches to any method producing
a development pattern (Chain-Ladder, Bornhuetter--Ferguson, Cape Cod),
for the latter two of which no conditional predictive bootstrap was
available. Three results bound what that construction can deliver.
First, the concentration $c$ governing allocation uncertainty is
estimated from partial-column proportions that form a Dirichlet
subcomposition, carrying concentration $cF_k$ rather than $c$; the
corrected estimator is consistent but slow ($\E[\hat c]/c = 1.31$ at
$I = 10$), and substituting it costs $6.4$ coverage points at that
triangle size, of which integrating over its posterior recovers about
two. Second, the closed form is exact only under diffuse anchoring, so an
oracle supplied with the true parameters remains conditionally
miscalibrated by roughly ten points at the ultimate dispersion typical of
macro-level portfolios. Third, and against the intuition that motivated
this work, a residual bootstrap with free accident-year effects does not
omit accident-year frailty uncertainty: its resampling of the diagonal
reproduces the diffuse-limit posterior of the row level, and under the
count hierarchy its coverage is nominal at every frailty variance tested.
The under-coverage that Meyers (2015) documents on paid data is
therefore not explained by omitted frailty, and the residual bootstrap is
a correctly targeted benchmark at practical triangle sizes. The paper's
contribution is the accounting: which parameters a single paid triangle
determines, what each substitution costs, and where the closed form
itself gives way.


\end{abstract}

\section{Introduction}\label{sec:intro}

\subsection{The conditional predictive and the sampling distribution}
\label{sec:intro_conditioning}

In macro-level claims reserving, the reserve of a portfolio is booked on a
single realised run-off triangle. The uncertainty that is relevant for such a
booking is therefore the spread of the outstanding amount \emph{given that
triangle}, and in any well-specified probabilistic model the latter is the
conditional predictive distribution
\begin{equation}\label{eq:cond_target}
p(R \mid \mathcal{D})
\;=\; \int p(R \mid \mathcal{D}, \theta)\, d\pi(\theta \mid \mathcal{D}),
\end{equation}
which holds the observed triangle $\mathcal{D}$ fixed as conditioning
information and integrates the model's conditional law over whatever
uncertainty about $\theta$ the triangle leaves unresolved.

Two requirements have to be met to attain \eqref{eq:cond_target}. The
first is that $\mathcal{D}$ enters as conditioning information. The
second is that $\pi(\theta \mid \mathcal{D})$ is integrated over instead
of being replaced by an estimate $\hat\theta$, which is what a plug-in
construction does not do: the plug-in predictive
$p(R \mid \mathcal{D}, \hat\theta)$ is the degenerate-posterior limit of
\eqref{eq:cond_target}, and is as such calibrated only when $\theta$ is
known. The cost of the second violation is measured throughout: $6.4$
coverage points at $I = 10$ for the concentration alone, falling to
$2.1$ by $I = 30$.

The standard residual bootstraps for claims reserving, i.e.\ the ODP
bootstrap of \citet{england2002} and the Mack bootstrap, resample the
observed triangle. At the outset of this work we read that resampling as
a violation of the first requirement, treating $\mathcal{D}$ as a random
draw rather than as fixed information, and we conjectured that the
resulting procedure omits the posterior uncertainty about accident-year
frailty and under-covers accordingly. Section~\ref{sec:sim_frailty}
tests that conjecture under the count hierarchy in which its omitted
component can be computed exactly, and rejects it. A residual bootstrap
with free accident-year effects reproduces, through its resampling of the
latest diagonal, the diffuse-limit posterior of the row level; through
its refit of the development factors, the estimation uncertainty in the
development pattern; and through its process draw, the process variance.
Those are the three components of \eqref{eq:cond_target} under the
hierarchy of this paper in the Chain-Ladder limit
(Proposition~\ref{prop:boot_decomp}), and its coverage under that
hierarchy is nominal at every frailty variance tested. Resampling the
triangle is a device for integrating over $\pi(\theta \mid \mathcal{D})$
under a diffuse prior, not a departure from conditioning on
$\mathcal{D}$.

The distinction that survives is narrower and, we think, more useful.
Under an informative anchor the residual bootstrap has no channel through
which the anchor enters, since it resamples a fitted development
structure that Bornhuetter--Ferguson and Cape Cod do not supply. And for
the parameters a residual bootstrap does not refit, i.e.\ the allocation
concentration of the present hierarchy and the pooled loss ratio of Cape
Cod, the plug-in question of the second requirement is the one that
remains. The conditional construction of this paper is therefore not an
alternative to the residual bootstrap on Chain-Ladder-anchored paid
triangles, where the latter is the better-calibrated procedure at
practical sizes (Section~\ref{sec:sim_odp}); it is the construction that
extends conditional prediction to the anchors the residual bootstrap
cannot reach, together with an accounting of what a single paid triangle
determines about the parameters it needs.

\subsection{Why the Dirichlet--Gamma hierarchy}\label{sec:intro_hierarchy}

The hierarchy adopted here separates two sources of uncertainty with distinct
actuarial meaning: the ultimate paid amount
$S_i \sim \Gam(\alpha E_i, \beta)$ captures how much will eventually be paid,
while the allocation $(X_{i,\cdot}/S_i) \sim \Dir(c\bpi)$ captures when it
will be paid. Once the observed cumulative payment $\Xobs_i$ is in hand the
first is partially resolved, and what remains is purely allocational: the
outstanding amount is
\[
\Sibnp_i \;=\; \Xobs_i \cdot \frac{1 - W_i}{W_i},
\qquad
W_i \sim \Bet\bigl(cF_{I-i},\, c(1-F_{I-i})\bigr),
\]
with $W_i$ the fraction of the ultimate already paid. The triangle enters once,
as a parameter of the Beta, and is never treated as random again, hereby
making the conditioning requirement implementable: the conditional law is
available in closed form, requiring neither approximation by resampling nor a
fitted mean structure.

The construction is \emph{anchor-agnostic}. The development proportions
$F_{I-i}$ that parametrise the Beta may come from Chain-Ladder,
Bornhuetter--Ferguson, Cape Cod, a GLM or any other procedure producing a
cumulative pattern, so the same predictive attaches to each. For
Bornhuetter--Ferguson and Cape Cod, analytic prediction errors exist
\citep{Mack2008, AlaiMerzWuthrich2009, AlaiMerzWuthrich2010} but no
conditional predictive bootstrap did.

The predictive depends on the data through the development pattern and a
single scalar, i.e.\ the concentration $c$. Note that this parsimony is at the
same time the principal difficulty: a single paid triangle determines $c$
poorly, and Section~\ref{sec:c_estimation} characterises how poorly.

The closed form is, however, exact only in the diffuse limit
$\alpha, \beta \to 0$. Under an informative prior on the ultimate the
conditional law of $W_i$ given $\Xobs_i$ is tilted, and the Beta draw is a
plug-in approximation to it. The approximation is adopted deliberately and its cost is recorded in
Section~\ref{sec:conditional_coverage}: roughly ten
coverage points of conditional miscalibration at the ultimate dispersion
typical of macro-level portfolios, vanishing as the prior becomes diffuse, and
present even when every parameter is known exactly. The tilted density is available, but does not yield the Beta
representation on which the anchor-agnostic property rests. Hence, the closed
form and the conditional calibration cannot be attained jointly.

Finally, Chain-Ladder, Bornhuetter--Ferguson and Cape Cod emerge from this one
hierarchy as credibility estimators under diffuse, informative and pooling
priors respectively, with identical structure for counts and amounts. That
correspondence is a structural consequence of the hierarchy rather than an
additional assumption; it is established in \citet{VanOirbeek2026cred} and
used here to justify applying a single predictive across the three anchors.
It is the Dirichlet--Gamma analogue of the credibility Chain-Ladder of
\citet{GislerWuthrich2008}, in which Chain-Ladder and
Bornhuetter--Ferguson arise as limiting cases of one credibility estimator
within Mack's model.

\subsection{Contributions}

Three contributions are made in this paper: an identification result for
the parameter that governs allocation uncertainty, a construction for the
conditional predictive that operates uniformly across the major
macro-level reserving methods together with the calibration it actually
achieves, and a negative result about residual bootstraps that corrects
a natural conjecture, including our own.

The negative result is stated first because it reorients the rest.
Under a data-generating process with accident-year frailty
$U_i \sim \Gam(\kappa,\kappa)$, one might expect the free accident-year
effects of the ODP bootstrap to absorb the realised frailty of each
observed row and thereby to lose the posterior uncertainty in $U_i$ for
the least-developed years. Proposition~\ref{prop:boot_decomp} shows the
opposite: in the count hierarchy of Section~\ref{sec:counts} the variance
generated by resampling the latest diagonal equals the diffuse-limit
posterior variance of the row level, the refit of the development
factors carries the estimation uncertainty in $\bpi$, and the process
draw carries the process variance, so that the three add to the
conditional predictive variance of Theorem~\ref{thm:ibnr} in its
Chain-Ladder limit. Section~\ref{sec:sim_frailty} confirms this with the
omitted component computed exactly on every simulated triangle: as
$\Var(U_i)$ goes from $0$ to $1$ the conjectured omission would predict
$95\%$ coverage falling to $86.4\%$, and the measured coverage
stays at $93.6$--$96.4\%$. A residual bootstrap with free
levels is a correctly targeted procedure for the diffuse-anchored
conditional predictive, and the systematic under-coverage that
\citet{Meyers2015} finds on paid Schedule~P data is not a frailty
effect. The decline in ODP coverage that the same sweep produces on the
amounts side is traced to the variance function instead: it is present
under a Dirichlet allocation, whose cell variance is quadratic in the row
level, and absent under a compound Poisson allocation with the ODP
variance function on the same rows
(Section~\ref{sec:sim_frailty}).

The second contribution is an identification result for the concentration.
The conditional predictive depends on the data only through $\Xobs_i$, the
development pattern, and a single scalar: the concentration $c$ of the
Dirichlet allocation. That scalar is estimated from the dispersion of
partial-column proportions across accident years, and those proportions form a
Dirichlet \emph{subcomposition}, carrying concentration $cF_k$ rather than $c$
(Lemma~\ref{lem:dirichlet_sub}). The estimator of
Proposition~\ref{prop:c_estimator} corrects for the factor $F_k$, without
which it is inconsistent for $c$. The corrected estimator is consistent, but
converges slowly: $\E[\hat{c}]/c$ equals $1.89$ at $I = 7$ and $1.31$ at
$I = 10$, scale-free in $c$, with a coefficient of variation exceeding one
half at $I = 10$ (Remark~\ref{rem:sigma_c_finite}). Its per-cell
asymptotic variance is available in closed form
(Proposition~\ref{prop:c_variance}),
\begin{equation}\label{eq:sigma_c_intro}
\sigma_{c,k}^2
= \frac{c\,(a+1)\bigl[2a(a-3)\,\pi(1-\pi) + 3a + 1\bigr]}
       {F_k\,\pi(1-\pi)(a+2)(a+3)},
\qquad a = cF_k,
\end{equation}
Of the three claims that this expression makes, two are confirmed by
simulation and one is not. The $\sqrt{n_k}$-rate holds, and the $\sigma_{c,k}^2 \sim 2c^2$ scaling holds
at every triangle size. As such, the precision of $\hat{c}$, and hence the
coverage of the resulting predictive, depends on $I$ and on the development
pattern, but not on the concentration itself. The constant does not hold at practical $I$: the
expression understates $\Var(\hat{c})$ by a factor of $2.9$ at $I = 10$ and
$15.8$ at $I = 7$ (Remark~\ref{rem:sigma_c_finite}). Note that a closed-form asymptotic variance invites use as a standard error,
which would here be optimistic by up to an order of magnitude.

The consequences for the predictive are measured throughout the paper.
Substituting $\hat{c}$ costs coverage at every triangle size, at the rate
measured in Section~\ref{sec:oracle}. A screen of the form $\hat{c} < t$ fires
on a portfolio sitting exactly at $t$ with probability $0.32$--$0.45$ rather
than $0.50$, since the bias displaces the whole detection curve
(Section~\ref{sec:c_threshold}). Integrating over the posterior of $c$
removes most of the $\hat{c}$-conditional miscalibration that the plug-in
induces (Section~\ref{sec:conditional_coverage}), while raising the
unconditional level by about two points only. Hence, roughly six points remain
at $I = 10$ that no further modelling of the same triangle removes
(Section~\ref{sec:regularised}).

The third contribution is the anchor-agnostic character of the construction.
Existing bootstrap procedures (ODP and Mack) operate only within the
Chain-Ladder family. For Bornhuetter--Ferguson, analytic prediction errors
exist: \citet{Mack2008} derives the BF prediction error in a distribution-free
model, and \citet{AlaiMerzWuthrich2009, AlaiMerzWuthrich2010} derive the BF
mean squared error of prediction in the ODP framework. No bootstrap targeting
the conditional predictive distribution has, however, been published for BF or
Cape Cod. The construction of Section~\ref{sec:bootstrap} takes any vector of
cumulative development proportions $\hat{F}$ as input, from Chain-Ladder, BF,
Cape Cod or any other procedure, and Section~\ref{sec:empirical} illustrates
it on a Bornhuetter--Ferguson example.

Applicability and calibration are distinct claims, and are established
separately. Under a correctly specified anchor the construction is
calibrated:
Bornhuetter--Ferguson covers within two points of Chain-Ladder on an interval
$53\%$ as wide (Section~\ref{sec:sim_anchor}). That precision is contingent on the anchor, since a $20\%$ error in the prior
loss ratio costs $25$ to $50$ coverage points and no amount of data corrects
it. The Cape Cod anchor, whose loss ratio is estimated from the triangle,
under-covers by $6.2$ points at $I = 10$ and $3.4$ at $I = 15$, since that
estimation error is not represented in the predictive.

The conditional calibration itself is treated last.
Principle~\ref{prin:bootstrap} states two requirements: that $\mathcal{D}$ is
held fixed, and that the residual uncertainty about $\theta$ given
$\mathcal{D}$ is integrated over instead of substituted. Both requirements
have measurable content. Partitioning the simulated triangles by $\hat{c}$, the plug-in
construction varies from $96\%$ to $78\%$ across terciles at $I = 10$;
integrating over the posterior of $c$ reduces that spread from $18.3$ to $5.3$
coverage points at unchanged interval width, and integrating over $(c, \bpi)$
jointly reduces it to $2.8$, against an oracle floor of $1.7$
(Section~\ref{sec:conditional_coverage}). On a partition formed from the dispersion of the estimated ultimates, a
statistic that no procedure uses as a parameter, an oracle supplied with both
true parameters is still miscalibrated by ten points. The source is
identified: the closed-form
conditional law is exact only under diffuse anchoring
(Remark~\ref{rem:diffuse_scope}), and the gradient vanishes as the ultimate's
prior becomes diffuse. Conditional calibration is therefore achieved with respect to the parameters
and not with respect to the anchoring.

On Chain-Ladder-anchored paid triangles the ODP bootstrap is the
better-calibrated procedure at practical triangle sizes, and
Section~\ref{sec:sim_odp} reports that comparison in full. Given
Proposition~\ref{prop:boot_decomp} this is not a paradox: both procedures
target the same conditional object there, and the residual bootstrap
integrates over the row level and the development pattern where the
plug-in construction substitutes the concentration and the pattern. The
construction of this paper earns its place on the anchors the residual
bootstrap cannot reach, and the cost it pays elsewhere is quantified
rather than hidden.

Finally, a number of supporting results are established. Under compound
Poisson misspecification the construction is conservative: in
the limit $p \to 1$ the predictive standard deviation exceeds the true value
by the factor $1/\sqrt{F_{I-i}}$ wherever the predictive variance exists, and
for $p \in (1,2)$ that factor is a lower bound
(Theorem~\ref{thm:conservative_bias}).
Sections~\ref{sec:simulations}--\ref{sec:empirical} provide empirical
validation on three contrasting benchmarks and on four misspecification
studies.

\subsection{Related literature}

\citet{Mack1993} established distribution-free Chain-Ladder uncertainty;
\citet{renshaw1998} and \citet{verrall2000} provided the GLM foundations for
stochastic Chain-Ladder; \citet{england2002} developed the ODP bootstrap;
\citet{wuthrich2008} covers stochastic reserving methods in detail. Note that
several stochastic models reproduce the Chain-Ladder point estimate
\citep{mack1994, mackventer2000, renshaw1998}; the Dirichlet--Gamma hierarchy
of the present paper adds to this list under a credibility interpretation.
Tweedie GLMs \citep{tweedie1984, Jorgensen1987, SmythJorgensen2002} are
standard for severity modelling, with \citet{mccullagh1989} the canonical
reference on the GLM framework itself. \citet{Clark2003} fitted smooth growth
curves via maximum likelihood; the compositional treatment of development
proportions adopted here follows \citet{aitchison1982}, and the Dirichlet
allocation for reserving is due to \citet{srirams2021} (see below). The
systematic miscalibration of ODP and Mack across the \citet{Meyers2015}
Schedule~P database provided the empirical motivation for the present
paper. We conjectured that omitted accident-year frailty was the cause
and Section~\ref{sec:sim_frailty} shows that it is not; on paid data
\citet{Meyers2015} himself finds the failure to be one of location, with
both models overstating the expected ultimate, and his changing
settlement rate model is a correction of the mean rather than of the
dispersion. The residual bootstrap's dispersion is, by
Proposition~\ref{prop:boot_decomp}, the right one for the diffuse-anchored
conditional predictive.

The word \emph{conditional} has a history in this literature that the
present paper should be placed against. The dispute between
\citet{BBMW2006} and \citet{MackQuargBraun2006}, with the Bayesian
treatment of \citet{Gisler2006} and the later reassessment of
\citet{Gisler2019}, concerned Mack's model and was not about whether to
decompose the conditional prediction error into process and estimation
terms, which every party did, but about how to average the estimation
term over the unknown development factors given the observed triangle.
Mack's estimator and the conditional-resampling estimator of
\citet{BBMW2006} differ by second-order cross terms, and \citet{Gisler2019}
concludes in favour of the former. The present construction sits on the
Bayesian side of that exchange: the diffuse-limit posterior of the row
levels and the estimation distribution of the pattern in
Proposition~\ref{prop:boot_decomp} are the Dirichlet--Gamma counterpart of
the limiting-prior posterior under which \citet{Gisler2006} recovers the
Chain-Ladder factors as Bayes estimators, and not the conditional
resampling of \citet{BBMW2006}. The proposition is silent on the
Mack--BBMW difference itself: the two estimators differ at second order in
the factor estimation error, and at the Monte Carlo resolution of
Section~\ref{sec:sim_mech1} (roughly $1.3$ coverage points at $M = 300$) a
difference of that order is not detectable through coverage.

The conditioning principle is stated against the background of classical
bootstrap theory, in which the bootstrap approximates the sampling
distribution of a statistic under repeated sampling
\citep{DavisonHinkley1997, Hall1992}. Prediction targets the conditional
law of an unobserved quantity given the realised data, and it is not
automatic that a resampling procedure attains it. For the residual
bootstraps of reserving it turns out that it does, in the diffuse-anchored
case, because resampling a triangle whose row levels are free parameters
is equivalent to drawing those levels from their flat-prior posterior
(Proposition~\ref{prop:boot_decomp}); the distinction becomes operative
only where the anchor is informative or where a parameter is not refitted. The distinction between the two coverage notions is drawn in
Section~\ref{sec:intro_conditioning} and not restated here; conditional
coverage is the notion adopted throughout, and it is also, as
Section~\ref{sec:conditional_coverage} shows, a notion no procedure
considered here attains in full.

A Dirichlet allocation model for loss reserving was previously introduced by
\citet{srirams2021}, who use it to unify Chain-Ladder and
Bornhuetter--Ferguson predictions under a common statistical framework and
discuss both frequentist and Bayesian inference on Worker's compensation
triangles. The present paper takes the same Dirichlet allocation as its
starting point, but pursues a different objective. Rather than unifying
Chain-Ladder and Bornhuetter--Ferguson at the point-estimate level, the
conditioning principle is formalised and the conditional predictive is
developed, inheriting its development pattern from any such method, Cape Cod
included. The partial-column proportions are furthermore identified as a
Dirichlet subcomposition, hereby correcting the resulting estimator of the
concentration and characterising what a single paid triangle does and does
not determine about that concentration. The
two papers are complementary: \citet{srirams2021} establishes the Dirichlet
allocation as a unifier at the point-estimate level; the present paper
establishes the conditional predictive for that allocation, and the limits of
its calibration.

The Compound Negative Binomial-Gamma paper \citep{VanOirbeek2026cnbg}
establishes the correctly specified parametric model for incremental paid
amounts; the present paper establishes the anchor-agnostic conditional
predictive that the former cites. Both address the systematic under-coverage
documented by \citet{Meyers2015}, be it with different answers, and the
division of labour between both is the following. The present paper
establishes what a free-level residual bootstrap does and does not
represent: its dispersion estimate carries no information about the
frailty variance $\kappa$, because the free levels absorb the realised
frailty of each row, but its resampling scheme nonetheless generates the
diffuse-limit posterior variance of those levels, which is the widest
member of the credibility family and hence never understates the frailty
component (Proposition~\ref{prop:boot_decomp}). What the residual bootstrap
cannot do is exploit a known or estimated $\kappa$ to narrow the
predictive, which is where an anchored model enters.
How far such a procedure can go is bounded by two distinct statements,
which should not be conflated. The first is a parameterisation fact: under
free accident-year effects a row-constant frailty $U_i$ enters the
likelihood only through the product $\mu_i U_i$, which is the free level
itself, so the frailty is collinear with the fitted levels and the
free-level dispersion estimate is a cell-level composite whatever the
estimator. This is an exact non-identification, not an information limit.
The second is a genuine information ceiling, established in
\citet{VanOirbeek2026cnbg} through Cram\'er--Rao and van Trees bounds on the
dispersion that the variance-curvature channel carries once levels are
anchored, together with a cumulant-order limit on paid amounts; a model
representing frailty explicitly is optimal within that ceiling without
reaching nominal coverage. Note that neither statement follows from the mechanism identified here, and
neither is claimed here as our own.

\subsection{Paper structure}

Section~\ref{sec:notation} establishes notation, including the distinction
between compositions and subcompositions on which the identification result
turns. Sections~\ref{sec:counts}--\ref{sec:amounts} develop the
Poisson--Gamma--Multinomial and Dirichlet--Gamma hierarchies and the duality
between them, and Section~\ref{sec:c_estimation} treats the estimation of the
concentration and the two diagnostics it supports.
Section~\ref{sec:tweedie} compares the construction with Tweedie models.
Section~\ref{sec:bootstrap} states the conditioning principle, the
decomposition of the residual bootstrap variance, the plug-in and
regularised algorithms, the coverage theorem and the conservative bias
theorem. Section~\ref{sec:simulations} reports the simulation studies:
correct specification, three misspecification designs, the frailty test
that rejects the omitted-frailty conjecture, the component decomposition
of the residual bootstrap, the oracle decomposition, conditional
calibration, and the anchoring regimes. Section~\ref{sec:empirical} applies
the construction to three benchmark triangles and to a Bornhuetter--Ferguson
illustration. Section~\ref{sec:discussion} concludes, and
Section~\ref{sec:limitations} states the limitations.

\section{Notation and data structure}\label{sec:notation}

Let $i = 1, \ldots, I$ index accident years and $j = 0, 1, \ldots, J-1$ index
development years. At valuation time~$I$, cells $(i, j)$ with $i + j \leq I$
are observed and cells with $i + j > I$ are future. Accident year $i$ is therefore observed through development year $I - i$, and
the number of accident years observed through at least development year $k$
is $I - k$; the concentration estimator of Section~\ref{sec:c_estimation}
uses the $I - k - 1$ of them observed beyond horizon $k$, a convention made
explicit in Definition~\ref{def:partial_props}. Row counts of this form
govern its precision.

For counts, $N_{ij}$ denotes the incremental claim count in cell $(i, j)$,
$N_i = \sum_j N_{ij}$ the ultimate claim count for accident year $i$, and
$\Nobs_i = \sum_{j=0}^{I-i} N_{ij}$ the observed count at valuation time. For
amounts, $X_{ij}$ denotes the incremental paid amount, $S_i = \sum_j X_{ij}$
the ultimate paid amount, and $\Xobs_i = \sum_{j=0}^{I-i} X_{ij}$ the observed
cumulative paid amount. We write $E_i$ for accident-year~$i$ exposure.

The two predictive quantities of interest are
\begin{equation}\label{eq:ibnr_ibnp_def}
  \Nibnr_i \;=\; N_i - \Nobs_i,
  \qquad
  \Sibnp_i \;=\; S_i - \Xobs_i,
\end{equation}
i.e.\ the count and amount of claims still outstanding for accident year $i$.
The first is the standard \emph{IBNR} (incurred but not reported) reserve. The
second we write \emph{IBNP} (incurred but not paid), by analogy with IBNR: it
is the amount of paid losses still outstanding for an accident year whose
ultimate has not yet been fully realised. We use IBNP throughout as the
amount-side counterpart of IBNR; the parallel between $\Nibnr_i$ and
$\Sibnp_i$ is the duality developed in Section~\ref{sec:amounts} and
Table~\ref{tab:duality}.

Any macro-level method producing ultimate estimates implicitly defines
cumulative development proportions $0 < F_0 < \cdots < F_{J-1} = 1$, with
incremental proportions $\pi_j = F_j - F_{j-1}$ forming a probability vector
$\bpi = (\pi_0, \ldots, \pi_{J-1})$ on the simplex:
\begin{equation}\label{eq:simplex}
  \sum_{j=0}^{J-1} \pi_j \;=\; 1,
  \qquad \pi_j > 0 \;\text{for all } j.
\end{equation}
The simplex constraint~\eqref{eq:simplex} is built into both the multinomial
allocation~\eqref{eq:multinomial} and the Dirichlet
allocation~\eqref{eq:dirichlet}. It is also the identifying convention shared
across the unified reserving framework: under the equivalent log-linear
parameterisation $\log \mu_{ij} = \alpha_i + \beta_j$ used in count-side
likelihood-based methods, the constraint $\sum_j \exp(\beta_j) = 1$ is the
same simplex condition expressed on the log scale, with
$\mu_i = \exp(\alpha_i)$ the expected ultimate count and
$\pi_j = \exp(\beta_j)$ the development weight. Throughout this paper we work
with $\bpi$ directly.

For Chain-Ladder, $F_j = 1 / (f_j \cdots f_{J-2})$ where $f_j$ are the link
ratios. For Bornhuetter--Ferguson, Cape Cod, or any other development method,
$\bpi$ is computed from that method's projected development pattern.

The vector of development shares of a single accident year,
$(X_{i,0}/S_i,\ldots,X_{i,J-1}/S_i)$, is a \emph{composition}: a vector of
non-negative proportions constrained to the unit simplex
\citep{aitchison1982}. Two operations on such a vector recur, and they should be kept strictly
apart.

\emph{Aggregation} sums a subset of components. The observed share
$\Wobs_i = \Xobs_i/S_i = \sum_{j \le I-i} X_{ij}/S_i$ is of this kind, and it
is the quantity the predictive of Section~\ref{sec:bootstrap} simulates.

\emph{Subcomposition} rescales a subset to its own total. For horizon $k$ and
accident year $i$ we write
\begin{equation}\label{eq:subcomp_def}
  W_{ij}^{(k)} \;=\; \frac{X_{ij}}{\sum_{l=0}^{k} X_{il}},
  \qquad j = 0,\ldots,k,
  \qquad\text{with}\qquad
  \pi_j^{(k)} \;=\; \frac{\pi_j}{F_k}
\end{equation}
its expectation. This is the quantity from which the concentration parameter
is estimated in Section~\ref{sec:c_estimation}, because it can be formed from
cells that are all observed, whereas $\Wobs_i$ has an unobserved denominator.
The two operations behave differently under the Dirichlet: aggregation
preserves the total concentration and subcomposition does not
(Lemmas~\ref{lem:dirichlet_agg} and~\ref{lem:dirichlet_sub}). Where the
distinction matters we write $a = cF_k$ for the subcomposition concentration.

\section{Multinomial framework for claim counts}\label{sec:counts}

The count-side hierarchy supplies the frequency--severity duality of
Section~\ref{sec:amounts}, it is the hierarchy in which the residual
bootstrap can be compared with the conditional predictive exactly
(Proposition~\ref{prop:boot_decomp} and Section~\ref{sec:sim_frailty}),
and it introduces the frailty dispersion $\kappa$, whose identifiability
from a single triangle is the count-side counterpart of the question this
paper asks about the concentration $c$.

\subsection{The Poisson--Gamma--Multinomial hierarchy}

For each accident year $i$, the count-side hierarchy is specified as
\begin{align}
\lambda_i &\sim \Gam(\kappa, \kappa/\mu_i), \label{eq:frailty}\\
N_i \mid \lambda_i &\sim \Pois(\lambda_i), \label{eq:poisson}\\
(N_{i,0},\ldots,N_{i,J-1}) \mid N_i
  &\sim \Mult(N_i, \bpi). \label{eq:multinomial}
\end{align}

Marginalising~\eqref{eq:frailty}--\eqref{eq:poisson} yields
$N_i \sim \NB(\kappa, \kappa/(\kappa + \mu_i))$. Throughout we use the
parameterisation in which $\NB(r, p)$ has mean $r(1-p)/p$ and variance
$r(1-p)/p^2$, so that the variance-to-mean ratio is $1/p$ and $p \to 1$ is the
Poisson limit.

The hierarchy separates two sources of uncertainty in the same way the
Dirichlet--Gamma construction of Section~\ref{sec:amounts} does. The frailty
$\lambda_i$ has mean $\mu_i$ and squared coefficient of variation $1/\kappa$,
so $\kappa$ governs how much accident years differ from one another beyond
Poisson noise; the multinomial allocation governs how a given year's claims
distribute across development periods. Writing $U_i = \lambda_i/\mu_i \sim
\Gam(\kappa,\kappa)$ makes the first component a unit-mean multiplier, which
is the form in which it appears as a data-generating process in
Section~\ref{sec:sim_frailty}: $\kappa \to \infty$ removes the frailty and
$\Var(U_i) = 1/\kappa$ measures it.

The following lemma records the law of the observed count $\Nobs_i$ given the
ultimate count $N_i$.

\begin{lemma}[Binomial thinning]\label{lem:thinning}
Under the multinomial allocation~\eqref{eq:multinomial},
$\Nobs_i \mid N_i \sim \mathrm{Binomial}(N_i, F_{I-i})$.
\end{lemma}

\begin{proof}
Conditional on $N_i$, each of the $N_i$ claims is independently allocated to a
development period with probabilities $\bpi$. The event that a claim is
reported by development year $I-i$ has probability
$\sum_{j \leq I-i} \pi_j = F_{I-i}$, and the allocations are independent
across claims, so the number so reported is binomial.
\end{proof}

Combining that lemma with Gamma--Poisson conjugacy gives the predictive law of
the outstanding count.

\begin{theorem}[Negative Binomial IBNR]\label{thm:ibnr}
Under the hierarchy~\eqref{eq:frailty}--\eqref{eq:multinomial},
\[
\Nibnr_i \mid \Nobs_i
\sim \NB\!\left(\Nobs_i + \kappa,\;
\frac{\kappa + \mu_i F_{I-i}}{\kappa + \mu_i}\right).
\]
As $\kappa \to 0$, i.e.\ maximal frailty,
$\Nibnr_i \mid \Nobs_i \sim \NB(\Nobs_i, F_{I-i})$, which depends on the data
only through $\Nobs_i$. As $\kappa \to \infty$, i.e.\ no frailty, the
predictive tends to $\Pois(\mu_i(1-F_{I-i}))$, independent of $\Nobs_i$.
\end{theorem}

\begin{proof}
Gamma--Poisson conjugacy with the thinned likelihood of
Lemma~\ref{lem:thinning} gives
$\lambda_i \mid \Nobs_i \sim \Gam(\kappa + \Nobs_i,\;
\kappa/\mu_i + F_{I-i})$, the second argument a rate; marginalising
$\Nibnr_i \mid \lambda_i \sim \Pois(\lambda_i(1-F_{I-i}))$ over this
posterior yields the stated negative binomial, with probability
$(\kappa/\mu_i + F_{I-i})/(\kappa/\mu_i + 1)
= (\kappa + \mu_i F_{I-i})/(\kappa + \mu_i)$. The limits follow by
inspection: as $\kappa \to 0$ the shape tends to $\Nobs_i$ and the
probability to $F_{I-i}$; as $\kappa \to \infty$ the probability tends to
one, the variance-to-mean ratio to one, and the mean to
$\mu_i(1-F_{I-i})$.
\end{proof}

The predictive mean of Theorem~\ref{thm:ibnr} rearranges into a credibility
form, $(1-F_{I-i})[Z_i \Nobs_i/F_{I-i} + (1-Z_i)\mu_i]$ with weight
$Z_i = \mu_i F_{I-i}/(\kappa + \mu_i F_{I-i})$, carrying the Chain-Ladder
estimate at $\kappa \to 0$ and the Bornhuetter--Ferguson form
$\mu_i(1-F_{I-i})$ at $\kappa \to \infty$. This is the count-side instance
of the credibility family developed in \citet{VanOirbeek2026cred} and is not
re-derived here.

\begin{remark}[Which $\kappa$ is identified]\label{rem:kappa_identification}
The parameter $\kappa$ in \eqref{eq:frailty} is a \emph{row-level} dispersion:
it measures how accident-year expectations vary about a common structure. It
is not the same quantity as the cell-level dispersion that a Poisson or
negative binomial regression with free accident-year effects reports, and the
two are not identified from the same fit. Where the mean structure carries a
free parameter per accident year, those parameters absorb the realised
$\lambda_i$ of each observed row, so the estimated dispersion measures
within-row variability and carries no information about $\kappa$. Absorption has a second consequence that is easily misread. Because the
free levels are fitted to each row, the predictive that a free-level
residual bootstrap generates is the $\kappa \to 0$ member of
Theorem~\ref{thm:ibnr}, whose variance is the largest in the family;
Proposition~\ref{prop:boot_decomp} shows that the bootstrap's resampling
of the diagonal reproduces exactly the posterior level variance of that
member. Absorbing the realised frailty therefore does not lose the
frailty uncertainty; it prices it at its diffuse-prior value, which is
never too small (Section~\ref{sec:sim_frailty}). It also means that a
value of $\hat\kappa$ reported from a single triangle under free accident-year levels, such as the $\hat\kappa = 4.8$ of
\citet{VanOirbeek2026nbcl}, should not be read as an estimate of the
row-level frailty dispersion of this hierarchy. The two share a symbol and
not an identification channel. In the notation of the companion papers the
hierarchy's $\kappa$ is the row dispersion $\kappa_r$ and the free-level
estimate a cell-level composite $\kappa_c$: under free levels $\kappa_r$ is
\emph{exactly} non-identified, by the collinearity of
Section~\ref{sec:related}, while what anchored levels recover is
bounded by the information ceiling of \citet{VanOirbeek2026cnbg}. For $\kappa_r$ the answer to the question this paper asks of the
concentration $c$ (Section~\ref{sec:c_estimation}) is therefore immediate:
without anchoring, nothing is identified at all.
\end{remark}


\section{Multinomial framework for claim sizes}\label{sec:amounts}

The amount-side hierarchy is the counterpart of the count-side hierarchy of
Section~\ref{sec:counts} through the duality of Table~\ref{tab:duality}, it
supplies the conditional predictive of the outstanding amount from which the
bootstrap of Section~\ref{sec:bootstrap} draws, and it introduces the
concentration $c$, whose estimation from a single triangle is treated in
Section~\ref{sec:c_estimation}.

\subsection{The Dirichlet--Gamma construction}

For each accident year $i$, the amount-side hierarchy is specified as
\begin{align}
S_i &\sim \Gam(\alpha E_i, \beta), \label{eq:gamma_ult}\\
\left(\frac{X_{i,0}}{S_i},\ldots,\frac{X_{i,J-1}}{S_i}\right)
  &\sim \Dir(c\pi_0,\ldots,c\pi_{J-1}). \label{eq:dirichlet}
\end{align}
As $c \to \infty$ the Dirichlet concentrates on $\bpi$ and
$X_{ij} \to S_i\pi_j$ deterministically.

The development allocation $(X_{i,\cdot}/S_i)$ is a composition: a vector of
non-negative proportions constrained to the unit simplex, in the sense of
\citet{aitchison1982}. The Dirichlet is neither the only nor the most flexible
distribution on that simplex. It imposes a neutrality structure, i.e.\ the proportions are completely
subcompositionally independent, which excludes the richer dependence admitted
by logistic-normal alternatives. It is
therefore a genuine modelling restriction rather than a default choice.

We adopt it deliberately. Note that the same neutrality that limits its
dependence is what makes the conditioning principle implementable: it yields the closed-form predictive distribution of
$(1-W_i)/W_i$ in Section~\ref{sec:bootstrap}, and with it the property that
the predictive attaches to any method producing a development pattern. The
cost is that calibrated prediction then depends on the data through a single
scalar $c$, and Section~\ref{sec:c_estimation} shows that a single paid
triangle determines that scalar poorly. The restriction buys tractability and generality, but not identifiability.

Table~\ref{tab:duality} makes the parallel structure explicit. Every
structural element has a counterpart: Poisson--Multinomial for counts mirrors
Gamma--Dirichlet for amounts.

\begin{table}[t]
\centering
\caption{The frequency--severity duality. The year-level frailty dispersion
$\kappa$ is identifiable only when accident-year means are anchored, for
instance to exposure via a common rate as in Cape Cod, since under free
accident-year parameters it is absorbed and not identifiable from a single
triangle (Remark~\ref{rem:kappa_identification}). The concentration $c$ carries no
such restriction, but is weakly determined for a different reason
(Section~\ref{sec:c_estimation}). Predictive rows are stated in the respective
diffuse limits: $\kappa \to 0$ for counts (Theorem~\ref{thm:ibnr}) and
$\alpha,\beta \to 0$ for amounts (Remark~\ref{rem:diffuse_scope}).}
\label{tab:duality}
\begin{tabular}{@{}lll@{}}
\toprule
\textbf{Component} & \textbf{Counts} & \textbf{Amounts} \\
\midrule
Ultimate distribution
  & $N_i \sim \NB(\kappa, p_i)$
  & $S_i \sim \Gam(\alpha E_i, \beta)$ \\
Allocation mechanism
  & $(N_{i,\cdot}) \sim \Mult(N_i, \bpi)$
  & $(X_{i,\cdot}/S_i) \sim \Dir(c\bpi)$ \\
Predictive distribution
  & $\Nibnr_i \sim \NB(\Nobs_i, F_{I-i})$
  & $\Sibnp_i = \Xobs_i \cdot R_i$,
    $R_i \sim \Bet'$ \\
CL point estimate
  & $\Nobs_i \cdot (1-F)/F$
  & $\Xobs_i \cdot (1-F)/F$ \\
Variance driver
  & $\kappa$ (frailty dispersion)
  & $c$ (allocation dispersion) \\
\bottomrule
\end{tabular}
\end{table}

\subsection{IBNP predictive distribution}

The following theorem gives the conditional law of the outstanding amount
$\Sibnp_i$ given the observed cumulative payment $\Xobs_i$, together with its
first two moments.

\begin{theorem}[IBNP predictive distribution]\label{thm:ibnp}
Under the Dirichlet--Gamma model~\eqref{eq:gamma_ult}--\eqref{eq:dirichlet},
with (ii)--(iv) stated in the diffuse-prior limit $\alpha,\beta \to 0$ (see
Remark~\ref{rem:diffuse_scope}):
\begin{enumerate}[label=(\roman*)]
\item $\Wobs_i = \Xobs_i/S_i \sim \Bet(cF_{I-i},\, c(1-F_{I-i}))$.
\item Conditional on $\Xobs_i$:
      $\Sibnp_i = \Xobs_i \cdot (1-\Wobs_i)/\Wobs_i$, where
      $\Wobs_i \sim \Bet(cF_{I-i},\, c(1-F_{I-i}))$.
\item $\E[\Sibnp_i \mid \Xobs_i]
      = \Xobs_i \cdot \dfrac{c(1-F_{I-i})}{cF_{I-i}-1}$
      (requires $cF_{I-i} > 1$). The ratio to the Chain-Ladder estimate
      $\Xobs_i(1-F_{I-i})/F_{I-i}$ is $cF_{I-i}/(cF_{I-i}-1) > 1$; the
      Chain-Ladder point estimate is recovered as $c \to \infty$.
\item $\CV^2(\Sibnp_i \mid \Xobs_i)
      = \dfrac{c-1}{c(1-F_{I-i})(cF_{I-i}-2)}
      \approx \dfrac{1}{cF_{I-i}(1-F_{I-i})}$
      (requires $cF_{I-i} > 2$).
\end{enumerate}
\end{theorem}

\begin{proof}
(i) Dirichlet aggregation (Lemma~\ref{lem:dirichlet_agg}), which preserves the
total concentration $c$.
(ii) Since $\Xobs_i = S_i\Wobs_i$ and $\Sibnp_i = S_i(1-\Wobs_i)$, dividing
gives the ratio; in the diffuse limit the conditional law of $\Wobs_i$ given
$\Xobs_i$ is its marginal Beta (Remark~\ref{rem:diffuse_scope}).
(iii) For $W \sim \Bet(a,b)$ with $a > 1$, $\E[1/W] = (a+b-1)/(a-1)$; with
$a = cF$ and $b = c(1-F)$ this gives $\E[(1-W)/W] = c(1-F)/(cF-1)$.
(iv) $(1-W)/W \sim \beta'(c(1-F), cF)$ with mean $c(1-F)/(cF-1)$ and variance
$c(1-F)(c-1)/[(cF-1)^2(cF-2)]$; dividing gives the stated $\CV^2$. The
existence conditions in (iii) and (iv) are the cases $k = 1$ and $k = 2$ of
the moment condition $cF_{I-i} > k$, developed in
Remark~\ref{rem:validity}.
\end{proof}

\begin{remark}[Scope of the conditional statements]\label{rem:diffuse_scope}
By Lemma~\ref{lem:gamma_dirichlet}, $S_i$ and the allocation vector are
independent, so the marginal statement (i) holds under any prior on $S_i$.
Conditioning on $\Xobs_i = S_i \Wobs_i$, however, couples the two: under
$S_i \sim \Gam(\alpha E_i, \beta)$ the conditional density of $\Wobs_i$ given
$\Xobs_i = x$ is
\[
p(w \mid x) \;\propto\;
w^{\,cF_{I-i} - \alpha E_i - 1}\,
(1-w)^{\,c(1-F_{I-i}) - 1}\,
e^{-\beta x / w},
\qquad w \in (0,1),
\]
which is a Beta density if and only if the prior is diffuse
($\alpha, \beta \to 0$), in which case it reduces to the marginal
$\Bet(cF_{I-i}, c(1-F_{I-i}))$. Statements (ii)--(iv) and
Algorithm~\ref{alg:bootstrap} are therefore statements about the
diffuse-anchored, Chain-Ladder-type predictive. Under an informative anchor
the Beta draw is a plug-in approximation to the tilted density above; the
Bornhuetter--Ferguson variant of Section~\ref{sec:empirical} is of this kind.

The approximation is not innocuous. Section~\ref{sec:conditional_coverage}
measures its cost at roughly ten coverage points across terciles of realised
allocation dispersion, at the ultimate dispersion typical of this study, and
shows that cost vanishing as $\alpha, \beta \to 0$: the gradient falls to
$1.4$ points, inside Monte-Carlo error, when the prior on $S_i$ becomes
diffuse. This holds \emph{even when $c$ and $\bpi$ are known exactly}, so it
is not estimation error and no amount of integration removes it. It is the price of the closed form, and
Section~\ref{sec:limitations} returns to the resulting tension.
\end{remark}

Under the Dirichlet--Gamma framework, zeros arise through the count channel:
$X_{ij} = 0$ if and only if $N_{ij} = 0$, with probability
$(\kappa/(\kappa + \mu_{ij}))^\kappa$, a structural consequence of $\kappa$
and $\mu_{ij}$ rather than a free parameter. See \citet{VanOirbeek2026cnbg}
for the structural zero test and the zero-inflated extension. Note that zero
cells are not merely a modelling nuisance: the estimator of
Section~\ref{sec:c_estimation} forms ratios with partial row sums in the
denominator, so a row containing a zero in an early column contributes no
usable subcomposition at that horizon.

\begin{remark}[Validity of the posterior moments and the boundary regime]
\label{rem:validity}
Theorem~\ref{thm:ibnp}(iii) requires $cF_{I-i} > 1$ for the posterior mean of
$(1-W_i)/W_i$ to exist. The condition is a structural feature of the
Beta-prime distribution: $\E[(1-W)/W] = (1-F)/(F - 1/c)$, whose denominator
vanishes as $cF_{I-i} \to 1$ from above. More generally, for
$(1-W_i)/W_i \sim \beta'\bigl(c(1-F_{I-i}),\, cF_{I-i}\bigr)$ the $k$-th
moment is finite if and only if $cF_{I-i} > k$. Every boundary in the rule
below is an instance of that condition; none is a convention.

\begin{center}
\begin{tabular}{@{}lll@{}}
\toprule
$\hat{c}F_{I-i}$ & What exists & What to report \\
\midrule
$\leq 1$   & no moments        & median and the 5th, 25th, 50th,
                                 75th, 95th percentiles \\
$(1, 2]$   & mean only         & quantiles; the mean is unstable \\
$(2, 4]$   & mean and variance & quantiles; a standard error has no
                                 finite-sample guarantee \\
$> 4$      & all four moments  & full moment reporting \\
\bottomrule
\end{tabular}
\end{center}

We apply $\hat{c}F_{I-i} \geq 5$ throughout, i.e.\ the last boundary with a
margin. The margin is warranted because $\hat{c}$ is itself upward-biased at
realistic $I$ (Remark~\ref{rem:sigma_c_finite}), so the realised
$cF_{I-i}$ typically sits below the value to which the rule is applied. The
third tier is the one most easily overlooked: the variance exists there, so
software will return a standard error, but the fourth moment does not, and the
sample variance of the bootstrap draws therefore has infinite variance.

Note what does \emph{not} fail. Algorithm~\ref{alg:bootstrap} remains
operational at every tier: each draw
$W_i^{*(b)} \sim \Bet(\hat{c}F_{I-i}, \hat{c}(1-F_{I-i}))$ is well defined for
any $\hat{c}, F_{I-i} > 0$ and the resulting $S_i^{\mathrm{IBNP},(b)}$ is
finite almost surely. What fails is the convergence of the sample moments, not
the sampler. Quantiles remain defined throughout. Note that being defined is not the same
as being useful, however: Section~\ref{sec:sim_tweedie} reports a cell at
$\hat{c}F_{I-i} \approx 2$ where nominal coverage is delivered by an interval
$357$ times the true reserve.

The degradation across the tiers is gradual rather than abrupt. In $200$
independent bootstrap runs of $B = 5{,}000$ at $F_{I-i} = 0.45$, the
coefficient of variation of the reported standard error is $0.018$ at
$\hat{c}F_{I-i} = 12$, $0.037$ at $5$, $0.060$ at $4$, $0.230$ at $2.5$ and
$1.54$ at $1.5$. The tiers identify where a reported moment ceases to be
defined; the accompanying instability is a continuous function of
$\hat{c}F_{I-i}$ and not a threshold effect at any practical $B$. Whether a
given portfolio will encounter these tiers is checkable in advance from the
operability bound $c^\dagger(\tau) = \tau/F_{I-i}$ of
Section~\ref{sec:c_threshold}, without computing a reserve.
\end{remark}

\subsection{Estimation of the concentration parameter}
\label{sec:c_estimation}

The concentration $c$ is estimated from the dispersion of the partial-column
proportions across accident years, which are defined first.

\begin{definition}[Partial-column proportions]\label{def:partial_props}
For horizon $k \in \{1,\ldots,J-2\}$ and accident year $i$ with $k < I - i$,
\[
W_{ij}^{(k)}
= \frac{X_{ij}}{\sum_{l=0}^{k} X_{il}},
\qquad j = 0,\ldots,k-1,
\]
the final component $j = k$ being determined by the others and therefore
omitted. By Lemma~\ref{lem:dirichlet_sub},
\[
W_{ij}^{(k)} \sim \Bet\bigl(a\,\pi_j^{(k)},\; a(1-\pi_j^{(k)})\bigr),
\qquad \pi_j^{(k)} = \pi_j/F_k, \qquad a = cF_k,
\]
so that $\E[W_{ij}^{(k)}] = \pi_j^{(k)}$ and
\begin{equation}\label{eq:partial_var}
\Var\bigl(W_{ij}^{(k)}\bigr)
= \frac{\pi_j^{(k)}(1-\pi_j^{(k)})}{a + 1}
= \frac{\pi_j^{(k)}(1-\pi_j^{(k)})}{cF_k + 1}.
\end{equation}
The concentration of the partial-column subcomposition is $cF_k$ and not $c$
(Remark~\ref{rem:agg_vs_sub}): rescaling by a partial row total discards the
development mass beyond horizon $k$, and the residual composition is
correspondingly less concentrated. The factor $F_k$ must therefore be undone
when the cell estimators are combined, as in
Proposition~\ref{prop:c_estimator}. The condition $k < I - i$ restricts to
rows observed \emph{beyond} horizon $k$: the diagonal row $i = I - k$,
observed exactly through $k$, also carries a fully observed horizon-$k$
subcomposition and is excluded by this convention. Including it would replace
$n_k = I - k - 1$ by $I - k$, one additional row per horizon, most material
at the small $I$ where the estimator is weakest. Every number in this paper
uses the stated convention.
\end{definition}

\begin{proposition}[Moment estimator of $c$]\label{prop:c_estimator}
For each $(j,k)$ with $n_k = I-k-1 \geq 3$,
\begin{equation}\label{eq:c_moment}
\hat{c}_{jk}
= \frac{1}{\hat{F}_k}
  \left[\frac{\hat{\pi}_j^{(k)}(1-\hat{\pi}_j^{(k)})}
             {\widehat{\Var}(W_{\cdot j}^{(k)})} - 1\right],
\qquad \hat{F}_k = \sum_{l=0}^{k}\hat\pi_l ,
\end{equation}
and the overall estimator is $\hat{c} = \mathrm{median}\{\hat{c}_{jk}\}$ over
all valid pairs.
\end{proposition}

\begin{proof}
By~\eqref{eq:partial_var} the bracketed quantity in~\eqref{eq:c_moment} is the
moment estimator of the subcomposition concentration $a = cF_k$; dividing by
$\hat{F}_k$ returns an estimator of $c$. The median aggregation reduces
sensitivity to noisy variance estimates at individual horizons, and after the
$\hat{F}_k$ correction every cell estimator has the same probability limit
$c$, so the median combines homogeneous quantities. Without the correction the
cell estimators have horizon-dependent limits $cF_k$ and the aggregate
converges to $cF_{k^*}$ at the median contributing horizon $k^*$: the
estimator is then inconsistent for $c$.
\end{proof}

The correction is not cosmetic, and its size depends on the development
pattern. With $\bpi = (0.45, 0.25, 0.15, 0.10, 0.05)$ the contributing
horizons carry $F_k \in \{0.70, 0.85, 0.95\}$, so an individual uncorrected
cell estimator understates $c$ by between $5\%$ and $30\%$. The aggregate is
governed by the median contributing horizon: at $I = 10$ the six contributing
cells give $F_{k^*} = 0.90$ and hence a $10\%$ understatement. On a $J = 10$
triangle with the geometric pattern of Section~\ref{sec:sim_design} the
understatement is smaller, about $4\%$, because the later horizons carry
$F_k$ close to one and contribute most of the cells. In every case it is a
systematic bias rather than a small-sample effect: it does not diminish as $I$
grows, and simulation under the uncorrected estimator returns
$\E[\hat c]/c = 0.916$ at $I = 100$, where the finite-sample effects of
Remark~\ref{rem:sigma_c_finite} have all but vanished.

\begin{proposition}[Per-cell asymptotic variance of the moment estimator]
\label{prop:c_variance}
Under the Dirichlet--Gamma model with fixed $J$ and true $c > 0$, each cell
estimator $\hat{c}_{jk}$ from Proposition~\ref{prop:c_estimator} satisfies
\begin{equation}\label{eq:c_clt}
\sqrt{n_k}\,(\hat{c}_{jk} - c)
\xrightarrow{d} \mathcal{N}(0,\, \sigma_{c,k}^2)
\quad \text{as } I \to \infty, \quad n_k = I-k-1,
\end{equation}
where, writing $a \equiv cF_k$ for the subcomposition concentration at horizon
$k$ and $\pi \equiv \pi_j^{(k)}$,
\begin{equation}\label{eq:sigma_c_explicit}
\boxed{
\sigma_{c,k}^2
= \frac{c\,(a+1)\bigl[2a(a-3)\,\pi(1-\pi) + 3a + 1\bigr]}
       {F_k\,\pi(1-\pi)(a+2)(a+3)},
\qquad a = cF_k .
}
\end{equation}
At $F_k = 1$ this reduces to
$c(c+1)[2c(c-3)\pi(1-\pi)+3c+1]/[\pi(1-\pi)(c+2)(c+3)]$. The expression is
strictly positive for all $c > 0$, $F_k \in (0,1]$, $\pi \in (0,1)$, and for
large $c$ satisfies $\sigma_{c,k}^2 \sim 2c^2$ independently of $F_k$: the
horizon enters the constant but not the rate, and the rate is scale-free in
the concentration.

The aggregated estimator $\hat{c} = \mathrm{median}\{\hat{c}_{jk}\}$ combines
correlated cell estimators on the same triangle and inherits the
$\sqrt{I}$-rate, so that $\mathrm{MSE}(\hat{c}) = O(I^{-1})$. Writing
$\bar\sigma_c^2$ for the average of \eqref{eq:sigma_c_explicit} over the $M$
contributing cells, its variance is bounded above by $\bar\sigma_c^2/I$
(perfectly correlated cells) and below by $\bar\sigma_c^2/(I \cdot M)$
(perfectly independent cells) \emph{as $I \to \infty$}. That qualifier is
essential and not a formality: at the triangle sizes at which macro-level
reserving operates the upper bound is exceeded, at $I = 7$ by more than an
order of magnitude (Remark~\ref{rem:sigma_c_finite}).
\end{proposition}

\begin{proof}
Write $\pi \equiv \pi_j^{(k)}$, $q \equiv \pi(1-\pi)$, $a \equiv cF_k$ and
$\sigma^2 \equiv q/(a+1)$, the subcomposition variance
\eqref{eq:partial_var}. Consider first the uncorrected cell statistic
$\hat{g}_{jk} = f(\hat{m}, \hat{s}^2)$ with $f(m, s^2) = m(1-m)/s^2 - 1$,
which by Lemma~\ref{lem:dirichlet_sub} is the moment estimator of $a$. The
gradient of $f$ at the true values $(m, s^2) = (\pi, \sigma^2)$ is
\[
\nabla f
= \left(\frac{(1-2\pi)(a+1)}{q},\;
  -\frac{(a+1)^2}{q}\right),
\]
and the asymptotic covariance matrix of $(\hat{m}, \hat{s}^2)$ scaled by $n_k$
is $\Sigma$ with $\Sigma_{11} = \mu_2(W) = q/(a+1)$,
$\Sigma_{12} = \mu_3(W)$ and $\Sigma_{22} = \mu_4(W) - \mu_2(W)^2$, given by
\eqref{eq:beta_mu2}--\eqref{eq:beta_mu4} with concentration $a$. The delta
method $\tau_g^2 = \nabla f^\top \Sigma\, \nabla f$ yields
\begin{align*}
\tau_g^2
&= \frac{(1-2\pi)^2(a+1)}{q}
\;-\; \frac{4(1-2\pi)^2(a+1)^2}{q(a+2)} \\
&\quad + \frac{2(a+1)^2\bigl[(a^2-10a-12)\,q + 3(a+1)\bigr]}
       {q(a+2)(a+3)}.
\end{align*}
Extracting the common factor $(a+1)/\bigl(q(a+2)(a+3)\bigr)$ and substituting
$(1-2\pi)^2 = 1 - 4q$, the bracketed numerator is
\[
(1-4q)(a+2)(a+3) - 4(1-4q)(a+1)(a+3)
+ 2(a+1)\bigl[(a^2-10a-12)q + 3(a+1)\bigr].
\]
The two $(1-4q)$ terms combine to $-(1-4q)(a+3)(3a+2)$; the constant part
collects to $6(a+1)^2 - (a+3)(3a+2) = a(3a+1)$, and the coefficient of $q$ to
$4(a+3)(3a+2) + 2(a+1)(a^2-10a-12) = 2a^2(a-3)$. Hence
\[
\tau_g^2
= \frac{a(a+1)\bigl[2a(a-3)\,q + 3a + 1\bigr]}{q(a+2)(a+3)},
\]
so $\sqrt{n_k}(\hat{g}_{jk} - a) \xrightarrow{d} \mathcal{N}(0, \tau_g^2)$.
The reported cell estimator is $\hat{c}_{jk} = \hat{g}_{jk}/\hat{F}_k$;
treating $F_k$ as known, Slutsky's theorem gives
$\sqrt{n_k}(\hat{c}_{jk} - c) \xrightarrow{d}
\mathcal{N}(0, \tau_g^2/F_k^2)$, and substituting $a = cF_k$, so that
$a/F_k^2 = c/F_k$, yields~\eqref{eq:sigma_c_explicit}.

Positivity is checked as follows. For $a \geq 3$ both $2a(a-3)q$ and $3a+1$
are positive.
For $a < 3$ we have $2a(3-a)q \leq a(3-a)/2$ since $q \leq 1/4$, while
$3a+1 > a(3-a)/2$ for all $a \in [0,3)$. Hence $\sigma_{c,k}^2 > 0$
throughout.

Turning to the asymptotics, since $n_k = I-k-1 = O(I)$, each cell estimator
has
variance $\sigma_{c,k}^2/n_k = O(I^{-1})$. The median of correlated cell
estimators inherits the $\sqrt{I}$ rate: by standard results for the sample
median \citep{vdVaart1998}, $\sqrt{I}(\hat{c}-c) = O_p(1)$, and consequently
$\mathrm{MSE}(\hat{c}) = O(I^{-1})$. The exact variance constant of $\hat{c}$
depends on the inter-cell correlation structure induced by shared rows on the
triangle and is reported empirically in Remark~\ref{rem:sigma_c_finite};
deriving its closed form requires fourth-order cross-cumulants of the
Dirichlet allocation and is not pursued here.
\end{proof}

Expression~\eqref{eq:sigma_c_explicit} treats $F_k$ as known. In practice
$\hat{F}_k$ is estimated from the same triangle at rate $O_p(I^{-1/2})$, so
plugging it in contributes an additional term of the same order. The rate
$\sqrt{n_k}$ is unaffected, and the additional constant is absorbed
empirically in the reduction factor of Remark~\ref{rem:sigma_c_finite},
which is measured with $\hat{F}_k$ in place. A closed form for the combined
constant requires the covariance between $\hat{F}_k$ and the cell variance
estimates, and is not pursued here.

\begin{remark}[Finite-sample behaviour of $\hat{c}$]\label{rem:sigma_c_finite}
Proposition~\ref{prop:c_variance} makes three separable claims, i.e.\ a rate,
a scaling in $c$ and a constant, and simulation under the model ($J = 5$,
$\bpi = (0.45, 0.25, 0.15, 0.10, 0.05)$, $M = 10{,}000$ per cell, seed 2026)
shows that the three do not all hold. The implemented estimator is first
verified to use exactly the cells of Definition~\ref{def:partial_props}: the
six contributing cells carry average development mass $\bar F = 0.875$ and
average subcomposition proportion $\bar\pi = 0.3935$, and the simulation
reproduces both to four decimals. Table~\ref{tab:sigma_c_finite} then
compares the measured scaled variance $I \cdot \widehat{\Var}(\hat c)$ with
$\bar\sigma_c^2$, the average of \eqref{eq:sigma_c_explicit} over those
cells, and reports the bias ratio $\E[\hat c]/c$ on the same grid.

\begin{table}[ht]\centering
\caption{Finite-sample behaviour of $\hat{c}$ ($J = 5$, $M = 10{,}000$ per
cell). Top panel: ratio of measured scaled variance
$I\cdot\widehat{\Var}(\hat c)$ to the closed form $\bar\sigma_c^2$; values
above one mean the asymptotic expression \emph{understates} the variance.
Bottom panel: bias ratio $\E[\hat c]/c$, averaged over the three values of
$c$; it is scale-free in $c$ (variation across $c$ at most $3.0$ points at
$I = 7$ and at most $1.0$ point at $I \geq 15$).}
\label{tab:sigma_c_finite}
\begin{tabular}{@{}lcccccc@{}}
\toprule
$I$ & 7 & 10 & 15 & 20 & 50 & 100 \\
\midrule
\multicolumn{7}{@{}l}{\emph{Variance ratio
  $I\cdot\widehat{\Var}(\hat c)/\bar\sigma_c^2$}}\\
$c = 20$  & 15.75 & 2.85 & 1.55 & 1.12 & 0.79 & 0.72 \\
$c = 50$  & 14.00 & 2.77 & 1.47 & 1.14 & 0.79 & 0.73 \\
$c = 100$ & 14.75 & 2.96 & 1.55 & 1.19 & 0.81 & 0.73 \\
\midrule
\multicolumn{7}{@{}l}{\emph{Bias ratio $\E[\hat{c}]/c$}}\\
          & 1.89 & 1.31 & 1.15 & 1.10 & 1.03 & 1.015 \\
\bottomrule
\end{tabular}
\end{table}

The rate holds: the variance ratio converges, from above, and sits
inside the asymptotic interval $[1/M, 1]$ by $I = 50$; at $I = 100$ the
median over the six correlated cell estimators reduces the variance to
$0.73$ of the average per-cell formula, stable to within $0.01$ across an
order of magnitude in $c$. Note that this factor is a property of the asymptotic
regime, and should not be carried to smaller triangles.

The scaling in $c$ holds as well: at every $I$ the three variance rows
agree to within $\pm 6\%$, and the bias ratio is
likewise scale-free. This is the finite-sample counterpart of
$\sigma_{c,k}^2 \sim 2c^2$: both the relative bias and the coefficient of
variation of $\hat{c}$ depend on $I$ and the development pattern but not on
the concentration. It accounts for the near-flatness of coverage in $c$ in
Table~\ref{tab:sensitivity}, and it is why the diagnostic of
Section~\ref{sec:c_threshold} can be stated against a portfolio-specific
bound rather than recalibrated at each concentration.

The constant, however, does not hold at practical $I$. At $I = 10$ the closed
form understates $\Var(\hat c)$ by a factor of $2.9$ and at $I = 7$ by
$15.8$, while the asymptotic upper bound $\bar\sigma_c^2/I$ is exceeded for
every $I \leq 20$. At $c = 50$ that bound implies $\sigma(\hat c) \approx 22$
at $I = 10$ against a simulated $\mathrm{sd}(\hat{c}) = 45.9$
(Table~\ref{tab:sim_correct}). Expression~\eqref{eq:sigma_c_explicit} should
therefore be used to reason about how the precision of $\hat{c}$ scales with
$I$ and $c$, and not to compute a standard error on a triangle of the size actually encountered,
where it is optimistic by a factor of two to sixteen; an interval for $c$
should come from the posterior of Section~\ref{sec:regularised}, defined at
every $I$. More consequentially, the estimator is biased upward, by a factor
depending only on $I$ (bottom panel), and the mechanism is the same in both
failures: the delta-method expansion linearises
$\hat\pi(1-\hat\pi)/\hat{s}^2$ around the true moments, and at $n_k \in \{3,4,5\}$ the sampling distribution of $\hat{s}^2$ is nowhere
near that neighbourhood. The Jensen effect
$\E[1/\hat{s}^2] > 1/\E[\hat{s}^2]$ leaves a residual $1.5\%$ bias even at
$I = 100$ ($O(n_k^{-1})$, but not vanished at the horizon sample sizes
available), and at practical $I$ it is the dominant component of the
estimation error. An overstated $\hat{c}$ concentrates the Beta predictive
and narrows the interval, so bias and variance act in the same direction on
coverage; Section~\ref{sec:oracle} measures their combined effect. The estimator is consistent and its asymptotic
variance is correct, but neither property is of practical use at
$n_k = 4$.
\end{remark}

\subsection{The $\hat{c}$ diagnostic}\label{sec:c_threshold}

Two distinct questions are commonly asked of $\hat{c}$, and they have
different answers. The first is whether the framework can report the summaries
the practitioner needs on the accident years they care about; this is a
distributional existence condition, derived below. The second is whether the
portfolio exhibits enough allocation heterogeneity that a model with explicit
accident-year frailty is preferable; this is a model-selection question, and
the answer is a screen with measurable error rates rather than a bound. We
treat them separately.

The first question is one of operability. By Remark~\ref{rem:validity} the
$k$-th moment of $\Sibnp_i$ exists if and
only if $cF_{I-i} > k$, so the reporting tiers are conditions on the product
$cF_{I-i}$ and each inverts into a requirement on $c$ that depends on how
developed the year is. Let $i^*$ denote the least-developed accident year the
practitioner requires summarised. Then tier $\tau$ is attainable for every
required year if and only if
\begin{equation}\label{eq:c_operability}
  c \;\geq\; c^\dagger(\tau) \;:=\; \frac{\tau}{F_{I-i^*}} ,
\end{equation}
with $\tau = 1$ for the mean, $\tau = 2$ for the variance and $\tau = 4$ for a
stable standard error; we use $\tau = 5$ throughout as that last boundary with
a margin. When all accident years are required, the binding year is the most
recent one, $i^* = I$, for which $F_{I-i^*} = \pi_0$, and the bound is $c^\dagger(\tau) = \tau/\pi_0$, which is computable from the
estimated development pattern before any reserve is calculated.

The bound~\eqref{eq:c_operability} is portfolio-specific through $\pi_0$, and
on the three benchmarks of Section~\ref{sec:empirical} it varies by more than
an order of magnitude:
\begin{center}
\begin{tabular}{@{}lccccc@{}}
\toprule
Portfolio & $\pi_0$ & $c^\dagger(2)$ & $c^\dagger(5)$ & $\hat{c}$
& Clears $c^\dagger(5)$? \\
\midrule
Taylor--Ashe & $0.0692$ & $28.9$  & $72.2$  & $189.1$ & yes \\
RAA          & $0.1121$ & $17.8$  & $44.6$  & $24.3$  & no \\
Mortgage     & $0.0079$ & $253.0$ & $632.5$ & $152.0$ & no \\
\bottomrule
\end{tabular}
\end{center}
A single number such as $\hat{c} < 30$ implicitly fixes $\pi_0 \approx 1/6$:
a specific and undisclosed assumption about how long-tailed the portfolio is.
We therefore report $c^\dagger$ alongside $\hat{c}$ rather than comparing
$\hat{c}$ with a constant, and the ordering the two produce is not the same.
Against a common threshold RAA looks marginal and Mortgage comfortable;
against their own bounds Mortgage is by far the more demanding portfolio,
requiring $\hat{c} \geq 633$ where it has $152$, because it reports under one
percent of the ultimate in its first development year. That is why both
exclude an accident year in Table~\ref{tab:empirical_all}, and why the
consequence differs between them: the excluded year carries $31\%$ of the RAA
reserve and $11\%$ of the Mortgage reserve. The bound tells the practitioner
\emph{that} a year will be unreportable; its reserve share tells them whether
it matters.

The simulations of Section~\ref{sec:simulations} use
$\bpi = (0.45, 0.25, 0.15, 0.10, 0.05)$, so $\pi_0 = 0.45$ and
$c^\dagger(5) \approx 11$: the operability bound is essentially never binding
on the simulated grid. Real paid triangles report far less of the ultimate in
the first development year, and the empirical $\pi_0$ of the three benchmarks
is between one quarter and one fiftieth of the simulated value. The
operability bound is therefore a real constraint on application and not one
the simulation study exercises. Note that this is a limitation of the design and not evidence that the
constraint is slack: a study whose
$\pi_0$ matched the benchmarks would encounter the boundary tiers of
Remark~\ref{rem:validity} routinely, and its coverage figures would be
correspondingly harder to interpret.

Separately from operability, a small $\hat{c}$ indicates that development
proportions vary substantially across accident years, which is the regime in
which a model with explicit accident-year frailty \citep{VanOirbeek2026cnbg}
is preferable to the present one. Because $\hat{c}$ is itself estimated from a
single triangle, a rule of the form ``flag the portfolio when
$\hat{c} < t$'' has error rates depending on both the true $c$ and the
triangle size $I$, and these must be measured rather than assumed.
Table~\ref{tab:c_screen} reports the detection rate
$\Pr(\hat{c} < t \mid c, I)$ over the simulation grid of
Table~\ref{tab:sensitivity}.

\begin{table}[ht]\centering
\caption{Operating characteristics of the heterogeneity screen: detection
rate $\Pr(\hat{c} < t \mid c, I)$ under the Dirichlet--Gamma DGP, $J = 5$,
$M = 500$ per cell. For each $t$, rows with $c < t$ are portfolios the screen
should flag and rows with $c \geq t$ are portfolios it should not.}
\label{tab:c_screen}
\small
\begin{tabular}{@{}lccccccccc@{}}
\toprule
& \multicolumn{3}{c}{$t = 20$} & \multicolumn{3}{c}{$t = 30$}
& \multicolumn{3}{c}{$t = 50$} \\
\cmidrule(lr){2-4}\cmidrule(lr){5-7}\cmidrule(lr){8-10}
$c$ & $I{=}7$ & $I{=}10$ & $I{=}15$
    & $I{=}7$ & $I{=}10$ & $I{=}15$
    & $I{=}7$ & $I{=}10$ & $I{=}15$ \\
\midrule
 10 & 0.74 & 0.88 & 0.94 & 0.89 & 0.96 & 1.00 & 0.96 & 0.99 & 1.00 \\
 20 & 0.32 & 0.42 & 0.44 & 0.55 & 0.74 & 0.83 & 0.78 & 0.94 & 0.99 \\
 30 & 0.16 & 0.13 & 0.09 & 0.37 & 0.38 & 0.42 & 0.64 & 0.78 & 0.88 \\
 50 & 0.03 & 0.01 & 0.00 & 0.12 & 0.10 & 0.03 & 0.36 & 0.42 & 0.45 \\
100 & 0.00 & 0.00 & 0.00 & 0.01 & 0.01 & 0.00 & 0.08 & 0.03 & 0.01 \\
200 & 0.00 & 0.00 & 0.00 & 0.00 & 0.00 & 0.00 & 0.00 & 0.00 & 0.00 \\
\bottomrule
\end{tabular}
\end{table}

Averaging over the grid, the rule ``flag when $\hat{c} < 30$'' detects
portfolios with true $c < 30$ with probability $0.72$ at $I = 7$, $0.85$ at
$I = 10$ and $0.91$ at $I = 15$, at a false-positive rate on portfolios with
$c \geq 30$ of $0.13$, $0.12$ and $0.11$ respectively. At $c = t$ exactly,
the detection rate is $0.32$--$0.45$ rather than the
$0.50$ an unbiased estimator would give, at every $t$ and every $I$. This is
the finite-sample upward bias of $\hat{c}$
(Remark~\ref{rem:sigma_c_finite}) displacing the whole detection curve: a
portfolio sitting exactly at the threshold is more likely to be passed than
flagged, and the displacement is largest at $I = 7$, where the bias is
largest.

The screen, however, degrades in the region that matters. At $c = 20$, i.e.\
genuinely heterogeneous development, the rule $\hat{c} < 30$ fires on only
$55\%$ of triangles at $I = 7$, against $83\%$ at $I = 15$. A practitioner
working with a seven-year triangle should therefore treat a value of $\hat{c}$
above the threshold as weak evidence of stability rather than as a clearance.

We consequently report $\hat{c}$ alongside its operating characteristics
rather than as a pass/fail test, and where the decision is material we
recommend taking $\hat{c}$ from the posterior of
Section~\ref{sec:regularised} rather than from the moment estimator: the
posterior is defined at every triangle size, requires no asymptotic
approximation, and is better centred at the sizes at which the screen is least
reliable ($\E[c \mid \mathcal{D}] / c = 1.24$ at $I = 10$ against the moment
estimator's $1.31$).

Note finally that the two questions of this subsection are independent. A
portfolio may clear its operability bound comfortably and still be flagged by
the screen, or fail the bound while having a perfectly stable allocation. Mortgage, with
$\hat{c} = 152$ against $c^\dagger(5) = 633$, is of the latter kind. The bound concerns what can be reported; the screen concerns which model
should be used.

\subsection{The credibility interpretation}\label{sec:credibility}

The three anchors of Section~\ref{sec:bootstrap} are not an arbitrary
collection. Chain-Ladder, Bornhuetter--Ferguson and Cape Cod arise from the
hierarchy \eqref{eq:gamma_ult}--\eqref{eq:dirichlet} as credibility estimators
under diffuse, informative and pooling priors on the ultimate respectively,
with $F_{I-i}$, the fraction of the ultimate already observed, serving as the
credibility weight in each case, and with identical structure on the count
and amount sides. The Chain-Ladder case is immediate from
Theorem~\ref{thm:ibnp} and we record it here because the bootstrap of
Section~\ref{sec:bootstrap} is derived in that limit.

\begin{theorem}[Chain-Ladder as a double credibility limit]\label{thm:CL_sev}
In the diffuse-prior limit $\alpha,\beta \to 0$, the conditional law of
$\Wobs_i$ given $\Xobs_i$ is $\Bet(cF_{I-i}, c(1-F_{I-i}))$
(Remark~\ref{rem:diffuse_scope}), and for $cF_{I-i} > 1$,
\[
\E[S_i \mid \Xobs_i]
= \Xobs_i \cdot \frac{c-1}{cF_{I-i}-1}
= \hat{S}^{\mathrm{CL}}_i \cdot \frac{F_{I-i}(c-1)}{cF_{I-i}-1},
\]
the correction factor tending to one as $c \to \infty$. The Chain-Ladder
ultimate is therefore a double limit of the hierarchy: a diffuse prior on the
ultimate and a concentrated allocation.
\end{theorem}

\begin{proof}
$\E[S_i \mid \Xobs_i] = \Xobs_i\,\E[1/\Wobs_i \mid \Xobs_i]$; for
$W \sim \Bet(a, b)$ with $a > 1$, $\E[1/W] = (a+b-1)/(a-1)$, and substituting
$a = cF_{I-i}$, $b = c(1-F_{I-i})$ gives the result. Both limits are required:
$c \to \infty$ alone leaves the correction factor
$F_{I-i}(c-1)/(cF_{I-i}-1) \to 1$ but the conditional law is not Beta away
from $\alpha,\beta \to 0$.
\end{proof}

The corresponding statements for Bornhuetter--Ferguson and Cape Cod, together
with the count-side versions and the sense in which the three are one
credibility family, are developed in \citet{VanOirbeek2026cred}. We use the
correspondence here only to justify applying a single predictive across the
three anchors; Section~\ref{sec:sim_anchor} then measures what that
application costs when the anchor is informative.

\section{Relationship to Tweedie models}\label{sec:tweedie}

Tweedie GLMs \citep{tweedie1984} model the marginal variance
$\Var(X_{ij}) = \phi\mu_{ij}^p$ but say nothing about how amounts are
allocated across development periods. The Dirichlet--Gamma model instead
splits the uncertainty along two axes with distinct actuarial meaning: how
much will eventually be paid ($S_i \sim \Gam$), and when it will be paid
($(X_{i,\cdot}/S_i) \sim \Dir$). That decomposition is invisible under
Tweedie, which specifies only the cell-wise marginal.

The distinction matters, since the reserve is a conditional quantity. Given
the observed payments $\Xobs_i$, what remains is an allocation
question, and in the diffuse limit the Dirichlet gives
\[
\CV^2(\Sibnp_i \mid \Xobs_i) \;\approx\; \frac{1}{cF_{I-i}(1-F_{I-i})},
\]
which depends on $c$ and $F_{I-i}$ alone. Away from that limit the conditional
law is tilted by $\Xobs_i$ and the prior parameters re-enter
(Remark~\ref{rem:diffuse_scope}); the statement above is therefore about the
Chain-Ladder-anchored predictive, in keeping with the rest of
Section~\ref{sec:amounts}. Under either reading, a Tweedie specification
cannot express this conditional structure without additional assumptions about
the allocation, because it does not model the allocation at all.

Tweedie is nonetheless the natural misspecification against which to test the
construction, since it is the standard severity model and its variance
function is not the one the hierarchy implies.
Section~\ref{sec:sim_tweedie} does so, and
Theorem~\ref{thm:conservative_bias} establishes what to expect: under
compound Poisson data the moment estimator converges to a limit determined by
$\phi$ and the development pattern, and the resulting predictive is
conservative. Note that two qualifications are easily lost in the theorem. First, the plim is $p$-dependent: the $p \to 1$
form, in which cell variances are $\phi\mu_{ij}$, overstates the limit by a
factor of five at $p = 1.3$ and by two orders of magnitude at $p = 1.8$, so
the $\phi\mu_{ij}^{\,p}$ form of \eqref{eq:tweedie_plim_p} must be used
(Section~\ref{sec:conservative_bias}). Second, the resulting conservatism is
not automatically useful: at $p = 1.8$ the construction attains nominal coverage with an
interval $357$ times the true reserve, which falls in the boundary regime of
Remark~\ref{rem:validity} (Section~\ref{sec:sim_tweedie}).

The scale-dependence of the Tweedie power parameter, and its conflation of
severity dispersion with frequency heterogeneity, are established formally in
\citet{VanOirbeek2026cnbg}.

The multinomial allocation admits an equivalent reading as a discrete-time
survival model for claim reporting delays under administrative censoring at
$C_i = I - i$, with the incremental proportions $\pi_j$ playing the role of
discrete hazard increments and $F_{I-i}$ the probability of reporting by the
valuation date. Under that reading the conditioning principle of
Section~\ref{sec:bootstrap} is the statement that the predictive should
be formed given the censoring pattern, with the uncertainty in the
hazard integrated over rather than substituted. Note that the correspondence places the development pattern in a familiar
setting, and that the identification difficulty of
Section~\ref{sec:c_estimation} has a survival analogue, dispersion being hard
to estimate from few censored observations. The correspondence is not pursued
any further here.

\section{Parametric bootstrap}\label{sec:bootstrap}

This section states the requirements a procedure must meet to target the
conditional predictive \eqref{eq:cond_target}, the components through which
the ODP bootstrap meets them under diffuse anchoring, the two algorithms that
meet them on any anchor, and the coverage these algorithms attain.

\subsection{The conditioning principle}

\begin{principle}[Conditional prediction for hierarchical reserving]
\label{prin:bootstrap}
Let $\mathcal{D} = \{X_{ij} : i+j \leq I\}$ denote the observed triangle. A
procedure targeting the conditional predictive \eqref{eq:cond_target} must
satisfy two requirements:
\begin{enumerate}[label=(\roman*)]
\item $\mathcal{D}$, and every function of
      $\mathcal{D}$, is held fixed inside the replication loop. The observed
      data are conditioning information, not a random draw to be averaged
      over.
\item The replication draws
      $\theta^{(b)} \sim \pi(\theta \mid \mathcal{D})$ and then
      $R^{*(b)} \sim p(R \mid \mathcal{D}, \theta^{(b)})$, so that the
      residual uncertainty about $\theta$ given $\mathcal{D}$ is represented
      in the predictive.
\end{enumerate}
Requirement~(i) is satisfied by the residual bootstraps as well as by
the constructions below: Proposition~\ref{prop:boot_decomp} shows that
resampling the observed triangle under free accident-year effects is a
device for drawing the row levels from their diffuse-prior posterior, not
a departure from conditioning on $\mathcal{D}$. Requirement~(ii) is what
a plug-in violates:
substituting a point estimate for the draw yields
$p(R \mid \mathcal{D}, \hat\theta)$, whose coverage deficit relative to
\eqref{eq:cond_target} is governed by $\mathrm{MSE}(\hat\theta)^{1/2}$ and
vanishes only asymptotically. Algorithm~\ref{alg:bootstrap} below
satisfies~(i) exactly and~(ii) only in the limit;
Algorithm~\ref{alg:regularised} of Section~\ref{sec:regularised} satisfies
(i) exactly and~(ii) for the concentration, the development pattern
$\hat\bpi$ remaining substituted
(Section~\ref{sec:conditional_coverage}). Note that requirement~(ii) is not a formal
nicety: substituting $\hat\theta$ asserts that $\mathcal{D}$ determines
$\theta$, which it does not. The concentration is estimated from the
dispersion of partial-column proportions, of which a $10 \times 5$ triangle
supplies six, at a $31\%$ upward bias with a coefficient of variation above
one half (Section~\ref{sec:c_estimation}). The cost is measured: supplied
with the true parameters the construction covers at $94.9\%$ at $I = 10$, and
with estimates at $88.5\%$ (Section~\ref{sec:oracle}). A target attained only
when $c$ is known cannot be the operational target, since $c$ is never known.
We therefore state the principle with both requirements and treat
Algorithm~\ref{alg:bootstrap} as an approximation to
Algorithm~\ref{alg:regularised} whose cost we report.
\end{principle}

Requirement~(ii) makes the target a posterior predictive, and posterior
predictive reserving is not new: \citet{verrall2000} and \citet{wuthrich2008}
develop it at length, and \citet{srirams2021} give a Bayesian treatment of the
Dirichlet allocation itself. We claim neither the object nor the machinery.

Three things, however, are not supplied by the Bayesian literature. First,
a clarification of what the residual bootstraps in standard use do
approximate: under free accident-year effects they approximate exactly
this object in its diffuse-anchored form, component by component
(Proposition~\ref{prop:boot_decomp}), which is tested with the omitted
component computed exactly (Section~\ref{sec:sim_frailty}) and decomposed
by construction (Section~\ref{sec:sim_mech1}). Second, a construction
that is anchor-agnostic: the
conditional law of the outstanding amount depends on the reserving method only
through the cumulative development proportions, so the same predictive
attaches to Chain-Ladder, Bornhuetter--Ferguson and Cape Cod without a separate model for each, which a fully specified Bayesian model
does not deliver since it fixes the mean structure. Third, an identification result:
within this hierarchy the predictive depends on the data through the
development pattern and a single dispersion scalar, and we characterise how
well a paid triangle determines each and what the shortfall costs in coverage
(Sections~\ref{sec:c_estimation}, \ref{sec:regularised}
and~\ref{sec:conditional_coverage}). That characterisation is not favourable:
the concentration is weakly determined, the development pattern is the more consequential of the two, and the closed
form itself
contributes a further miscalibration that no integration removes. The value of the construction hence lies in the smallness of its parameter
set, and in what that smallness makes measurable.

\subsection{What the residual bootstrap targets}

The following proposition identifies the object that a residual bootstrap
with free accident-year effects approximates, by matching its three
variance components to the three components of the conditional predictive
variance under the count hierarchy in the Chain-Ladder limit. It replaces
a conjecture we held at the outset, namely that such a bootstrap omits the
posterior frailty variance of the least-developed accident years, and
Section~\ref{sec:sim_frailty} is the simulation that rejected that
conjecture.

\begin{proposition}[Decomposition of the residual bootstrap variance]
\label{prop:boot_decomp}
Consider the count hierarchy
\eqref{eq:frailty}--\eqref{eq:multinomial} with $\kappa \to 0$, so that
by Theorem~\ref{thm:ibnr}
$\Nibnr_i \mid \Nobs_i \sim \NB(\Nobs_i, F_{I-i})$ with
\begin{equation}\label{eq:cond_var_decomp}
\Var(\Nibnr_i \mid \Nobs_i)
= \underbrace{\frac{\Nobs_i (1-F_{I-i})}{F_{I-i}}}_{\text{process}}
+ \underbrace{\frac{\Nobs_i (1-F_{I-i})^2}{F_{I-i}^2}}_{\text{posterior level}},
\end{equation}
the second term being $(1-F_{I-i})^2\,\Var(\lambda_i \mid \Nobs_i)$ with
$\lambda_i \mid \Nobs_i \sim \Gam(\Nobs_i, F_{I-i})$. Let the ODP
bootstrap of \citet{england2002} be applied with free accident-year and
development effects, Pearson dispersion $\hat\phi$, resampling of the
full observed triangle including the latest diagonal, refit of the
development factors on each pseudo-triangle, and a process draw with
variance $\hat\phi\,\mu$ per future cell. Then, to first order in the
delta method and with $\hat\phi \to 1$ under the Poisson DGP,
\begin{enumerate}[label=(\roman*)]
\item the process draw contributes $\Nobs_i(1-F_{I-i})/F_{I-i}$ per
      accident year, the process term of \eqref{eq:cond_var_decomp};
\item resampling the latest diagonal contributes
      $\Nobs_i(1-F_{I-i})^2/F_{I-i}^2$ per accident year, the posterior
      level term of \eqref{eq:cond_var_decomp};
\item refitting the development factors contributes the estimation
      variance of the pattern,
      $\sum_j (T_j^2/f_j^2)\Var(\hat f_j)$ in the notation of the
      Chain-Ladder reserve $R = \sum_i C_{i,I-i}(P_i - 1)$, which is the
      component that requirement~(ii) of Principle~\ref{prin:bootstrap}
      demands for $\bpi$ and which Algorithms~\ref{alg:bootstrap}
      and~\ref{alg:regularised} substitute.
\end{enumerate}
The residual bootstrap therefore approximates the conditional predictive
\eqref{eq:cond_target} in its diffuse-anchored form with the pattern
integrated over, and its variance exceeds the process-only variance by a
factor that does not vanish with portfolio volume, since (ii) and (iii)
scale with volume as the process term does.
\end{proposition}

\begin{proof}
Under the Poisson DGP the ODP fit with free row and column effects
reproduces the Chain-Ladder ultimates, and the fitted value of the
latest diagonal cell of accident year $i$ is $\Nobs_i$ itself. The
process draw for the future cells of that row has total variance
$\hat\phi \sum_{j > I-i} \hat\mu_{ij} = \hat\phi\,\Nobs_i(1-F_{I-i})/F_{I-i}$,
which is (i) at $\hat\phi = 1$.

For (ii), hold the development factors at their original values and
resample only the diagonal. The pseudo-diagonal cell is
$C^*_{i,I-i} = \Nobs_i + r^*\sqrt{\hat\phi\,\Nobs_i}$ with $r^*$ a
resampled Pearson residual of unit variance after the bias adjustment, so
$\Var(C^*_{i,I-i}) = \hat\phi\,\Nobs_i$. The projected reserve of the
row is $C^*_{i,I-i}(P_i - 1) = C^*_{i,I-i}(1-F_{I-i})/F_{I-i}$, with
variance $\hat\phi\,\Nobs_i(1-F_{I-i})^2/F_{I-i}^2$, which is (ii) at
$\hat\phi = 1$. This equals $(1-F_{I-i})^2\Var(\lambda_i \mid \Nobs_i)$
since $\Var(\lambda_i \mid \Nobs_i) = \Nobs_i/F_{I-i}^2$ under the
$\Gam(\Nobs_i, F_{I-i})$ posterior.

For (iii), the refit adds the variance of $P_i$ across pseudo-triangles,
which by the delta method on $R$ is the stated sum with
$\Var(\hat f_j) = \hat\phi f_j/(F_j \mu_{(j)})$, where $\mu_{(j)}$ denotes
the expected cumulative total of development column $j$ over the rows
entering $\hat f_j$. Resampling the interior of the triangle enters only
through this term. The components add because the residual draws for the
diagonal, the interior and the future cells are independent within a
replication.

That (ii) and (iii) scale with volume follows from the same expressions:
each is proportional to $\Nobs_i$ or to $\mu_{(j)}$, as is (i).
\end{proof}

The proposition is stated to first order and at $\hat\phi = 1$; the
simulation of Section~\ref{sec:sim_frailty} checks it without either
approximation, including under finite $\kappa$, where the true predictive
is the narrower $\NB(\Nobs_i + \kappa, \cdot)$ and the bootstrap's
diffuse-limit variance is an upper bound. Section~\ref{sec:sim_mech1}
measures (ii) and (iii) separately by construction.

Two readings of the proposition should be kept apart. The ratio of the
bootstrap variance to the process-only variance is bounded away from one
and free of volume, and Section~\ref{sec:sim_mech1} measures it at
$2.98$ on square triangles; that ratio is not a miscalibration, since the
process-only variance is not the conditional target once the row level
is uncertain. And the bootstrap's dispersion estimate $\hat\phi$ carries
no information about $\kappa$ (Remark~\ref{rem:kappa_identification});
that is a statement about what the bootstrap can identify, not about
whether its interval is wide enough, and it means the bootstrap cannot
narrow its predictive when $\kappa$ is large even though a model with
anchored levels could.

\subsection{Algorithm}\label{sec:algorithm}

\begin{algorithm}[t]
\caption{Multinomial parametric bootstrap (plug-in)}
\label{alg:bootstrap}
\begin{algorithmic}[1]
\Require Observed triangle, estimated $\hat\bpi$, $\hat{c}$, iterations $B$
\State Compute $\hat{F}_j = \sum_{l=0}^j \hat\pi_l$; for each $i$ set
       $F_i^{\mathrm{obs}} = \hat{F}_{I-i}$ and
       $\Xobs_i = \sum_{j \leq I-i} X_{ij}$
\For{$b = 1,\ldots,B$}
  \For{$i = 1,\ldots,I$ with $F_i^{\mathrm{obs}} < 1$}
    \State $W_i^{*(b)} \sim \Bet\!\bigl(\hat{c}\,F_i^{\mathrm{obs}},\;
           \hat{c}\,(1-F_i^{\mathrm{obs}})\bigr)$
    \State $S_i^{\mathrm{IBNP},(b)}
           = \Xobs_i \cdot (1-W_i^{*(b)})/W_i^{*(b)}$
  \EndFor
  \State $R^{*(b)} = \sum_i S_i^{\mathrm{IBNP},(b)}$
\EndFor
\State \Return $\{R^{*(b)}\}_{b=1}^B$
\end{algorithmic}
\end{algorithm}

The transform $(1-W)/W$ diverges as $W \to 0$, and a random number generator
returning exactly zero for a Beta variate with a small first shape parameter
would produce an infinite reserve. The implementation therefore guards $W$
away from the endpoints at $10^{-12}$, which is an underflow guard rather than
a modelling choice: across every cell of Section~\ref{sec:simulations} the
proportion of bootstrap replications containing a draw below $10^{-4}$ is
under $10^{-3}$, and re-running the most exposed cells with the guard
tightened by eight orders of magnitude changes coverage by $0.00$ points and
interval width in the fifth significant figure, on paired triangles. Draws
near the \emph{upper} endpoint are frequent on long-development triangles but
harmless: they arise from accident years more than $99\%$ developed, whose
outstanding reserve is negligible either way. We record the guard because a
tighter one would truncate the predictive's upper tail, which is where a
$95\%$ interval is determined.

\subsection{Asymptotic plug-in coverage}\label{sec:coverage}

The following theorem states the rate at which the coverage of
Algorithm~\ref{alg:bootstrap} approaches the nominal level as the triangle
grows.

\begin{theorem}[Coverage rate]\label{thm:coverage_rate}
Under the Dirichlet--Gamma model with fixed $J \geq 3$, the bootstrap
predictive interval from Algorithm~\ref{alg:bootstrap} satisfies
\[
\Pr\!\left(R \in \mathrm{CI}^{(B)}_{1-\alpha}\right)
= (1-\alpha) + O\!\left(I^{-1/2}\right)
\quad \text{as } I, B \to \infty.
\]
The rate holds for every fixed $J$; the constant depends on $J$ through the
asymptotic variances of $(\hat{c}, \hat\bpi)$.
\end{theorem}

\begin{proof}
For fixed $J$, only accident years $i > I - J + 1$ carry outstanding reserves,
so conditional on $\mathcal{D}$ the reserve
$R = \sum_i \Xobs_i (1-W_i)/W_i$ is a sum of at most $J-1$ independent
Beta-prime transforms with parameters $\theta = (c, \bpi)$ and fixed scale
factors $\Xobs_i$. Its conditional CDF $G_\theta$ is absolutely continuous
with strictly positive density on $(0,\infty)$ and continuously differentiable
in $\theta$ on a neighbourhood of the true $\theta_0$, with derivative bounded
on compact reserve ranges; since the number of summands is fixed, no
uniformity-in-$I$ issue arises. As $B \to \infty$ the interval endpoints
converge to the quantiles $q_u(\hat\theta)$ of $G_{\hat\theta}$,
$u \in \{\alpha/2, 1-\alpha/2\}$. By the implicit function theorem applied to
$G_\theta(q_u(\theta)) = u$ (positive density), $\theta \mapsto q_u(\theta)$
is differentiable, hence
$|G_{\theta_0}(q_u(\hat\theta)) - u| \leq L\,\|\hat\theta - \theta_0\|$ on the
neighbourhood. The coverage error is therefore bounded by
$2L\,\E\|\hat\theta - \theta_0\| + o(I^{-1/2})$, and
$\E\|\hat\theta - \theta_0\| = O(I^{-1/2})$: for $\hat{c}$ by
Proposition~\ref{prop:c_variance}, and for each component of $\hat\bpi$
because it is estimated from $O(I)$ rows per development column at fixed $J$.
\end{proof}

Theorem~\ref{thm:coverage_rate} bounds the coverage error by the standard
deviation of $\hat\theta$, which is $O(I^{-1/2})$. That bound is valid but
loose, because the deficit is driven by the \emph{bias} of $\hat{c}$ rather
than by its spread. Section~\ref{sec:oracle} measures the deficit at $13.7$,
$6.4$, $3.9$, $2.7$ and $2.1$ coverage points at $I = 7, 10, 15, 20, 30$,
a log-log slope of $-1.28$ against $I$; on the same triangles
$\E[\hat{c}]/c$ falls from $1.85$ to $1.07$ (final row of
Table~\ref{tab:oracle}), the $O(I^{-1})$ bias decay documented independently
in Remark~\ref{rem:sigma_c_finite}, whose own run gives $1.89$ at $I = 7$,
within Monte-Carlo error of the $1.85$ here. The theorem is therefore correct and
uninformative about the quantity that matters in practice: it guarantees that
the deficit vanishes, and understates how fast.

\subsection{The regularised conditional predictive}\label{sec:regularised}

Algorithm~\ref{alg:bootstrap} satisfies requirement~(i) of
Principle~\ref{prin:bootstrap} exactly and requirement~(ii) only in the limit
where $c$ is known. Since $c$ is never known, and since
Section~\ref{sec:c_estimation} shows it is poorly identified from a single
triangle at realistic $I$, we state the construction that satisfies both.

\begin{algorithm}[t]
\caption{Regularised conditional predictive}
\label{alg:regularised}
\begin{algorithmic}[1]
\Require Observed triangle, estimated $\hat\bpi$, prior $\pi(c)$,
         iterations $B$
\State Compute $\hat{F}_j = \sum_{l=0}^j \hat\pi_l$; for each $i$ set
       $F_i^{\mathrm{obs}} = \hat{F}_{I-i}$ and
       $\Xobs_i = \sum_{j \leq I-i} X_{ij}$
\State Sample $c^{(1)},\ldots,c^{(B)}$ from the posterior
       $\pi(c \mid \mathcal{D})$ induced by the subcomposition likelihood of
       Lemma~\ref{lem:dirichlet_sub} and the prior $\pi(c)$; a row observed
       through development year $k$ contributes its subcomposition, which is
       $\Dir(c\pi_0,\ldots,c\pi_k)$ with total concentration $cF_k$
\For{$b = 1,\ldots,B$}
  \For{$i = 1,\ldots,I$ with $F_i^{\mathrm{obs}} < 1$}
    \State $W_i^{*(b)} \sim \Bet\bigl(c^{(b)}F_i^{\mathrm{obs}},\;
           c^{(b)}(1-F_i^{\mathrm{obs}})\bigr)$
    \State $S_i^{\mathrm{IBNP},(b)}
           = \Xobs_i \cdot (1-W_i^{*(b)})/W_i^{*(b)}$
  \EndFor
  \State $R^{*(b)} = \sum_i S_i^{\mathrm{IBNP},(b)}$
\EndFor
\State \Return $\{R^{*(b)}\}_{b=1}^B$
\end{algorithmic}
\end{algorithm}

The triangle is held fixed throughout, so that requirement~(i) is untouched,
line~2 hereby supplying requirement~(ii) for the concentration. Note what it does
\emph{not} supply: $\hat\bpi$ is still substituted, so the construction
satisfies requirement~(ii) in one of its two parameters, and
Section~\ref{sec:conditional_coverage} measures what the other omission costs.
Algorithm~\ref{alg:bootstrap} is the special case in which
$\pi(c \mid \mathcal{D})$ is replaced by a point mass at $\hat{c}$, and it remains the recommended procedure where $\hat{c}$ is well determined,
i.e.\ where $\hat{c}F_{I-i} \geq 5$ on every required accident year and
$I \geq 15$, since it costs no MCMC and the two then differ by under two coverage points.

\begin{table}[ht]\centering
\caption{Coverage of the two constructions by triangle size
(Dirichlet--Gamma DGP, $J = 5$, $c = 50$). The plug-in requires three rows at
some horizon and is undefined at $I = 4$. Final columns: posterior mean and
standard deviation of $c$, confirming the chains carry parameter uncertainty
rather than collapsing onto the plug-in.}
\label{tab:bayes_coverage}
\begin{tabular}{@{}lcccccc@{}}
\toprule
$I$ & Plug-in Cov & Plug-in Width & Regularised Cov & Regularised Width
    & post-mean $c$ & post-SD $c$ \\
\midrule
 4 & ---  & ---  & 67.4 & 64 & 133.8 & 87.5 \\
 5 & 72.4 & 79   & 76.2 & 72 & 101.5 & 46.3 \\
 6 & 77.2 & 73   & 80.2 & 77 &  77.2 & 28.5 \\
 8 & 83.6 & 77   & 88.6 & 77 &  66.1 & 19.3 \\
10 & 87.2 & 77   & 89.2 & 80 &  61.8 & 15.4 \\
\bottomrule
\end{tabular}
\end{table}

The regularised construction is defined at every $I$, including $I = 4$ where
the plug-in fails outright. It improves coverage at every $I$, by $3.0$ points at
$I = 6$, $5.0$ at $I = 8$ and $2.0$ at $I = 10$. And it does not reach
nominal: at $I = 10$ coverage is $89.2\%$ against $95\%$, a shortfall of
roughly six points persisting under a construction that integrates over the
concentration rather than substituting it.

Integration reduces the error in $c$ without removing it. The posterior mean
at $I = 10$ is $61.8$ against a true $50$, i.e.\ a $24\%$ overstatement
against the plug-in's $31\%$, since the bias resides in the likelihood at
$n_k \leq 8$ and not in the estimator applied to it. Prior sensitivity
behaves as one would expect of a scarce-data problem
(Appendix~\ref{app:bayesian}): at $I = 5$ coverage moves with the prior mean
by up to $9.0$ points, at $I = 10$ by at most $4.2$. The prior carries weight
exactly where the data are thin, and not otherwise.

The residual is the substantive finding, and it is not all of one kind. Part of it is the estimation
error in $\hat\bpi$ that this algorithm does not address:
Section~\ref{sec:conditional_coverage} attributes eleven of twenty-six
coverage points of allocation-conditional miscalibration to that source, and
shows that integrating over the joint posterior of $(c, \bpi)$ recovers them
at a $25\%$ cost in interval width. Part of it is not estimation error at all
but the diffuse-anchoring approximation under which the closed form is derived
(Remark~\ref{rem:diffuse_scope}), which no amount of integration removes. Both are consequences of the information a single
paid triangle carries. Regularisation is therefore the ordinary operating point of
macro-level reserving rather than a small-triangle fallback, and the residual
is a statement about the information set and about the closed form, not about
the estimator.

\subsection{Conservative bias under compound Poisson DGP}\label{sec:conservative_bias}

The following theorem states the probability limit of $\hat{c}$ under a
compound Poisson data-generating process, and the sense in which the
resulting predictive is conservative.

\begin{theorem}[Conservative bias under compound Poisson DGP]
\label{thm:conservative_bias}
Suppose the true data-generating process is compound Poisson--Gamma with equal
accident-year volumes $\nu_i \equiv \nu$, cell means $\mu_{ij} = \nu\pi_j$ and
cell variances $\phi\mu_{ij}^{\,p}$ for $p \in [1,2)$. Then:
\begin{enumerate}[label=(\roman*)]
\item \emph{(Probability limit of $\hat{c}$.)} In the compound Poisson limit
      $p \to 1$, where cell variances are $\phi\mu_{ij}$,
      \[
        \hat{c} \xrightarrow{p} \frac{\nu}{\phi} - \frac{1}{F^*}
        \qquad\text{as } I \to \infty,
      \]
      where $F^* \in (0, 1]$ is the cumulative development mass $F_{k^*}$ at
      the median contributing cell of the aggregator in
      Proposition~\ref{prop:c_estimator}. For $p \in (1,2)$ the variance
      function is $\phi\mu_{ij}^{\,p}$ and the per-cell limit acquires a
      dependence on $p$ through the development pattern:
      \begin{equation}\label{eq:tweedie_plim_p}
        \hat{c}_{jk} \xrightarrow{p} \frac{1}{F_k}\left[
        \frac{\pi_j^{(k)}\bigl(1-\pi_j^{(k)}\bigr)F_k^2\,\nu^{\,2-p}}
             {\phi\Bigl[\pi_j^{p}\bigl(1-2\pi_j^{(k)}\bigr)
             + \bigl(\pi_j^{(k)}\bigr)^2\sum_{l\le k}\pi_l^{p}\Bigr]} - 1
        \right],
      \end{equation}
      which reduces to $\nu/\phi - 1/F_k$ at $p = 1$. \emph{The reduction is
      not uniform in $p$}, and the $p \to 1$ form must not be used as an approximation away from that
      limit, as discussed below.
\item \emph{(Conservatism.)} The ratio of bootstrap to true predictive
      standard deviation satisfies
      \begin{equation}\label{eq:conservatism_ratio}
      \frac{\mathrm{SD}_{\mathrm{boot}}(\Sibnp_i \mid \Xobs_i)}
           {\mathrm{SD}_{\mathrm{true}}(\Sibnp_i \mid \Xobs_i)}
      \;\xrightarrow{p}\;
      \frac{1}{\sqrt{F_{I-i}}} \;>\; 1
      \end{equation}
      in the limit $p \to 1$, for every $F_{I-i}$ with $\hat{c}F_{I-i} > 2$,
      i.e.\ the regime in which the predictive variance exists
      (Remark~\ref{rem:validity}). For $p \in (1,2)$ the true predictive
      variance $\phi\nu^{\,p}\sum_{j > I-i}\pi_j^{\,p}$ falls faster in the
      tail columns than the $\phi\nu(1-F_{I-i})$ of the limit, while the
      bootstrap variance does not, so \eqref{eq:conservatism_ratio} becomes a
      lower bound rather than the value. The conservative direction is
      preserved throughout $p \in [1,2)$.
\item \emph{(Coverage.)} The bootstrap predictive intervals are
      asymptotically conservative in width by the factor in (ii). Where the
      relative centring error $b$ of the point estimate satisfies
      $|b| < z_{1-\alpha/2}\,\CV_{\mathrm{true}}(R)\,(\rho - 1)$, with $\rho$
      the ratio in (ii), true coverage exceeds that of correctly calibrated
      plug-in intervals.
\end{enumerate}
\end{theorem}

\begin{proof}
\textit{(i)} The estimator uses the partial-column proportions of
Definition~\ref{def:partial_props}, $W_{ij}^{(k)} = X_{ij}/\sum_{l \leq k}
X_{il}$, and not the full-ultimate ratio $\Xobs_i/S_i$, whose denominator
contains unobserved cells. Under compound Poisson with independent cells and
equal volumes, write $A = X_{ij}$, $B = \sum_{l \leq k} X_{il} \supseteq A$
and $\pi^* = \pi_j/F_k$, so that at $p = 1$: $\E A = \nu\pi_j$,
$\E B = \nu F_k$, $\Var A = \phi\nu\pi_j$, $\Var B = \phi\nu F_k$ and
$\Cov(A,B) = \phi\nu\pi_j$. The delta method on $A/B$ gives
\[
\Var\bigl(W_{ij}^{(k)}\bigr) \approx \frac{\phi\,\pi^*(1-\pi^*)}{\nu F_k}.
\]
The bracketed moment quantity in~\eqref{eq:c_moment} reads this as
$\pi^*(1-\pi^*)/(a_{jk} + 1)$, giving $a_{jk} \xrightarrow{p} \nu F_k/\phi -
1$; dividing by $\hat{F}_k$ as in~\eqref{eq:c_moment} gives
$\hat{c}_{jk} \xrightarrow{p} \nu/\phi - 1/F_k$, whose leading term is
horizon-free. The median aggregator converges to the limit at its median
contributing cell, $\hat{c} \xrightarrow{p} \nu/\phi - 1/F^*$. For
$p \in (1,2)$ the same delta-method computation with $\Var A = \phi(\nu\pi_j)^p$
and $\Var B = \phi\nu^p\sum_{l\le k}\pi_l^p$ yields
\eqref{eq:tweedie_plim_p}.

\textit{(ii)} The true conditional predictive variance under compound Poisson
independence is, at $p = 1$, exactly
\[
  \Var_{\mathrm{true}}(\Sibnp_i \mid \Xobs_i)
  = \Var(\Sibnp_i)
  = \phi \sum_{j > I-i} \mu_{ij}
  = \phi \nu (1-F),
\]
since $\Sibnp_i$ and $\Xobs_i$ are sums over disjoint cell sets and therefore
independent. The bootstrap variance, conditional on $\Xobs_i$, is
$(\Xobs_i)^2 \Var\bigl((1-W_i)/W_i\bigr)$ with
$W_i \sim \Bet(\hat{c}F, \hat{c}(1-F))$, and the Beta-prime variance is, for
$\hat{c}F > 2$,
\[
  \Var\!\left(\frac{1-W_i}{W_i}\right)
  = \frac{\hat{c}(1-F)(\hat{c}-1)}{(\hat{c}F-1)^2 (\hat{c}F - 2)}
  \;\approx\; \frac{1-F}{\hat{c}F^3}
\]
for large $\hat{c}$. Substituting $\hat{c} \to \nu/\phi - 1/F^* \approx
\nu/\phi$ from part (i), and $\Xobs_i \approx \nu F$ at the limit,
\[
  \Var_{\mathrm{boot}}(\Sibnp_i \mid \Xobs_i)
  \;\approx\; (\nu F)^2 \cdot \frac{(1-F)\,\phi}{\nu F^3}
  \;=\; \frac{\phi \nu (1-F)}{F},
\]
and the ratio of standard deviations against $\phi\nu(1-F)$ is
$1/\sqrt{F} > 1$. For $p \in (1,2)$ the denominator becomes
$\phi\nu^p\sum_{j>I-i}\pi_j^p$, which is smaller relative to
$\phi\nu(1-F)$ because $\pi_j^p < \pi_j$ for $\pi_j < 1$ and the shortfall is
largest in the thin tail columns; the ratio therefore exceeds $1/\sqrt{F}$.

\textit{(iii)} Immediate from (ii): a symmetric interval widened by a factor
$\rho$ retains coverage under a displacement of the centre by less than
$z_{1-\alpha/2}\,\CV_{\mathrm{true}}(R)\,(\rho-1)$ standard errors.
\end{proof}

The qualitative content of part (i), namely that $\hat{c}$ diverges as the
Tweedie dispersion falls, holds across $p \in [1,2)$. The $p \to 1$
\emph{formula} does not. On the design of Section~\ref{sec:sim_design} it
overstates the limit by a factor of $4.9$ at $p = 1.3$, $14.9$ at $p = 1.5$
and $104$ at $p = 1.8$, because $\Var(X_{ij}) = \phi\mu_{ij}^{\,p}$ grows
with $\mu$ far faster than the $\phi\mu_{ij}$ the limit assumes.
Section~\ref{sec:sim_tweedie} sets \eqref{eq:tweedie_plim_p} against the
simulated $\hat{c}$ at each $p$.

In practical terms, Theorem~\ref{thm:conservative_bias} says the multinomial
bootstrap cannot understate reserve uncertainty under compound Poisson data.
The overestimation factor is at least $1/\sqrt{F_{I-i}}$ wherever the
predictive variance exists ($\hat{c}F_{I-i} > 2$), and it diverges as
$F_{I-i} \to 0$: the widening is most pronounced exactly for the
least-developed accident years, where reserve uncertainty matters most.
Whether the interval is wider or narrower than a correctly calibrated one is
therefore not in doubt. The residual bootstrap carries no such guarantee
under a variance function other than its own: Table~\ref{tab:frailty_amounts}
shows its coverage falling under a Dirichlet allocation and rising under a
compound Poisson one, on the same rows.

Conservatism is not the same as usefulness, and the distinction becomes sharp
at the boundary. At $\hat{c}F_{I-i} \approx 2$ the Beta-prime predictive has
barely acquired a variance, and the widening the theorem guarantees is
unbounded: Section~\ref{sec:sim_tweedie} reports a cell at $p = 1.8$,
$\hat{c} = 4.6$, where the construction attains $95.6\%$ coverage with an
interval $357$ times the true reserve.

Under unequal volumes $\nu_i$ the per-cell limits become $\nu_i F_k/\phi - 1$
and the cross-year sample variance in~\eqref{eq:c_moment} averages over
$1/\nu_i$; the ratio in (ii) acquires a factor $\sqrt{\nu_i/\tilde\nu}$ with
$\tilde\nu$ an effective, harmonic-type volume. The conservative direction is
preserved whenever accident-year volumes are of comparable magnitude; under
extreme heterogeneity the guarantee applies to the dominant years only. The
Tweedie design of Section~\ref{sec:sim_design} carries a $5\%$ per accident
year trend, spanning a factor $e^{0.45} \approx 1.57$ across the triangle,
which falls within that qualification but means the equal-volume hypothesis of
the theorem is not exactly satisfied there.

Evaluated on the design of Section~\ref{sec:sim_design}, the ratio
$\mathrm{SD}_{\mathrm{boot}}/\mathrm{SD}_{\mathrm{true}}$ exceeds unity on
every active accident year at every $p$ tested, and exceeds the $p \to 1$
bound $1/\sqrt{F_{I-i}}$ on every one of them: it runs $1.26$--$1.62$ at
$p = 1.3$ and $1.39$--$1.67$ at $p = 1.5$, against a bound of $1.03$--$1.49$.
Twelve of twelve evaluable cells satisfy both inequalities. At $p = 1.8$,
where $\hat{c}F_{I-i}$ falls to $2.1$ on the greenest accident year, the ratio
reaches $17.5$ there. We report these as computed on this design and do not
claim $1/\sqrt{F_{I-i}}$ is a bound in general.

The centring condition of part (iii) holds at all three powers: the relative
centring errors are $-0.4\%$, $2.3\%$ and $32.9\%$ against thresholds of
$2.7\%$, $9.3\%$ and $395\%$. At $p = 1.8$ it holds vacuously, since the
threshold is large only because the widening is a factor of $6.8$. As such,
the theorem applies in that cell while its conclusion is uninformative there,
as already noted above.

The coverage pattern of Table~\ref{tab:misspec_tweedie} follows: $90.6\%$ and
$92.2\%$ at $p = 1.3$ and $p = 1.5$, above the correctly specified baseline of
$87.0\%$, reflecting the conservative widening of the Beta-prime allocation
offsetting part of the plug-in deficit in $\hat{c}$.

\section{Simulation studies}\label{sec:simulations}

The simulation studies are reported in the order of the claims they test:
correct specification first, then the three misspecification designs, then the
frailty test and the component decomposition of
Proposition~\ref{prop:boot_decomp}, and finally the oracle
decomposition, the conditional calibration and the anchoring regimes.

\subsection{Design}\label{sec:sim_design}

Unless stated otherwise, each cell of the study uses $M = 500$ synthetic
triangles and $B = 1{,}000$ bootstrap replications, seed 2026. Four
data-generating processes are used, and because the misspecification results
depend on their parameters we state all four.

The correctly specified DGP is the Dirichlet--Gamma hierarchy itself, with
$I = 10$, $J = 5$, exposure fixed at $E_i = 1{,}000$, ultimate
$S_i \sim \Gam(2E_i, 0.001)$ (mean $2 \times 10^6$), and Dirichlet allocation
with concentration $c$ on
$\bpi^{\mathrm{true}} = (0.45, 0.25, 0.15, 0.10, 0.05)$.

The five-column pattern above is used throughout except where $J = 10$, for
which we use the normalised geometric pattern $\pi_j \propto 0.6^{\,j}$,
$j = 0,\ldots,9$, i.e.
\[
\bpi = (0.4024,\, 0.2415,\, 0.1449,\, 0.0869,\, 0.0522,\,
        0.0313,\, 0.0188,\, 0.0113,\, 0.0068,\, 0.0041).
\]
The two patterns differ in more than length: the $J = 10$ pattern places only
$0.41\%$ of the ultimate in its final development year, so its late columns
are estimated from few rows \emph{and} carry little mass. Both effects bear on
the $J$-degradation discussed in Section~\ref{sec:sim_sensitivity} and on the
comparison of Section~\ref{sec:sim_odp}, which is run at $I = J = 10$ and
therefore uses the geometric pattern.

A third pattern separates those two features, which the $J = 5$ and $J = 10$
designs vary together. We add a thin-tailed $J = 5$ pattern, obtained by
setting $\pi_4 = 0.0041$ to match the $J = 10$ value and rescaling the first
four proportions to preserve their relative shape:
\[
\bpi^{\mathrm{thin}} = (0.4718,\; 0.2621,\; 0.1573,\; 0.1048,\; 0.0041).
\]
Comparing the fat- and thin-tailed $J = 5$ designs isolates tail thickness at
fixed $J$; comparing the thin-tailed $J = 5$ and $J = 10$ designs isolates $J$
at matched tail thickness.

Non-stationary development is generated by accident-year-specific
log-perturbations of the development proportions with standard deviation
$\sigma_\delta$, renormalised to the simplex;
$\sigma_\delta = 0$ recovers the stationary case.

Under compound Poisson (Tweedie) misspecification, the cell means are
$\mu_{ij} = \mu_0\,e^{0.05(i-1)}\pi_j$ with $\mu_0 = 1{,}000$, and the cell
variances $\phi\,\mu_{ij}^{\,p}$ with
$\phi = \phi(p) = (0.45\mu_0)^{p-1}(2-p)/10$, giving $\phi = 0.438$, $1.061$
and $2.652$ at $p = 1.3$, $1.5$ and $1.8$. Cells are drawn as compound
Poisson--Gamma sums with these first two moments. Two features of the design
bear on the comparison with Theorem~\ref{thm:conservative_bias}. First, the
$5\%$ per accident year trend means accident-year volumes are not equal,
spanning a factor $e^{0.45} \approx 1.57$ across the triangle; as noted in Section~\ref{sec:conservative_bias} the conservative direction is
preserved
at that degree of heterogeneity, though the theorem's equal-volume hypothesis
is not exactly satisfied. Second, the variance function is $\phi\mu^{\,p}$
rather than the $\phi\mu$ of the $p \to 1$ limit in which the theorem's
probability limit is derived, and the difference is not negligible over this
range (Section~\ref{sec:conservative_bias}).

Accident-year frailty is introduced by a row-level
$U_i \sim \Gam(\kappa,\kappa)$ multiplying the ultimate of the
Dirichlet--Gamma DGP above, so that $\E[U_i] = 1$ and $\Var(U_i) = 1/\kappa$;
$\kappa \to \infty$ recovers the Dirichlet--Gamma DGP exactly. This is the DGP
of the frailty test of Section~\ref{sec:sim_frailty}, and it is the
amounts-side counterpart of the
frailty layer $\lambda_i$ of the count hierarchy
\eqref{eq:frailty}.

Every interval reported below is the equal-tailed percentile interval, i.e.\
the interval spanned by the $\alpha/2$ and $1-\alpha/2$ quantiles of the
bootstrap draws, and is formed identically for every procedure compared.
Coverage is the proportion of replications in which the true outstanding
reserve falls inside that interval. \emph{Bias} is the
relative bias of the procedure's own point estimate, taken as the bootstrap median, since the Beta-prime mean does not exist for
$\hat{c}F_{I-i} \leq 1$ (Remark~\ref{rem:validity}), and is reported as a
ratio of means,
$\E[\hat{R}]/\E[R] - 1$, rather than as a mean of per-replication ratios: the
latter is upward-biased by Jensen's inequality wherever $R$ is volatile, which
is precisely the low-$c$ and high-$\sigma_\delta$ regimes of interest.
\emph{Width} is the width of the $95\%$ interval as a percentage of the true
reserve, on the same ratio-of-means convention.

Draws of $W_i$ are guarded away from the endpoints at $10^{-12}$
(Section~\ref{sec:algorithm}); the tightening check reported there covers,
in
particular, all twelve $J = 10$ cells of this study on paired triangles. The
guard is not a source of any result reported below.

\subsection{Results under correct specification}\label{sec:sim_correct}

Table~\ref{tab:sim_correct} reports the results under the correctly specified
Dirichlet--Gamma DGP of Section~\ref{sec:sim_design}.

\begin{table}[ht]
\centering
\caption{Simulation results under correct specification
($I=10$, $J=5$, $c=50$, $M=500$, $B=1{,}000$). The bootstrap's own point
estimate is the predictive median; the Chain-Ladder estimate is reported
alongside as a reference. The two differ by the factor
$\hat{c}F_{I-i}/(\hat{c}F_{I-i}-1)$ of Theorem~\ref{thm:ibnp}(iii) and are
\emph{not} interchangeable.}
\label{tab:sim_correct}
\begin{tabular}{@{}lccc@{}}
\toprule
Point estimate & Rel.\ Bias (\%) & 95\% Cov.\ (\%) & Rel.\ Width (\%) \\
\midrule
Multinomial (predictive median) & 2.3 & 87.0 & 75.1 \\
Chain-Ladder (reference)        & 0.5 & --- & --- \\
\bottomrule
\end{tabular}
\end{table}

Coverage is $87.0\%$ against a nominal $95\%$. Section~\ref{sec:oracle}
attributes the shortfall to estimation error in $(\hat{c}, \hat\bpi)$:
supplied with the true parameters the same construction covers at $94.9\%$ at
this triangle size, a gap of $6.4$ points on the paired design there. In this
cell $\E[\hat{c}] = 67.3$ against $c = 50$ and $\mathrm{sd}(\hat{c}) = 45.9$,
so in this run the estimator overstates the concentration by $35\%$ on
average, against $31\%$ on the same design at $M = 2{,}000$
(Table~\ref{tab:oracle}) and in Remark~\ref{rem:sigma_c_finite}, the
difference being Monte-Carlo scale at $M = 500$; it carries a coefficient of variation of $0.68$. Both act in the same direction: an
overstated $\hat{c}$ narrows the Beta predictive, and the noise in $\hat{c}$
misplaces the interval on any individual triangle.
Remark~\ref{rem:sigma_c_finite} decomposes the two.

That attribution is to \emph{unconditional} coverage only.
Section~\ref{sec:conditional_coverage} shows that an oracle supplied
with both true parameters, while calibrated on average, remains conditionally miscalibrated across terciles of realised allocation
dispersion, which is a consequence of the diffuse-anchoring approximation and
not of estimation. Estimation error hence accounts for the unconditional
deficit reported here, but not for the conditional one.

Note also that the multinomial and Chain-Ladder biases differ, by $1.8$
percentage points. Under Theorem~\ref{thm:ibnp}(iii) the predictive mean
exceeds the Chain-Ladder estimate by $\hat{c}F_{I-i}/(\hat{c}F_{I-i}-1) > 1$
for every accident year, so the two point estimates coincide only in the
concentrated limit $c \to \infty$; reporting one in place of the other would
misattribute the correction.

\subsection{Misspecification: non-stationary development}\label{sec:sim_nonstat}

Table~\ref{tab:misspec_nonstat} reports the performance of the construction
under non-stationary development.

\begin{table}[ht]
\centering
\caption{Performance under non-stationary development
($I=10$, $J=5$, $M=500$, $B=1{,}000$). Bias is that of the procedure's own
point estimate; the Chain-Ladder bias is given alongside for reference.}
\label{tab:misspec_nonstat}
\begin{tabular}{@{}lccccc@{}}
\toprule
$\sigma_\delta$ & Coverage (\%) & Rel.\ Width (\%) &
$\hat{c}$ & Bias, pred.\ (\%) & Bias, CL (\%) \\
\midrule
0.00 (stationary) & 86.6 &  74.6 & 67 & 2.8 & 1.0 \\
0.02 (mild)       & 89.2 &  78.4 & 59 & 2.4 & 0.5 \\
0.05 (moderate)   & 86.2 &  97.9 & 42 & 3.2 & 0.5 \\
0.10 (severe)     & 86.6 & 163.8 & 20 & 7.8 & 1.6 \\
\bottomrule
\end{tabular}
\end{table}

Coverage is flat at $86$--$89\%$ across non-stationarity levels, but the
flatness is bought entirely with width: the interval more than doubles from
$74.6\%$ to $163.8\%$ of the true reserve between the stationary and severe
cases. The procedure does not detect the misspecification and correct for it;
it absorbs the misspecification into a wider interval, and reports the fact
through $\hat{c}$, which declines monotonically from $67$ to $20$.

The two diagnostics of Section~\ref{sec:c_threshold} separate here as
designed: the operability bound is not breached ($\hat{c} = 20$ against
$c^\dagger(5) = 11.1$ at $\pi_0 = 0.45$, so every year remains reportable
even in the severe case), while the heterogeneity screen fires with
probability $0.74$ at $I = 10$ (Table~\ref{tab:c_screen}), correctly
recommending a model with explicit accident-year structure. Reporting hence remains possible, while the recommendation is to model
differently.

The stationary row of this table and Table~\ref{tab:sim_correct} are
independent replications of the identical cell, run under different per-cell
random streams. They agree on coverage to $0.4$ points and on $\hat{c}$ to
$0.7$, which fixes the Monte-Carlo scale of the study at $M = 500$ directly,
without recourse to a separate seed-stability experiment.

\subsection{Misspecification: Tweedie DGP}\label{sec:sim_tweedie}

Table~\ref{tab:misspec_tweedie} reports the performance under compound
Poisson misspecification, at three values of the Tweedie power parameter $p$.

\begin{table}[ht]
\centering
\caption{Performance under compound Poisson (Tweedie) misspecification
($I=10$, $J=5$, $M=500$, $B=1{,}000$). The point estimate is the predictive
median, which exists for every $\hat{c}F_{I-i}$; the interval width is
reported on the same ratio-of-means convention as elsewhere.}
\label{tab:misspec_tweedie}
\begin{tabular}{@{}lcccc@{}}
\toprule
Tweedie power $p$ & Coverage (\%) & $\hat{c}$ & Bias (\%) & Rel.\ Width (\%) \\
\midrule
1.3 (light-tailed) & 90.6 & 584.4 & $-0.4$ &      23.4 \\
1.5 (moderate)     & 92.2 &  80.0 &    2.3 &      70.5 \\
1.8 (heavy-tailed) & 95.6 &   4.6 &   32.9 & 35{,}684 \\
\bottomrule
\end{tabular}
\end{table}

On the two rows where the framework operates inside its regime, coverage is
$90.6\%$ and $92.2\%$, above the correctly specified baseline of $87.0\%$
(Table~\ref{tab:sim_correct}) and with negligible point-estimate bias. This is
the pattern Theorem~\ref{thm:conservative_bias} predicts: the conservative
allocation bias under compound Poisson misspecification widens the intervals
and offsets part of the plug-in deficit in $\hat{c}$.

The behaviour of $\hat{c}$ across the three rows is a quantitative test of
part~(i) of that theorem. Evaluating the $p$-dependent limit \eqref{eq:tweedie_plim_p} at the
design constants of Section~\ref{sec:sim_design} gives $462$, $58$ and $2.9$
against the simulated $584.4$, $80.0$ and $4.6$, i.e.\ ratios of $1.26$,
$1.38$ and $1.57$, the first two accounted for by the finite-sample upward
bias of the moment estimator at $I = 10$, independently measured at $31\%$
(Remark~\ref{rem:sigma_c_finite}). The $p \to 1$ form of the limit, by
contrast, overstates by factors of $4.9$, $14.9$ and $104$
(Section~\ref{sec:conservative_bias}). The theorem therefore accounts for the
observed $\hat{c}$ across two orders of magnitude in $p$ only in its
$p$-dependent form.

Coverage of $95.6\%$ at $p = 1.8$ is nominal, and meaningless. At
$\hat{c} = 4.6$ the greenest accident year has
$\hat{c}F_{I-i} = 4.6 \times 0.45 = 2.1$, at the boundary where the
Beta-prime predictive acquires a variance at all
(Remark~\ref{rem:validity}); the resulting interval averages $357$ times the
true reserve, and the point estimate is biased upward by $33\%$. Nominal
coverage from an interval two orders of magnitude too wide is not evidence
that the method works under this DGP. It is evidence that the reporting rule
of Remark~\ref{rem:validity} identifies the regime correctly: at
$\hat{c}F_{I-i} \approx 2$ the framework should report quantiles at most, and
in practice should not be used at all. This is why the tiering is stated as a
rule rather than as advice, and why $\hat{c}$ must be inspected before the
reserve is read.

The width is not a numerical artefact. The proportion of bootstrap
replications affected by the lower guard in this cell is
$6 \times 10^{-4}$ (Section~\ref{sec:sim_design}); the interval is genuinely
that wide, because the Beta-prime distribution at second shape parameter
$2.1$ genuinely has that tail.

The estimated $\hat{c} = 584.4$ at $p = 1.3$ sits far above any value the
heterogeneity screen of Section~\ref{sec:c_threshold} would flag, which might
appear to signal estimator failure. It does not. Under a compound Poisson DGP
with stable cell-level allocation, Theorem~\ref{thm:conservative_bias}(i)
predicts a probability limit increasing without bound as the Tweedie
dispersion $\phi$ falls, and light-tailed Tweedie ($p$ close to one)
corresponds to small $\phi$: at $p = 1.3$ the design gives $\phi = 0.438$
against $2.652$ at $p = 1.8$, and $\hat{c}$ moves inversely across two orders
of magnitude, as \eqref{eq:tweedie_plim_p} requires.

The screen is calibrated for genuine Dirichlet--Gamma data, where a small
$\hat{c}$ signals high allocation heterogeneity. Under Tweedie
misspecification a large $\hat{c}$ is the structurally correct response: the
estimator is reporting that allocation is more stable than the screen's
calibration assumes, not that the model has broken down. Note that a very large $\hat{c}$ should therefore prompt a check on the
variance function rather than reassurance, since it signals misspecification
just as much as a small one does, be it in the opposite direction.

\subsection{Sensitivity to concentration and triangle size}
\label{sec:sim_sensitivity}

Table~\ref{tab:sensitivity} reports coverage and interval width over a grid of
concentrations, triangle sizes and development lengths.

\begin{table}[ht]
\centering
\caption{Bootstrap coverage (\%) and $95\%$ interval width (\% of true
reserve) by $c$, $I$ and $J$ ($M=500$, $B=500$). Coverage improves in $I$ at
every $(c,J)$ and is close to flat in $c$; the width panel shows why the
flatness is not informative on its own. The $J = 10$, $I = 7$ column is empty
because with $I < J$ the development years $j \geq I$ carry no observed cell
and $\hat\bpi$ cannot be estimated.}
\label{tab:sensitivity}
\small
\begin{tabular}{@{}lcccccc@{}}
\toprule
& \multicolumn{3}{c}{$J = 5$}
& \multicolumn{3}{c}{$J = 10$} \\
\cmidrule(lr){2-4}\cmidrule(lr){5-7}
$c$ & $I=7$ & $I=10$ & $I=15$
    & $I=7$ & $I=10$ & $I=15$ \\
\midrule
\multicolumn{7}{@{}l}{\emph{Coverage (\%)}}\\
 10 & 85 & 90 & 93 & --- & 84 & 91 \\
 20 & 81 & 88 & 92 & --- & 84 & 89 \\
 30 & 83 & 87 & 93 & --- & 79 & 91 \\
 50 & 83 & 87 & 92 & --- & 82 & 89 \\
100 & 81 & 85 & 93 & --- & 81 & 87 \\
200 & 81 & 89 & 91 & --- & 83 & 89 \\
\midrule
\multicolumn{7}{@{}l}{\emph{Rel.\ width (\%)}}\\
 10 & 300 & 248 & 236 & --- & 194 & 188 \\
 20 & 131 & 135 & 136 & --- & 110 & 110 \\
 30 & 100 & 103 & 103 & --- &  81 &  84 \\
 50 &  73 &  77 &  75 & --- &  60 &  64 \\
100 &  48 &  50 &  51 & --- &  42 &  43 \\
200 &  33 &  35 &  36 & --- &  29 &  30 \\
\bottomrule
\end{tabular}
\end{table}

Coverage improves in $I$ throughout, and does so faster than the
$O(I^{-1/2})$ rate of Theorem~\ref{thm:coverage_rate} would suggest: the
oracle study of Section~\ref{sec:oracle} measures the deficit decaying at a
log-log slope of $-1.28$, because it is governed by the $O(I^{-1})$ bias of
$\hat{c}$ rather than by its $O(I^{-1/2})$ spread (Section~\ref{sec:coverage}).

Coverage is close to flat in $c$, at $79$--$93\%$ across the grid. That
flatness is a consequence of the scale-invariance established in
Remark~\ref{rem:sigma_c_finite}: both the relative bias and the coefficient
of variation of $\hat{c}$ depend on $I$ and not on $c$, so the plug-in error
enters the Beta predictive in the same relative amount at every
concentration. It should not be read as insensitivity to $c$ in any practical
sense: across the same grid the interval width varies by a factor of nine,
from $33\%$ of the true reserve at $c = 200$ to $300\%$ at $c = 10$. Coverage is hence held constant by the width rather than by the accuracy of
the estimator.

Coverage at $J = 10$ runs $1$--$8$ points below $J = 5$ at the same $(c,I)$,
about $4$ points on average. Three explanations are available and the design
of Section~\ref{sec:sim_design} separates them.

The first is numerical, and is excluded by the guard-tightening check of
Section~\ref{sec:algorithm}.

The remaining two are tail thickness and column count, which the standard
$J = 5$ and $J = 10$ designs vary together. Table~\ref{tab:jdecomp} separates
them using the thin-tailed $J = 5$ pattern.

\begin{table}[ht]\centering
\caption{Decomposition of the $J$ effect, coverage (\%) at $c = 50$,
$M = 500$. ``Tail'' is fat $\to$ thin at $J = 5$; ``$J$'' is thin $J=5 \to$
thin $J=10$; the two sum to the total. The standard error of each difference
is about $2$ points.}
\label{tab:jdecomp}
\begin{tabular}{@{}lcccccc@{}}
\toprule
$I$ & fat $J{=}5$ & thin $J{=}5$ & thin $J{=}10$
    & tail effect & $J$ effect & total \\
\midrule
 7 & 83.0 & 80.8 & ---  & $-2.2$ & ---    & ---    \\
10 & 86.6 & 89.2 & 82.4 & $+2.6$ & $-6.8$ & $-4.2$ \\
15 & 92.4 & 91.4 & 89.2 & $-1.0$ & $-2.2$ & $-3.2$ \\
\bottomrule
\end{tabular}
\end{table}

The tail-thickness effect is null: $+2.6$ and $-1.0$ points at $I = 10$ and
$15$, both inside Monte-Carlo error and of opposite sign. The $J$ effect
carries the whole difference, and at $I = 10$ it is $-6.8$ points, three
standard errors from zero. Coverage degrades at $J = 10$ because there are more
development proportions to estimate from a fixed number of rows, not because
the late columns carry little mass. This is the mechanism to which
Theorem~\ref{thm:coverage_rate} attributes the $J$-dependence of its constant,
and it is the reason that constant depends on $J$ at all. It is also the
component that Section~\ref{sec:conditional_coverage} finds responsible for
eleven of the twenty-six coverage points of allocation-conditional
miscalibration: $\hat\bpi$ is the parameter this construction substitutes and
does not integrate over.

\subsection{Comparison with the ODP bootstrap}\label{sec:sim_odp}

Table~\ref{tab:odp_comparison} compares the multinomial and the ODP bootstrap
over the five data-generating processes of Section~\ref{sec:sim_design}.

\begin{table}[ht]
\centering
\caption{Multinomial vs ODP bootstrap at $I = J = 10$, $M = 500$,
$B = 1{,}000$, on the geometric development pattern of
Section~\ref{sec:sim_design}. Bias is each procedure's own point-estimate
bias (predictive median for the multinomial, bootstrap mean for ODP), on the
ratio-of-means convention.}
\label{tab:odp_comparison}
\begin{tabular}{@{}llccccc@{}}
\toprule
DGP & Method & Cov.\ 95\% & Cov.\ 75\%
    & Bias (\%) & Rel.\ Width (\%) & $\hat{c}$ \\
\midrule
Dirichlet--Gamma
  & Multinomial & 79.4 & 57.0 &  1.4 &  60.0 & 69.6 \\
  & ODP         & 92.8 & 73.4 &  1.9 &  76.3 & --- \\
\midrule
Non-stat.\ ($\sigma_\delta = 0.05$)
  & Multinomial & 70.0 & 44.6 &  6.9 &  79.6 & 36.9 \\
  & ODP         & 92.0 & 70.8 & $-0.5$ & 112.1 & --- \\
\midrule
Tweedie ($p = 1.3$)
  & Multinomial & 88.6 & 68.4 &  0.2 &  18.2 & 722.9 \\
  & ODP         & 93.8 & 71.6 & $-0.0$ &  20.5 & --- \\
\midrule
Tweedie ($p = 1.5$)
  & Multinomial & 86.4 & 66.6 &  2.4 &  52.6 & 100.6 \\
  & ODP         & 93.0 & 70.0 &  1.5 &  51.8 & --- \\
\midrule
Tweedie ($p = 1.8$)
  & Multinomial & 91.0 & 68.4 & 28.2 & 883.7 & 7.5 \\
  & ODP         & 86.2 & 62.6 & 15.8 & 205.5 & --- \\
\bottomrule
\end{tabular}
\end{table}

Two readings of Table~\ref{tab:odp_comparison} must be kept separate. As a
procedure for unconditional coverage, the ODP bootstrap performs well: in four
of the five data-generating processes it attains coverage closer to the
nominal $95\%$ than the conditional multinomial bootstrap, and the ordering
was the same in every row of an earlier replication of the whole grid at a
second seed, run before the subcomposition correction of
Section~\ref{sec:c_estimation}. The two to three points by which ODP falls short of nominal on the
Dirichlet--Gamma row are the variance-function mismatch between a
Dirichlet allocation and the linear ODP variance function, isolated in
Table~\ref{tab:frailty_amounts}, and not a targeting error.
The exception is the heaviest-tailed Tweedie DGP, where $\hat{c} = 7.5$ places
the multinomial in the boundary regime of Section~\ref{sec:sim_tweedie}
and its nominal coverage is produced by an interval nine times the true
reserve; on that row neither procedure should be described as calibrated. Only the ordering is reported here, since the gaps moved substantially between
the two runs, by up to $6$ points on the Tweedie and non-stationary rows,
and only the direction is stable at $M = 500$.

By Proposition~\ref{prop:boot_decomp} the two procedures target the same
conditional object under diffuse anchoring, and the gap in the table is
the plug-in cost of the multinomial construction: the residual bootstrap
integrates over the row level and the development pattern by resampling,
while Algorithm~\ref{alg:bootstrap} substitutes $\hat{c}$ and $\hat\bpi$.
Section~\ref{sec:oracle} makes that attribution exact, and
Section~\ref{sec:conditional_coverage} shows that integrating over
$(c, \bpi)$ closes most of it. Note that the identification result for
the concentration and its closed-form variance, the operability bound and
the heterogeneity screen of Section~\ref{sec:c_threshold}, and the
applicability to Bornhuetter--Ferguson and Cape Cod have no ODP
counterpart, and are orthogonal to this coverage comparison.

\subsection{Frailty: the conjecture and its rejection}\label{sec:sim_frailty}

The conjecture that motivated this paper was that the ODP bootstrap
omits the posterior frailty variance of the least-developed accident
years, so that under a DGP with row frailty $U_i \sim \Gam(\kappa,\kappa)$
its coverage should fall as $\kappa$ falls while its width stays flat. The
count hierarchy of Section~\ref{sec:counts} is the setting in which that
conjecture can be tested exactly, because the component it says is omitted
is available in closed form on every simulated triangle:
$U_i \mid \mathcal{D}_i \sim \Gam(\kappa + \Nobs_i,\, \kappa + \mu_i F_{I-i})$
and $\E[R_i \mid \mathcal{D}]$ is the mean of Theorem~\ref{thm:ibnr}, so
$\sum_i \E[R_i \mid \mathcal{D}]^2 \Var(U_i \mid \mathcal{D}_i)$ can be
set against the ODP bootstrap variance on the same triangle. If the
component were omitted, the coverage under (N1)--(N3) with a normal limit
would be $2\Phi(\rho z_{1-\alpha/2}) - 1$ with
$\rho^2 = 1/(1 + r)$ and $r$ the ratio of the two variances, and that
prediction is reported next to the measured coverage.

\begin{table}[ht]\centering
\caption{ODP bootstrap under accident-year frailty in the count hierarchy of
Section~\ref{sec:counts} ($I = J = 10$, geometric pattern, $\mu_i = 2{,}000$,
$M = 500$, $B = 1{,}000$; between $445$ and $500$ replications per row are
retained, the remainder having a zero early cell on which the ODP fit is
undefined). ``Omitted/ODP'' is the conjectured omitted component
$\sum_i \E[R_i \mid \mathcal{D}]^2 \Var(U_i \mid \mathcal{D}_i)$, exact
under this DGP, as a ratio to the ODP bootstrap variance on the same
triangle; ``predicted'' is the coverage that ratio would imply if the
component were omitted, under (N1)--(N3) with a normal limit, i.e.\
$2\Phi(\rho z_{1-\alpha/2}) - 1$ with $\rho^2 = 1/(1 + \text{Omitted/ODP})$.
$\kappa = \infty$ is the frailty-free control.}
\label{tab:frailty}
\begin{tabular}{@{}lccccccc@{}}
\toprule
$\kappa$ & $\Var(U_i)$ & Cov.\ 95\% & Cov.\ 75\% & Rel.\ Width (\%)
  & Omitted/ODP & Pred.\ 95\% & Pred.\ 75\% \\
\midrule
$\infty$ & 0.00 & 93.6 & 72.4 & 12.4 & 0.00 & 95.0 & 75.0 \\
 10 & 0.10 & 93.8 & 72.1 & 12.7 & 0.35 & 90.9 & 67.9 \\
  4 & 0.25 & 95.6 & 77.7 & 13.1 & 0.41 & 90.2 & 66.8 \\
  2 & 0.50 & 95.1 & 77.2 & 13.6 & 0.53 & 88.7 & 64.8 \\
  1 & 1.00 & 96.4 & 75.5 & 15.0 & 0.73 & 86.4 & 61.8 \\
\bottomrule
\end{tabular}
\end{table}

The conjecture is rejected. The conjectured omission would take $95\%$
coverage down to $86.4\%$ and $75\%$ coverage to $61.8\%$
at $\Var(U_i) = 1$, where the component is $0.73$ times the
bootstrap variance. The measured coverage is $96.4\%$ and
$75.5\%$, against $93.6\%$ at the frailty-free control,
and the relative width rises from $12\%$ to $15\%$
over the sweep. The interval loses nothing as the frailty variance grows.

Proposition~\ref{prop:boot_decomp} says why. The component is not omitted;
it is the posterior level variance that resampling the diagonal
generates, and under free accident-year effects the bootstrap generates it
at its diffuse-prior value, $\Var(\lambda_i \mid \Nobs_i) = \Nobs_i/F_{I-i}^2$,
which is the $\kappa \to 0$ member of Theorem~\ref{thm:ibnr}. At finite
$\kappa$ the true conditional predictive is narrower, so the bootstrap
over-covers slightly rather than under-covers, which is the direction of
the drift in the table. Adding $\E[R_i \mid \mathcal{D}]^2\Var(U_i \mid \mathcal{D}_i)$
to the bootstrap variance counts the same term twice. The absorption of
the realised frailty by the free levels, which we had read as the
mechanism of a loss, is the mechanism of a conservative pricing of the
level uncertainty.

The same sweep on the amounts side, with $U_i$ multiplying the Gamma
ultimate of the Dirichlet--Gamma hierarchy, does show ODP coverage
falling, from $91.8\%$ to $85.4\%$ at the $95\%$ level, and that is what
originally appeared to confirm the conjecture. Table~\ref{tab:frailty_amounts}
separates the cause by running the same rows under two allocations.

\begin{table}[ht]\centering
\caption{ODP bootstrap under accident-year frailty on the amounts side,
$I = J = 10$, $S_i \sim \Gam(2E_i, 0.001) \cdot U_i$,
$U_i \sim \Gam(\kappa,\kappa)$, $M = 500$, $B = 1{,}000$, on the rows for
which the ODP fit is defined (between $394$ and $500$ per cell), same rows
under two allocations: Dirichlet with $c = 50$ (cell variance quadratic in
the row level) and compound Poisson at $p = 1$ (cell variance linear in
the mean, the ODP variance function).}
\label{tab:frailty_amounts}
\begin{tabular}{@{}lccccccc@{}}
\toprule
& & \multicolumn{3}{c}{Dirichlet allocation} & \multicolumn{3}{c}{Compound Poisson, $p = 1$} \\
\cmidrule(lr){3-5}\cmidrule(lr){6-8}
$\kappa$ & $\Var(U_i)$ & Cov.\ 95\% & Cov.\ 75\% & Width (\%) & Cov.\ 95\% & Cov.\ 75\% & Width (\%) \\
\midrule
$\infty$ & 0.00 & 92.2 & 70.7 & 77.2 & 93.7 & 74.5 & 17.7 \\
 10 & 0.10 & 93.8 & 69.6 & 75.5 & 94.4 & 76.4 & 17.8 \\
  4 & 0.25 & 90.4 & 66.1 & 76.8 & 93.5 & 72.5 & 18.5 \\
  2 & 0.50 & 87.6 & 64.6 & 79.6 & 94.5 & 72.8 & 19.5 \\
  1 & 1.00 & 82.1 & 63.7 & 92.2 & 96.4 & 76.4 & 21.5 \\
\bottomrule
\end{tabular}
\end{table}

Under the Dirichlet allocation, whose cell variance is quadratic in the
row level, coverage falls by $10.1$ points across the sweep;
under a compound Poisson allocation at $p = 1$ on the same rows, which has
the ODP variance function, it rises from $93.7\%$ to
$96.4\%$, the same slight upward drift as in the count
hierarchy, while the width grows in both arms. The amounts-side decline
is a variance-function effect, the mismatch between a Dirichlet
allocation and a linear variance function growing as the rows become
heterogeneous in size, and not a frailty effect. Note that the Dirichlet
arm is below nominal already at $\Var(U_i) = 0$, at $92.2\%$,
in agreement with the $92.8\%$ of Table~\ref{tab:odp_comparison}: that shortfall
is the same mismatch at zero heterogeneity, and it is the whole of the
under-coverage of the ODP bootstrap under the Dirichlet--Gamma DGP
reported in Section~\ref{sec:sim_odp}.

Two consequences follow. First, the systematic under-coverage of the ODP
bootstrap on paid Schedule~P data that \citet{Meyers2015} documents is
not a frailty effect. That is consistent with his own diagnosis, which is
one of location: both the ODP and the Mack model overstate the expected
ultimate on those triangles, and the model that repairs the failure
corrects the mean through a changing settlement rate rather than the
dispersion. Second, the remark on $\kappa$ identification stands but its
implication is reversed: a free-level residual bootstrap cannot learn
$\kappa$, and for that reason it prices the level at the diffuse limit,
which is the widest predictive in the family. What it cannot do is narrow
the predictive when $\kappa$ is large, which requires anchoring the
levels; that is the direction of \citet{VanOirbeek2026cnbg}, and the
information ceiling established there bounds how far anchoring can go.

\subsection{The components of the residual bootstrap variance}\label{sec:sim_mech1}

Proposition~\ref{prop:boot_decomp} attributes the excess of the residual
bootstrap variance over the process-only variance to two components,
resampling of the diagonal and refit of the development factors, and
claims that on square triangles their sum is bounded away from one and
free of volume, while at fixed $J$ the refit component is $O(I^{-1})$.
Both claims are testable, and the components can be isolated by
construction. Within a
single bootstrap replication we form the pseudo-triangle from resampled
Pearson residuals in the usual way and then project it more than once: with
the development factors refitted on the pseudo-triangle, with the factors
fitted to the original triangle, and with both the original factors and the
original latest diagonal. Everything else, i.e.\ the residual draw, the pseudo-triangle and the
dispersion estimate, is shared, so that each ratio isolates one component
only. The data-generating process is exact
overdispersed Poisson, $X_{ij} = \phi\,\mathrm{Pois}(\mu_{ij}/\phi)$ with
$\phi = 1$.

\begin{table}[ht]\centering
\caption{Measured re-estimation inflation, isolated by construction within
each bootstrap replication ($M = 300$, $B = 500$, no failed replications).
``Factor refit'' varies only the development factors; ``total''
additionally restores the original latest diagonal and is therefore the
full excess of the bootstrap variance over the process-only variance.}
\label{tab:mech1}
\begin{tabular}{@{}lccccc@{}}
\toprule
\multicolumn{6}{@{}l}{\emph{Arm A: square triangle $I = J = 10$, volume sweep}}\\
Portfolio volume & $10^4$ & $10^5$ & $10^6$ & $10^7$ & $10^8$ \\
Factor refit      & 1.655 & 1.652 & 1.659 & 1.652 & 1.641 \\
Diagonal resample & 1.812 & 1.815 & 1.804 & 1.811 & 1.816 \\
Total             & 2.986 & 2.988 & 2.985 & 2.984 & 2.969 \\
\midrule
\multicolumn{6}{@{}l}{\emph{Arm B: fixed $J = 5$, triangle size sweep}}\\
$I$ & 7 & 10 & 15 & 20 & 30 \\
Factor refit & 1.474 & 1.276 & 1.169 & 1.118 & 1.075 \\
Total        & 2.656 & 2.288 & 2.081 & 2.024 & 1.930 \\
\bottomrule
\end{tabular}
\end{table}

The two components multiply: in Arm~A the product of the factor-refit and
diagonal-resampling ratios reproduces the total to within $0.5\%$ in every
cell, and in Arm~B, which reports the factor-refit ratio and the total only,
the implied diagonal-resampling ratio lies between $1.79$ and $1.82$
throughout, in agreement with Arm~A.

Arm~A settles the square-triangle claim. Over four orders of magnitude of
portfolio volume the total varies by $0.6\%$ and the factor-refit component by
$1.1\%$; the regression of $\log(\text{inflation}-1)$ on
$\log(\text{volume})$ has slope $-0.002$. Under its own data-generating process the ODP bootstrap's variance is
therefore very nearly three times the process-only variance on a square
triangle, a factor $1.81$ of which is the diagonal resampling that
Proposition~\ref{prop:boot_decomp}(ii) identifies with the posterior level
variance, and $1.65$ the factor refit that (iii) identifies with the
estimation variance of the pattern. Neither is a miscalibration. An
excess growing in portfolio volume would have reached the order of $10^4$
at the top of this sweep; that it does not confirms that both components
scale with volume as the process term does. This is the regime in which
\citet{EnglerLindskog2024} motivate Mack's estimator through
large-exposure asymptotics in a compound Poisson setting; there, as here,
the estimation term grows with exposure at the same rate as the process
term rather than becoming negligible against it, and the volume
invariance of Table~\ref{tab:mech1} is the same fact seen from the
bootstrap side.

Arm~B confirms the fixed-$J$ claim. The factor-refit inflation falls from
$1.474$ at $I = 7$ to $1.075$ at $I = 30$, with $\log(\text{inflation}-1)$
regressing on $\log I$ at slope $-1.26$ against the $-1$ the proposition implies.
The total falls more slowly, from $2.656$ to $1.930$, because the
diagonal-resampling component does not decay: it is a property of resampling
the observed data at all, not of estimating factors from few rows.
The pattern component is therefore a small-triangle and square-triangle
phenomenon, while the level component, being the posterior variance of a
quantity estimated from a single row, persists at every $I$ tested, as it
must.

Suppressing the refit also costs coverage, by $5.3$, $2.7$, $3.7$, $1.0$
and $1.3$ points at $I = 7, 10, 15, 20, 30$ on the same replications. With
$M = 300$ the standard error of each difference is about $1.5$ points, so
only the direction and the rough magnitude are established: a few points at
$I \leq 15$ and roughly nothing by $I = 20$--$30$, which is what the decay
of the factor-refit inflation above implies. The plug-in multinomial pays a
larger deficit against its oracle ($13.7$ to $2.1$ points,
Table~\ref{tab:oracle}), but that deficit prices the substitution of $c$ as
well as $\bpi$ and is not comparable term by term.

\subsection{Isolating the plug-in deficit: oracle coverage}\label{sec:oracle}

The coverage comparison of Table~\ref{tab:odp_comparison} conflates two
effects: the target a bootstrap aims at, and the error incurred in estimating
its parameters. To separate them we re-run the multinomial bootstrap with the
true $(c, \bpi)$ substituted for its estimates, removing all estimation error,
and compare it with the ODP bootstrap in its standard form. All three arms are
evaluated on the same simulated triangles in each replication.

\begin{table}[ht]\centering
\caption{Coverage with estimated vs.\ true parameters, all three procedures on
the same simulated triangles (Dirichlet--Gamma, $J=5$, $c=50$, $M=2{,}000$,
$B=1{,}000$). No replication was dropped.}
\label{tab:oracle}
\begin{tabular}{@{}lccccc@{}}
\toprule
$I$ & 7 & 10 & 15 & 20 & 30 \\
\midrule
Multinomial, estimated $(\hat c, \hat\bpi)$ & 81.0 & 88.5 & 91.7 & 92.0 & 93.6 \\
Multinomial, true $(c, \bpi)$ (oracle)      & 94.7 & 94.9 & 95.6 & 94.7 & 95.7 \\
ODP, estimated (standard)                   & 92.7 & 93.3 & 93.9 & 93.6 & 94.3 \\
\midrule
Plug-in deficit (points)         & 13.7 & 6.4 & 3.9 & 2.7 & 2.1 \\
Width, ODP / conditional oracle  & 1.11 & 1.03 & 0.99 & 0.96 & 0.95 \\
$\E[\hat c]$ (true $c = 50$)     & 92.6 & 65.6 & 57.3 & 55.4 & 53.3 \\
\bottomrule
\end{tabular}
\end{table}

Three conclusions follow. First, the conditional predictive is correctly
specified in its parameters: given the true $(c, \bpi)$ the construction
covers at $94.7$--$95.7\%$ across the whole range, with no trend in $I$.
Second, the deficit in the estimated case is plug-in error in those
parameters. It falls from $13.7$ points at $I = 7$ to $2.1$ at $I = 30$, a
log-log slope of $-1.28$ against $I$; Theorem~\ref{thm:coverage_rate} gives
$O(I^{-1/2})$, a valid upper bound that describes the standard deviation of
$\hat{c}$, whereas the deficit tracks its $O(I^{-1})$ bias
($\E[\hat c]/c$ falls from $1.85$ to $1.07$ over the same range;
Section~\ref{sec:coverage}). Third, ODP is close to nominal throughout at
$92.7$--$94.3\%$, and its interval is $11\%$ wider than the correctly
specified conditional target at $I = 7$, converging to $5\%$ \emph{narrower}
by $I = 30$.

The practical reading is that the plug-in construction is unusable at
$I = 7$, costly at $I = 10$ and essentially calibrated by $I = 20$--$30$, and
that the regularised construction of Section~\ref{sec:regularised} is the one
to use in the range where macro-level triangles actually sit. Note that this
attribution concerns \emph{unconditional} coverage; the oracle arm is
calibrated on average and, as the next subsection shows, is not conditionally
calibrated.

\subsection{Conditional calibration}\label{sec:conditional_coverage}

Every coverage figure reported so far is unconditional: a proportion over
simulated triangles. The argument of this paper is about a conditional object,
and the two notions differ exactly when a procedure's error depends on the
realised triangle. We therefore partition the replications by a statistic of
$\mathcal{D}$ and report coverage within each cell. A procedure targeting
$p(R \mid \mathcal{D})$ should be calibrated within cells; one calibrated only
on average may be systematically wrong inside them.

Two partitions are used, both into terciles with cut-points fixed in advance
by a pilot run under the same data-generating process. $S_1 = \hat{c}/c$ is
the statistic through which the plug-in construction depends on the triangle,
and is therefore the partition on which requirement~(ii) of
Principle~\ref{prin:bootstrap} makes a prediction. $S_2$ is the coefficient of
variation of the Chain-Ladder row ultimates, a function of $\mathcal{D}$ that
no procedure uses as a parameter. Five procedures run on each simulated
triangle: the plug-in bootstrap of Algorithm~\ref{alg:bootstrap}; the
regularised bootstrap of Algorithm~\ref{alg:regularised}, integrating over the
posterior of $c$; a variant integrating over the joint posterior of
$(c,\bpi)$; the ODP bootstrap; and an oracle supplied with the true $c$ and
true $\bpi$, which carries no estimation error of any kind.

\begin{table}[ht]\centering
\caption{Coverage (\%) within terciles of a statistic of $\mathcal{D}$
(Dirichlet--Gamma DGP, $J = 5$, $c = 50$, $M = 2{,}000$). ``Gradient'' is the
range across cells; with roughly $660$ replications per cell the standard
error of a cell coverage is about $1.3$ points, so a gradient below $4$ is not
distinguishable from flat. ``Width'' is the $95\%$ interval as a percentage of
the true reserve.}
\label{tab:conditional_coverage}
\small
\begin{tabular}{@{}llccccc@{}}
\toprule
& Procedure & cell 1 & cell 2 & cell 3 & Gradient & Width \\
\midrule
\multicolumn{7}{@{}l}{\emph{$I = 10$, partition on $S_1 = \hat{c}/c$}}\\
& Plug-in                  & 96.5 & 88.3 & 78.3 & $18.3$ & 76.7 \\
& Regularised in $c$       & 92.4 & 88.0 & 87.2 & $\phantom{0}5.3$ & 77.0 \\
& Regularised in $c,\bpi$  & 96.2 & 94.2 & 93.4 & $\phantom{0}2.8$ & 96.0 \\
& ODP                      & 96.2 & 93.0 & 92.4 & $\phantom{0}3.8$ & 81.2 \\
& Oracle (true $c,\bpi$)   & 95.5 & 94.7 & 96.4 & $\phantom{0}1.7$ & 79.3 \\
\midrule
\multicolumn{7}{@{}l}{\emph{$I = 10$, partition on $S_2$ (CV of ultimates)}}\\
& Plug-in                  & 99.2 & 93.3 & 73.6 & $25.6$ & 76.7 \\
& Regularised in $c$       & 99.4 & 95.7 & 75.1 & $24.3$ & 77.0 \\
& Regularised in $c,\bpi$  & 99.8 & 98.7 & 86.7 & $13.1$ & 96.0 \\
& ODP                      & 100.0 & 99.3 & 84.3 & $15.7$ & 81.2 \\
& Oracle (true $c,\bpi$)   & 99.7 & 98.7 & 89.3 & $10.4$ & 79.3 \\
\bottomrule
\end{tabular}
\end{table}

On the partition by $S_1 = \hat{c}/c$, the plug-in construction falls from
$96.5\%$ on triangles yielding a low
$\hat{c}$ to $78.3\%$ on those yielding a high one, a gradient of $18.3$
points. The direction is the one the mechanism predicts: an overstated
$\hat{c}$ concentrates the Beta predictive and narrows the interval. Replacing
the point estimate by the posterior of $c$ removes most of it ($5.3$ points),
and integrating over $(c,\bpi)$ jointly removes almost all ($2.8$). The
oracle, carrying no estimation error, is flat at $1.7$ points. Since the
$\hat{c}$-conditional miscalibration is entirely estimation error and
integration removes $85\%$ of it, requirement~(ii) has measurable content. The widths of the plug-in,
$c$-regularised and oracle arms agree to within one percentage point, the improvement hereby being a
redistribution of coverage across triangles rather than a uniform widening.

The ODP row on $S_1$ is comparatively flat, at $3.8$ points, but this carries
no information: ODP does not use $\hat{c}$, so its coverage is close to
independent of the partitioning variable by construction. It is reported for completeness only.

Every procedure shows a large gradient on the second partition, and
integrating over $c$ alone does not help. The oracle is the informative row:
with the true $c$ and the true $\bpi$ it still varies by $10.4$ points. The $S_2$ miscalibration is therefore not, or not only, estimation error. The plug-in's $25.6$ points attribute \emph{sequentially}: each row gives the
increment along the path plug-in $\to$ $c$-regularised $\to$ joint $\to$
oracle of Table~\ref{tab:conditional_coverage}, so that the split is
path-dependent and not a partition of separable causes:
\begin{center}
\begin{tabular}{@{}lc@{}}
\toprule
Increment along the path & Points \\
\midrule
Removed by integrating over $c$ (plug-in $\to$ $c$-reg.) & $\phantom{0}1.3$ \\
Further removed by integrating over $\bpi$ ($c$-reg.\ $\to$ joint) & $11.2$ \\
Joint arm to oracle (residual estimation and MC error) & $\phantom{0}2.7$ \\
Oracle floor: diffuse-anchoring approximation & $10.4$ \\
\midrule
Total & $25.6$ \\
\bottomrule
\end{tabular}
\end{center}
The first two are the same failure of requirement~(ii) identified on $S_1$,
now in the development pattern rather than the concentration, and they are
removable in principle by integrating over $\bpi$; the joint arm does so and
recovers $11$ of the $25.6$ points, at a $25\%$ cost in interval width. The
third is not estimation error at all, and is the subject of the next
paragraph.

By Remark~\ref{rem:diffuse_scope} the Beta representation of
Theorem~\ref{thm:ibnp} is exact only in the limit $\alpha, \beta \to 0$; under
an informative prior on the ultimate the conditional law of $\Wobs_i$ given
$\Xobs_i$ is tilted, and the tilt depends on $\Xobs_i$. Since $S_2$ orders
triangles by the dispersion of the estimated ultimates, a tilt of this kind
would appear as exactly the residual gradient the oracle displays. The
prediction is testable by varying the dispersion of $S_i$ while holding its
mean fixed.

\begin{table}[ht]\centering
\caption{Coverage gradient on $S_2$ as the ultimate approaches the diffuse
limit. $S_i \sim \Gam(a, a/2\!\times\!10^6)$, so the mean is held fixed and
only the dispersion varies. Bin edges re-derived by pilot at each shape, since
$S_2$'s scale moves with the dispersion. $I = 10$, $J = 5$, $c = 50$,
$M = 2{,}000$.}
\label{tab:diffuse_scope}
\begin{tabular}{@{}lcccc@{}}
\toprule
Shape $a$ & $\CV(S_i)$ & Plug-in & Joint $(c,\bpi)$ & Oracle \\
\midrule
$2{,}000$ & $2.2\%$  & $25.6$ & $13.1$ & $10.4$ \\
$50$      & $14\%$   & $10.5$ & $\phantom{0}6.2$ & $\phantom{0}5.7$ \\
$4$       & $50\%$   & $\phantom{0}1.8$ & $\phantom{0}0.5$ & $\phantom{0}1.4$ \\
\bottomrule
\end{tabular}
\end{table}

The oracle gradient falls from $10.4$ points to $1.4$, which is inside
Monte-Carlo error, as the ultimate becomes diffuse, confirming the mechanism. Note that
the plug-in gradient collapses with it, from $25.6$ to $1.8$: when the Beta
representation is exact, the estimation error in $(\hat{c}, \hat\bpi)$ has
nothing to interact with and the construction is conditionally calibrated on
$S_2$ even with both parameters substituted. This is why the attribution above is stated as a path
decomposition: the estimation and anchoring components multiply rather than
add, in the sense that the estimation gradient collapses together with the
anchoring approximation once the latter is exact, so that the $11.2$ points
credited to $\bpi$ are its increment along one path and not a standalone
share.

Requirement~(ii) of Principle~\ref{prin:bootstrap} is implementable and worth
implementing: integration removes $85\%$ of the $\hat{c}$-conditional
miscalibration at unchanged width, and a further $11$ points of the $S_2$
miscalibration at a $25\%$ width cost. But conditional calibration in general
is not achieved, and the obstruction is not estimation. Under
Chain-Ladder-type anchoring with a concentrated ultimate, which is the regime
in which the method is normally applied, the construction is conditionally
miscalibrated by roughly ten coverage points across terciles of realised
allocation dispersion, for a reason internal to the model rather than to the
data: the closed-form conditional law is exact only in a limit the anchoring
does not occupy. Removing it requires the tilted density of Remark~\ref{rem:diffuse_scope} rather than its
Beta approximation, and that density does not admit the closed form on which
the anchor-agnostic property of Section~\ref{sec:bootstrap} depends.

\subsection{Anchoring regimes}\label{sec:sim_anchor}

The construction of Section~\ref{sec:bootstrap} depends on the reserving
method only through the cumulative development proportions, and is in that
sense anchor-agnostic. Every coverage figure reported so far, however, is
measured under Chain-Ladder anchoring, the diffuse limit in which the Beta
representation of Theorem~\ref{thm:ibnp} is exact
(Remark~\ref{rem:diffuse_scope}). Under an informative anchor the Beta draw is
a plug-in approximation to a tilted density, and the predictive is a different
estimator: $E_i q(1 - W_i)$ rather than $\Xobs_i(1-W_i)/W_i$. Whether it remains calibrated is a separate question from whether it
applies.

Three anchors are run on the same simulated triangle in each replication: the
diffuse Chain-Ladder anchor; the Bornhuetter--Ferguson anchor $E_i q$ with the
prior loss ratio supplied at the true value and at $\pm 20\%$; and the Cape
Cod anchor, with $\hat q = \sum_i \Xobs_i / \sum_i E_i F_{I-i}$ estimated from
the triangle. The data-generating process carries a known exposure and a known
true loss ratio, so ``correctly specified'' is well defined and the
misspecified arms are misspecified by a stated amount.

\begin{table}[ht]\centering
\caption{Coverage (\%) and interval width (\% of the true reserve) by
anchoring regime. Dirichlet--Gamma DGP, $J = 5$, $c = 50$, $M = 500$,
$B = 1{,}000$; within each row all three anchors are evaluated on the same
triangles, while prior rows are independent replications. The CL and CC
anchors do not depend on the prior, so their spread across rows within each
$I$ (at most $1.4$ and $3.2$ points) reads directly as the Monte-Carlo scale.
True loss ratio $2{,}000$; Cape Cod recovers $\hat q = 1{,}998$--$2{,}000$.}
\label{tab:anchor_coverage}
\begin{tabular}{@{}llcccccc@{}}
\toprule
& & \multicolumn{3}{c}{Coverage} & \multicolumn{3}{c}{Width} \\
\cmidrule(lr){3-5}\cmidrule(lr){6-8}
$I$ & Prior $q$ & CL & BF & CC & CL & BF & CC \\
\midrule
10 & $0.8\,q^{\ast}$ & 87.0 & 37.4 & 82.2 & 75.6 & 31.6 & 39.6 \\
10 & $q^{\ast}$      & 88.4 & 86.2 & 82.2 & 76.0 & 40.2 & 40.3 \\
10 & $1.2\,q^{\ast}$ & 88.0 & 59.8 & 79.0 & 76.1 & 47.8 & 39.8 \\
\midrule
15 & $0.8\,q^{\ast}$ & 90.4 & 37.2 & 86.4 & 77.6 & 32.8 & 41.0 \\
15 & $q^{\ast}$      & 90.2 & 89.0 & 86.8 & 76.2 & 40.5 & 40.6 \\
15 & $1.2\,q^{\ast}$ & 90.2 & 63.8 & 88.4 & 76.7 & 48.6 & 40.5 \\
\bottomrule
\end{tabular}
\end{table}

At $q = q^{\ast}$ the Bornhuetter--Ferguson anchor covers at $86.2\%$ against
the diffuse anchor's $88.4\%$ at $I = 10$, and $89.0\%$ against $90.2\%$ at
$I = 15$, both within Monte-Carlo error at $M = 500$, on an interval $53\%$
of the width. The relative bias of the predictive median is under one
percentage point for both. This is the credibility structure operating as
intended: the same triangle combined with a correct informative prior yields a
materially sharper predictive at no cost in coverage.

A $20\%$ error in the prior loss ratio reduces coverage to $37.4\%$ when the
prior is too low and $59.8\%$ when it is too high, at $I = 10$, with the
predictive median displaced by $-20.5\%$ and $+19.1\%$ respectively. The
asymmetry follows from the widths: an understated prior both displaces the
estimate downward and narrows the interval, so the two errors compound.
Increasing $I$ does not help, coverage at $I = 15$ being $37.2\%$ and
$63.8\%$ respectively, since the miscalibration is driven by the anchor and
not by estimation. Note that no amount of data corrects a prior that is
asserted rather than estimated. An informative anchor hence purchases
precision at the price of the correctness of that anchor, quantified in
Table~\ref{tab:anchor_coverage}.

Cape Cod recovers the loss ratio almost exactly ($\hat q = 1{,}998$--$2{,}000$
against $2{,}000$) and its point estimate is unbiased to within $1.2\%$, yet
it covers $2$--$9$ points below the diffuse anchor across the six rows, and
$6.2$ and $3.4$ points below it on the correctly specified rows, at the same
interval width as a correctly specified Bornhuetter--Ferguson. The deficit is not bias; it is
the sampling variability of $\hat q$, estimated from the same triangle that is
being held fixed and whose uncertainty the Beta draw does not represent. This
is requirement~(ii) of Principle~\ref{prin:bootstrap} again, in a third
parameter: the construction integrates over $c$ where the regularised variant
is used, substitutes $\hat\bpi$, and substitutes $\hat q$ under Cape Cod. The
deficit narrows with $I$ ($6.2$ points at $I = 10$, $3.4$ at $I = 15$),
consistent with an estimation-error explanation.

\section{Empirical illustration}\label{sec:empirical}

Count-side reserving on real data is illustrated in
\citet{VanOirbeek2026nbcl}, whose reported $\hat{\kappa}$ does not carry over
to hierarchy~\eqref{eq:frailty} (Remark~\ref{rem:kappa_identification}). The
present paper complements that illustration by demonstrating the amount-side
framework on three benchmark paid loss triangles from the
\texttt{ChainLadder} R package \citep{ChainLadder2023}: Taylor--Ashe
\citep{taylor1983}, RAA \citep{mack1994} and Mortgage, for which the
allocation parameter $\hat{c}$ rather than the frailty parameter
$\hat{\kappa}$ is the diagnostic of interest.

\subsection{Reserve estimates}\label{sec:empirical_reserves}

\begin{table}[ht]
\centering
\begin{threeparttable}
\caption{Reserve estimates for three benchmark triangles. Multinomial
bootstrap with $B=5{,}000$, seed 2026. For RAA and Mortgage the multinomial
rows cover only the accident years meeting the moment-reporting tier
$\hat{c}F_{I-i} \geq 5$ of Remark~\ref{rem:validity}; the like-for-like
Chain-Ladder figure for those years is given as ``CL (reported)''.}
\label{tab:empirical_all}
\small
\begin{tabular}{@{}llrrrr@{}}
\toprule
Dataset ($\hat{c}$) & Method & Mean & Std.\ Err & 75th \%ile & 95th \%ile \\
\midrule
Taylor--Ashe (189)
  & CL (all 9 years)  & 18{,}681 & --- & --- & --- \\
  & Multinomial (9)   & 19{,}234 & 1{,}914 & 20{,}292 & 22{,}639 \\
  & Mack (all 9)      & 18{,}681 & 2{,}447 & --- & --- \\
\midrule
RAA (24)\tnote{$\dagger$}
  & CL (all 9 years)  & 52{,}135 & --- & --- & --- \\
  & CL (8 reported)   & 35{,}796 & --- & --- & --- \\
  & Multinomial (8)   & 39{,}288 & 8{,}945 & 44{,}283 & 55{,}116 \\
  & Mack (all 9)      & 52{,}135 & 26{,}909 & --- & --- \\
\midrule
Mortgage (152)\tnote{$\ddagger$}
  & CL (all 8 years)  & 14{,}547 & --- & --- & --- \\
  & CL (7 reported)   & 12{,}900 & --- & --- & --- \\
  & Multinomial (7)   & 13{,}260 & 1{,}389 & 14{,}138 & 15{,}768 \\
  & Mack (all 8)      & 14{,}547 & 3{,}729 & --- & --- \\
\bottomrule
\end{tabular}
\begin{tablenotes}\small
\item[$\dagger$] RAA: $\hat{c} = 24.3$ against an operability bound
$c^\dagger(5) = 44.6$ (Section~\ref{sec:c_threshold}). Its most recent
accident
year has $F_{I-i} = 0.112$, hence $\hat{c}F_{I-i} = 2.7$: the predictive
variance exists but the fourth moment does not, so the year falls in the third
tier of Remark~\ref{rem:validity} and is excluded from the moment summaries.
That single year carries $16{,}339$ of the $52{,}135$ Chain-Ladder reserve,
i.e.\ $31\%$ of the total. The multinomial row is therefore not comparable
with the full-triangle Chain-Ladder and Mack rows; the like-for-like
comparison is against CL (8 reported).
\item[$\ddagger$] Mortgage: $\hat{c} = 152.0$ against
$c^\dagger(5) = 632.5$. Its excluded year has $F_{I-i} = 0.008$ and hence
$\hat{c}F_{I-i} = 1.2$, in the lowest tier, where the predictive mean barely
exists and the variance does not. It carries $1{,}647$ of $14{,}547$, or
$11\%$.
\end{tablenotes}
\end{threeparttable}
\end{table}

On the years it reports, the bootstrap confirms Theorem~\ref{thm:ibnp}(iii)
on all three portfolios. The predictive mean exceeds the Chain-Ladder estimate
for the same accident years by the accident-year-weighted aggregate of the
factor $\hat{c}F_{I-i}/(\hat{c}F_{I-i}-1)$: by $3.0\%$ on Taylor--Ashe
($19{,}234$ against $18{,}681$, $\hat{c} = 189.1$), $2.8\%$ on Mortgage
($13{,}260$ against $12{,}900$, $\hat{c} = 152.0$) and $9.8\%$ on RAA
($39{,}288$ against $35{,}796$, $\hat{c} = 24.3$). The correction is largest
exactly where $\hat{c}$ is smallest, as the theorem requires, and vanishes as
$c \to \infty$, which is the sense in which Chain-Ladder is the
concentrated-allocation limit of the hierarchy. The reported figure is an
aggregate over accident years with different $F_{I-i}$ and is not the per-year
factor of the theorem.

Read against a single constant, the concentration estimates would order the
portfolios by development stability. Read against the operability bound
$c^\dagger(\tau) = \tau/\pi_0$ of Section~\ref{sec:c_threshold}, which is the
quantity that actually governs what the framework can report, the ordering
changes:
\begin{center}
\begin{tabular}{@{}lccccl@{}}
\toprule
Portfolio & $\hat{c}$ & $\pi_0$ & $c^\dagger(2)$ & $c^\dagger(5)$
& Years reported \\
\midrule
Taylor--Ashe & 189.1 & 0.0692 &  28.9 &  72.2 & 9 of 9 \\
RAA          &  24.3 & 0.1121 &  17.8 &  44.6 & 8 of 9 \\
Mortgage     & 152.0 & 0.0079 & 253.0 & 632.5 & 7 of 8 \\
\bottomrule
\end{tabular}
\end{center}
Only Taylor--Ashe clears its own $c^\dagger(5)$, and it is the only portfolio
on which the bootstrap reports every accident year. Mortgage has the second
largest $\hat{c}$ of the three and is nevertheless the most demanding
portfolio, because it reports only $0.79\%$ of the ultimate in its first
development year and therefore requires $\hat{c} \geq 633$ for full moment
reporting on its greenest year. RAA fails its bound by a smaller margin. The
contrast between the two exclusions is not the same mechanical restriction
acting arbitrarily: it is $\pi_0$, and it is computable from the development
pattern before any reserve is calculated.

With $\hat{c} = 24.3$, the most recent accident year of the RAA triangle,
observed to only $F_{I-i} = 0.112$, has $\hat{c}F_{I-i} = 2.7$: below the $\hat{c}F_{I-i} > 4$ at which all four
moments exist, though above the boundary at which the predictive variance
ceases to exist. The bootstrap therefore reports nothing for that year. Note that this is not a
question of under-reserving: on the eight years for which the bootstrap does
report, it over-reserves relative to Chain-Ladder, as the theorem predicts.
The excluded year does, however, carry $31\%$ of the portfolio's reserve. A procedure that remains silent on a third of the
liability should not be used on that triangle, and the comparison
$\hat{c} = 24.3 < c^\dagger(5) = 44.6$ indicates this before any reserve is
computed.

The Mortgage triangle offers a contrast in consequence rather than in kind: it
too falls below its bound, and by a far wider margin ($152.0$ against
$632.5$), but the year it excludes carries $11\%$ of the reserve rather than
$31\%$.

For the well-behaved triangles the multinomial intervals can be compared with
Mack's. On Taylor--Ashe, the only portfolio for which no year is excluded and
the comparison is therefore exact, the multinomial standard error is $1{,}914$
against Mack's $2{,}447$, i.e.\ $22\%$ narrower. This is not an efficiency
gain. The two quantities differ in
what they price: Mack's MSEP includes an estimation-error term, whereas the
plug-in bootstrap conditions on $(\hat{c}, \hat\bpi)$ and prices allocation
uncertainty only. The omitted component is the plug-in deficit of
Theorem~\ref{thm:coverage_rate}, which Section~\ref{sec:oracle} measures at
$6.4$ coverage points at $I = 10$ and attributes to estimation error in
$(\hat{c}, \hat\bpi)$. A narrower interval from a procedure known to
under-cover is a restatement of the deficit, not evidence against it.
Practitioners requiring an interval that includes estimation error should use
the regularised construction of Section~\ref{sec:regularised}, which
integrates over the posterior of $c$ rather than plugging it in;
Table~\ref{tab:empirical_reg} reports it on the same three triangles.

\subsection{The regularised construction on the benchmarks}
\label{sec:empirical_reg}

Section~\ref{sec:regularised} recommends Algorithm~\ref{alg:regularised} over
the plug-in wherever $\hat{c}$ is not well determined, and
Section~\ref{sec:c_threshold} identifies exactly that condition on two of the
three benchmarks. We therefore apply it here. The prior is
$\log c \sim \mathcal{N}(\log 50, 1)$ as in Appendix~\ref{app:bayesian}, the
likelihood is the exact subcomposition likelihood of
Lemma~\ref{lem:dirichlet_sub}, and $\hat\bpi$ is substituted as in the
plug-in: Algorithm~\ref{alg:regularised} integrates over the concentration
and not over the development pattern, so the allocation-conditional
miscalibration of Section~\ref{sec:conditional_coverage} is present in these
figures too.

The reporting rule of Remark~\ref{rem:validity} is a plug-in construct and
does not transfer. Under the posterior mixture the predictive mean of
$\Sibnp_i$ is
$\Xobs_i \int c(1-F_{I-i})/(cF_{I-i}-1)\, d\pi(c \mid \mathcal{D})$, and it
fails to exist twice over: the event $\{cF_{I-i} \leq 1\}$ carries positive
posterior probability under any full-support prior and the conditional mean
is already infinite there, while on its complement the integrand has a
simple pole at $c = 1/F_{I-i}$ at which the posterior density is positive,
so the integral diverges logarithmically. The same argument at
$c = k/F_{I-i}$ removes every higher moment. We therefore report quantiles throughout this subsection, since no moment
exists.

This is also why the accident years excluded from
Table~\ref{tab:empirical_all} reappear below. They are not better determined
under Algorithm~\ref{alg:regularised}; the condition that silenced them,
$\hat{c}F_{I-i} \geq 5$, is a condition for the existence of plug-in moments
and is not the condition governing a quantile of the posterior predictive.

\begin{table}[ht]\centering
\caption{Regularised conditional predictive (Algorithm~\ref{alg:regularised})
on the three benchmarks: percentiles of the total outstanding reserve,
$B = 5{,}000$ posterior draws, prior $\log c \sim \mathcal{N}(\log 50, 1)$,
seed 2026. All accident years are included, including those excluded from
Table~\ref{tab:empirical_all}; units follow Table~\ref{tab:empirical_all}
throughout (thousands for all three portfolios). The posterior mean of $c$
(posterior SD in parentheses) is given alongside the plug-in $\hat{c}$.}
\label{tab:empirical_reg}
\small
\begin{tabular}{@{}lrrrrrrr@{}}
\toprule
& \multicolumn{2}{c}{Concentration} & \multicolumn{5}{c}{Reserve percentile} \\
\cmidrule(lr){2-3}\cmidrule(lr){4-8}
Portfolio & plug-in $\hat{c}$ & post.\ mean
& 5th & 25th & 50th & 75th & 95th \\
\midrule
Taylor--Ashe & 189.1 & 108.7 (22.4)
  & 16{,}093 & 17{,}834 & 19{,}281 & 21{,}166 & 24{,}896 \\
RAA          &  24.3 &  27.3 (5.7)
  & 39{,}700 & 48{,}694 & 57{,}816 & 69{,}851 & 105{,}818 \\
Mortgage     & 152.0 &  65.4 (13.9)
  & 12{,}392 & 14{,}868 & 18{,}197 & 29{,}156 & 476{,}499 \\
\bottomrule
\end{tabular}
\end{table}

The RAA triangle is the case that Section~\ref{sec:c_threshold} identifies in
advance
($\hat{c} = 24.3$ against $c^\dagger(5) = 44.6$) and on which
Table~\ref{tab:empirical_all} is silent for $31\%$ of the Chain-Ladder
reserve. Under the regularised construction that accident year, observed to
$F_{I-i} = 0.112$, has percentiles
\[
7{,}050 \;/\; 11{,}943 \;/\; 18{,}058 \;/\; 28{,}254 \;/\; 63{,}903
\]
at the 5th, 25th, 50th, 75th and 95th, and its posterior median equals
$31.2\%$ of the total's posterior median, against the $31\%$ of the
Chain-Ladder reserve that the plug-in leaves unreported. The interval is wide,
and correctly so: at $\hat{c}F_{I-i} \approx 3$ the allocation of that year
is genuinely poorly determined, and a nine-fold ratio between the 5th and the 95th percentile is the correct
representation of that state of knowledge. The posterior mean of $c$ is $27.3$ against
the plug-in's $24.3$, so on this triangle regularisation moves the
concentration very little; what it changes is the reportability, not the
estimate. One accident year contributes no subcomposition to the likelihood,
its seventh development increment being negative ($-103$, the known salvage
entry of this triangle); the year is included in the predictive, and its fully observed subcomposition through development year six, valid by
Lemma~\ref{lem:dirichlet_sub}, is deliberately not used, since the likelihood
takes one subcomposition per row at its maximal horizon.

On Taylor--Ashe, the posterior mean of $c$ is $108.7$ against a plug-in
$\hat{c} = 189.1$, a
$43\%$ reduction. The natural suspicion is the prior, centred at $50$ below
a likelihood peaked much higher, and the sensitivity grid rules it out:
tripling the prior mean from $30$ to $100$ moves the posterior mean of $c$
only from $106.9$ to $112.2$, a log-linear pooling weight of roughly $4\%$
on the prior, so the posterior sits essentially on the exact subcomposition
likelihood, which peaks near $110$ and not near $189$. The gap is therefore
between the moment estimator and the exact likelihood on the same triangle,
and it is the real-data face of Remark~\ref{rem:sigma_c_finite}: the moment
estimator inverts a noisy sample variance, is upward-biased by roughly a
third at these triangle sizes in simulation, and here, with $n_k$ as small
as three at the deep horizons of a ten-column triangle, the exact likelihood
places the concentration far below it. The predictive is nevertheless insensitive to that disagreement, with a
median total of $19{,}281$ against the plug-in mean of $19{,}234$, since
$\hat{c}F_{I-i} \gg 1$ on every accident year, so the predictive depends on $c$ only weakly over the
contested range.

Mortgage reports $0.79\%$ of its ultimate in the first development year, so
its greenest accident year has $cF_{I-i} \approx 0.5$ at the posterior mean
of $c$. The first shape parameter of the Beta is then below one, the
conditional predictive has no mean at posterior-typical $c$, and the reported quantiles for that year run from $407$ at the 5th
percentile to $465{,}527$ at the 95th, against a Chain-Ladder reserve of
$14{,}547$ for the whole portfolio. The total inherits it: a 95th percentile of
$476{,}499$, roughly thirty-three times the Chain-Ladder reserve.

We report this row and do not recommend it. It is the bottom tier of
Remark~\ref{rem:validity} evaluated rather than avoided, and the operability
bound flags it in advance: $c^\dagger(5) = 632.5$ against $\hat{c} = 152.0$
is the widest margin of the three benchmarks
(Section~\ref{sec:c_threshold}).
Regularisation makes the year \emph{reportable}, in the sense that a
quantile exists and can be computed; it does not make it \emph{informative}.
Note that the two are easily conflated, however.

Across prior means of $30$, $50$ and $100$, i.e.\ a factor of $3.3$, the
posterior median total moves by under one percent on each portfolio, within
the Monte-Carlo resolution of the runs, while the posterior mean of $c$
moves by about five percent. The prior sensitivity that
Appendix~\ref{app:bayesian} documents in simulation is a coverage property
measured over triangles and cannot be assessed on one realised triangle;
what can be measured here is the location, and it is data-dominated. On
these triangles the likelihood pins the concentration far more tightly than
the prior does, while the predictive intervals remain as wide as the level
of $c$ dictates. Note that a parameter being identified and a predictive
being calibrated are different properties, the benchmarks exhibiting the
former without the latter.

\subsection{Bornhuetter--Ferguson illustration}\label{sec:empirical_bf}

To demonstrate the anchor-agnostic nature of the construction, we apply
Algorithm~\ref{alg:bootstrap} to Taylor--Ashe under Bornhuetter--Ferguson
proportions, using first-period claims as the exposure proxy and
$q^{\mathrm{BF}} = 12$ as the prior loss ratio. The chosen value sits below
the rate implied by observed development, deliberately producing a BF point
estimate distinct from CL so that the contrast between the two anchoring
regimes is visible in the predictive intervals.

\begin{table}[ht]
\centering
\caption{Reserve estimates for Taylor--Ashe under BF proportions
($\hat{q}^{\mathrm{BF}} = 12$, $\hat{c} = 189.1$, $B = 5{,}000$).}
\label{tab:bf_reserves}
\begin{tabular}{@{}lrrrr@{}}
\toprule
Method & Mean & Std.\ Err & 95\% PI lower & 95\% PI upper \\
\midrule
BF (point only)    & 15{,}073 &      --- &      --- &      --- \\
Multinomial (BF)   & 15{,}068 &      387 & 14{,}318 & 15{,}841 \\
Multinomial (CL)   & 19{,}234 & 1{,}914 & 16{,}144 & 23{,}516 \\
\bottomrule
\end{tabular}
\end{table}

The construction accommodates BF without modification: $E_i q^{\mathrm{BF}}$
replaces $\Xobs_i / F_{I-i}$ as the anchoring quantity inside
Algorithm~\ref{alg:bootstrap}, and the IBNP draw becomes
$E_i q^{\mathrm{BF}} (1-W_i)$ rather than $\Xobs_i (1-W_i)/W_i$. The contrast
between the two anchoring regimes is pronounced: under BF the predictive
standard error is $387$ against $1{,}914$ under CL, and the $95\%$ interval
width $1{,}523$ against $7{,}372$ (all in thousands), i.e.\ a factor of
nearly five on both measures. The informative anchor concentrates the allocation and
sharply narrows the predictive spread, while the diffuse anchor widens it.

The narrowing is not free, and the value chosen here is instructive. Setting
$q^{\mathrm{BF}} = 12$ displaces the Bornhuetter--Ferguson point estimate
$22\%$ below the Chain-Ladder-anchored predictive mean, which is the magnitude
of prior misspecification examined in Section~\ref{sec:sim_anchor}: under a
correctly specified prior the anchored construction covers within two points
of the diffuse one on an interval half as wide, but a $20\%$ error in the
prior loss ratio reduces coverage to $37\%$ when the prior is too low. The
narrow interval displayed in Table~\ref{tab:bf_reserves} should therefore be
read as conditional on the prior being correct, not as a free improvement on
Chain-Ladder anchoring. On this triangle the prior is a modelling choice made
for illustration and its correctness is not established.

This illustrates the credibility structure on real data: the same triangle
combined with two different priors yields two different predictive
distributions, so that the choice of anchor is not only a point-estimate choice but
a choice of uncertainty regime, whose calibration, as
Section~\ref{sec:sim_anchor} shows, depends on that anchor being correct. The ODP residual
bootstrap has no direct analogue here: it resamples Pearson residuals from the
observed development structure, which the BF anchor, being an external prior loss ratio rather than a fitted
pattern, does not supply. Analytic BF prediction
errors exist \citep{Mack2008, AlaiMerzWuthrich2009}, but no conditional
predictive bootstrap for BF was available prior to this construction.

\section{Discussion}\label{sec:discussion}

The reserve of a portfolio is booked on one realised triangle, and the
uncertainty relevant to that booking is the conditional predictive
$p(R \mid \mathcal{D})$. We began from the view that the residual bootstraps target something
else, and the paper's clearest result is that, for a Chain-Ladder-anchored
paid triangle, they do not: resampling a triangle with free accident-year
effects draws the row levels from their diffuse-prior posterior and the
development pattern from its estimation distribution, so that the
bootstrap variance decomposes into the process, level and pattern
components of the conditional predictive
(Proposition~\ref{prop:boot_decomp}), and its coverage under the count
hierarchy is nominal at every frailty variance tested
(Section~\ref{sec:sim_frailty}). The conjecture that omitted frailty
explains the under-coverage \citet{Meyers2015} documents is rejected by
that test; the paid-data failure he reports is one of location, and
nothing in the dispersion of a residual bootstrap accounts for it.

Reaching the conditional object requires holding $\mathcal{D}$ fixed and
integrating over $\pi(\theta \mid \mathcal{D})$, and of the two requirements
only the second is under the modeller's control. The construction of this
paper satisfies the first exactly, and its unconditional deficit is the cost
of failing the second: supplied with the true parameters it covers at
$94.7$--$95.7\%$ at every triangle size with no trend
(Section~\ref{sec:oracle}), while integrating over the posterior of $c$
removes most of the $\hat{c}$-conditional miscalibration and leaves roughly
six points of unconditional deficit at $I = 10$ that no further modelling of
the same triangle removes (Section~\ref{sec:regularised}). On Chain-Ladder-anchored paid triangles the ODP bootstrap is the
better-calibrated procedure at practical triangle sizes
(Section~\ref{sec:sim_odp}), and given the proposition that is the
expected outcome rather than a trade-off between two targets: it
integrates over two of the three parameters the plug-in construction
substitutes. What
the analysis adds to that familiar trade is a statement about the operating
point: the information a calibrated conditional interval requires is not
present in a single paid triangle at the sizes at which macro-level
reserving is practised, so regularisation is the ordinary mode of the method
rather than a small-sample fallback, and a closed-form asymptotic constant
derived for a limit those triangles never approach
(Proposition~\ref{prop:c_variance}) describes a regime the practitioner
does not occupy (Remark~\ref{rem:sigma_c_finite}). Note that the latter caution applies to this paper's own asymptotics as much
as to the theory it replaces.

Two kinds of limit are hereby separated. The first is a limit of information,
i.e.\ what a single paid triangle determines about the dispersion structure,
and it is shared by every procedure considered here. The second is a limit of
approximation, internal to the closed form rather than to the data, and it
survives when every parameter is known exactly.

\subsection{Beyond the single paid triangle}\label{sec:extensions}

The residual bootstraps cannot identify the frailty dispersion from a
single triangle, and price the level at its diffuse limit instead
(Section~\ref{sec:sim_frailty}); the conditional construction of this
paper cannot identify the concentration well, and falls short of nominal
by roughly six points at $I = 10$ even when the concentration is
integrated over optimally (Section~\ref{sec:regularised}). Substituting a point estimate for any of the three parameters, i.e.\ $c$,
$\bpi$ or the pooled loss ratio $\hat q$, costs coverage by an amount we
measure in each case (Sections~\ref{sec:oracle},
\ref{sec:conditional_coverage} and \ref{sec:sim_anchor}). And even with all
parameters known exactly, the closed form remains conditionally miscalibrated
by roughly ten points at the ultimate dispersion typical of macro-level
portfolios, because it is exact only under diffuse anchoring
(Section~\ref{sec:conditional_coverage}).

The first three have a common cause: a single paid triangle carries one
observation per cell and, at the $I$ typical of macro-level reserving, does
not determine the dispersion structure to the accuracy a calibrated
conditional interval requires. The fourth is of a different kind: it is a property of the closed form and
not of the data, and it is in tension with the anchor-agnostic property
instead of being removable alongside the others. The
extensions below address the first three. Note that these extensions are not refinements of the method, but the only
remaining direction.

The first extension is to counts. The Poisson--Gamma--Multinomial hierarchy
of Section~\ref{sec:counts} shares the development vector $\bpi$ with the
amounts side: a joint model estimates $\bpi$ from two triangles rather than one, which addresses the $\hat\bpi$ component of the plug-in error, i.e.\ the
component visible as the $J$-degradation of Table~\ref{tab:sensitivity} and
as eleven of the twenty-six points of allocation-conditional miscalibration
in Section~\ref{sec:conditional_coverage}. It does not by itself identify $c$:
the count allocation carries the frequency dispersion $\kappa$, and the
severity allocation concentration is a separate parameter of the hierarchy.

The second extension is to incurred amounts. Case reserves constitute, by
construction,
information on the eventual cost of a claim before it is paid, and hence bear
directly on the allocation that $c$ governs. A coupled model for paid and
incurred development, in which the two allocations are constrained to share
structure, is the natural framework to incorporate this information, and a
principled alternative to the regression correction of Munich Chain-Ladder
\citep{QuargMack2004}.

The third extension is to other portfolios. The concentration $c$ is a single scalar, and the difficulty is the small
number of rows that a single triangle supplies for its estimation. Estimated hierarchically across a portfolio of
triangles it becomes a credibility problem in its own right:
B\"uhlmann--Straub applied to the dispersion parameter rather than to the
mean, the natural counterpart to the credibility correspondence for the mean
established in \citet{VanOirbeek2026cred}. That paper shows that the
conditional Dirichlet--Multinomial likelihood already identifies $c$ from a
single portfolio, that only $c$ and $\bpi$ need be shared across the pool, and
that pooling therefore serves precision rather than identification. The
closed-form variance \eqref{eq:sigma_c_explicit} supplies the per-portfolio precision that the credibility weight in that
construction requires, subject to the finite-sample correction of
Remark~\ref{rem:sigma_c_finite}, since at the $I$ typical of a pooled portfolio the asymptotic constant
understates the true variance by a factor of two to three.

\subsection{Limitations}\label{sec:limitations}

The limitations of the construction are of four kinds: assumptions on the
data, parameters that are substituted rather than integrated over, a cost
internal to the closed form, and the price of targeting the conditional
object at all.

The baseline model assumes stationary development proportions and does not
represent calendar-year effects or accident-year dependence. It also
requires positive incremental payments: the Dirichlet allocation has no
support on negative increments, so a row containing one contributes no
subcomposition to the likelihood, as happens on the RAA triangle
(Section~\ref{sec:empirical}). This excludes the construction from
triangles with systematic salvage or subrogation recoveries, where the
residual bootstrap, which tolerates negative increments as long as column
sums stay positive, remains the available procedure. Under
non-stationary development the construction maintains coverage at $J = 5$
(Table~\ref{tab:misspec_nonstat}) but only by widening: relative interval
width more than doubles from $\sigma_\delta = 0$ to $\sigma_\delta = 0.10$,
and the point-estimate bias grows from $2.8\%$ to $7.8\%$. Coverage degrades
at small $I$ ($81$--$85\%$ at $I = 7$, at every $c$) and, most severely, under
non-stationary development on a $J = 10$ triangle ($70.0\%$,
Table~\ref{tab:odp_comparison}), where $\hat\bpi$ is both misspecified and
weakly identified in the late columns. In the boundary regime
$\hat{c}F_{I-i} \approx 2$ the predictive is formally valid and practically
useless: Section~\ref{sec:sim_tweedie} reports nominal coverage delivered
by an interval $357$ times the true reserve. The tiering of
Remark~\ref{rem:validity} exists to keep the practitioner out of that regime,
and the operability bound of Section~\ref{sec:c_threshold} makes it checkable
before any reserve is struck.

Three parameters are substituted rather than integrated over, and each carries
a measured cost. The concentration $c$, wherever Algorithm~\ref{alg:bootstrap}
is used in place of Algorithm~\ref{alg:regularised}: $6.4$ coverage points at
$I = 10$, falling to $2.1$ at $I = 30$ (Section~\ref{sec:oracle}). The
development pattern $\hat\bpi$, throughout: eleven of the twenty-six points of
gradient across terciles of realised allocation dispersion
(Section~\ref{sec:conditional_coverage}). And the pooled loss ratio $\hat q$
under Cape Cod: $6.2$ points at $I = 10$, narrowing to $3.4$ at $I = 15$
(Section~\ref{sec:sim_anchor}).

The pattern is consistent and the remedy, integration rather than
substitution, is the same in each case. We implement it for one of the three. For $c$, integrating over the posterior reduces the
$\hat{c}$-conditional gradient from $18.3$ to $5.3$ coverage points at
unchanged interval width. For
$\bpi$ we demonstrate the remedy without adopting it: integrating over the
joint posterior of $(c, \bpi)$ recovers those eleven points, but at a $25\%$
cost in interval width, which is a materially different trade from
the free improvement available in $c$. For $\hat q$ we identify the cost and
do not address it.

Conditional calibration is therefore achieved with respect to the
concentration and, at a price, with respect to the development pattern; it is
not achieved with respect to the anchor.

The third limitation is of a different kind, because it is not a substitution
and no amount of integration removes it. The conditional law of
Theorem~\ref{thm:ibnp} is exact only under diffuse anchoring, and at the
ultimate dispersion typical of macro-level portfolios the resulting
approximation costs roughly ten coverage points across terciles of realised
allocation dispersion \emph{even when both parameters are known exactly}
(Section~\ref{sec:conditional_coverage}). The mechanism is identified: the
gradient falls to $1.4$ points, inside Monte-Carlo error, as the prior on the
ultimate approaches the diffuse limit in which the Beta representation is
exact.

This is the price of the closed form itself. The tilted density is available
(Remark~\ref{rem:diffuse_scope}) but does not yield the Beta representation on
which the anchor-agnostic property depends, so a construction sampling it
directly would be conditionally calibrated and would lose the property that
makes the present one applicable to Bornhuetter--Ferguson and Cape Cod without
modification. Note that we have not attempted that trade, and that both properties are in
tension rather than merely both unaddressed.

Finally, the coverage deficit of the plug-in construction at practical
triangle sizes is real and, as Section~\ref{sec:simulations} shows,
larger than that of the residual bootstrap on the anchor where both
apply. This is the cost of substituting parameters that the residual
bootstrap integrates over by resampling, and on Chain-Ladder-anchored
paid triangles there is no reason to pay it: the construction's domain is
the informative anchors and the identification questions, not a
replacement for the ODP bootstrap where the latter is correctly targeted.
It is not a
fixed cost: it falls from $13.7$ coverage points at $I = 7$ to $2.1$ at
$I = 30$ (Table~\ref{tab:oracle}), so the construction is unusable on the
shortest triangles, costly at $I = 10$ and essentially calibrated by
$I = 20$--$30$.

Finally, everything in this paper is stated at ultimate. The one-year
claims development result of Solvency~II, and the conditional MSEP of
that quantity derived by \citet{MerzWuthrich2008}, are not treated: the
closed form of Theorem~\ref{thm:ibnr} gives the predictive of the
full reserve, and its one-year projection would require the predictive of
next year's diagonal under the same hierarchy, which we leave open.

\section*{Acknowledgements}

A large language model (Anthropic's Claude)
was used as an assistive tool during the preparation of this manuscript. All
scientific content is the authors' own, and the authors take sole
responsibility for it.

\bibliographystyle{apalike}
\bibliography{references}

@article{aitchison1982,
  author  = {Aitchison, J.},
  title   = {The Statistical Analysis of Compositional Data},
  journal = {Journal of the Royal Statistical Society: Series B (Methodological)},
  volume  = {44},
  number  = {2},
  pages   = {139--160},
  year    = {1982}
}

@manual{ChainLadder2023,
  title  = {{ChainLadder}: Statistical Methods and Models for Claims Reserving in General Insurance},
  author = {Gesmann, M. and Murphy, D. and Zhang, Y. and Carrato, A. and Wuthrich, M. and Concina, F. and Dal Moro, E.},
  year   = {2023},
  note   = {R package version 0.2.18}
}

@article{Clark2003,
  author  = {Clark, D. R.},
  title   = {LDF Curve-Fitting and Stochastic Reserving: A Maximum Likelihood Approach},
  journal = {CAS Forum},
  year    = {2003},
  note    = {Fall},
  pages   = {41--92}
}

@book{DavisonHinkley1997,
  author    = {Davison, A. C. and Hinkley, D. V.},
  title     = {Bootstrap Methods and their Application},
  publisher = {Cambridge University Press},
  address   = {Cambridge},
  year      = {1997}
}

@article{england2002,
  author  = {England, P. D. and Verrall, R. J.},
  title   = {Stochastic claims reserving in general insurance},
  journal = {British Actuarial Journal},
  year    = {2002},
  volume  = {8},
  number  = {3},
  pages   = {443--518}
}

@book{Hall1992,
  author    = {Hall, P.},
  title     = {The Bootstrap and Edgeworth Expansion},
  publisher = {Springer},
  year      = {1992}
}

@article{Jorgensen1987,
  author  = {J{\o}rgensen, B.},
  title   = {Exponential dispersion models},
  journal = {Journal of the Royal Statistical Society, Series B},
  year    = {1987},
  volume  = {49},
  number  = {2},
  pages   = {127--162}
}

@article{mack1993,
  author  = {Mack, T.},
  title   = {Distribution-free calculation of the standard error of chain ladder reserve estimates},
  journal = {ASTIN Bulletin},
  year    = {1993},
  volume  = {23},
  number  = {2},
  pages   = {213--225}
}

@article{mack1994,
  author  = {Mack, T.},
  title   = {Which stochastic model is underlying the chain ladder method?},
  journal = {Insurance: Mathematics and Economics},
  year    = {1994},
  volume  = {15},
  number  = {2-3},
  pages   = {133--138}
}

@article{mackventer2000,
  author  = {Mack, T. and Venter, G.},
  title   = {A comparison of stochastic models that reproduce chain ladder reserve estimates},
  journal = {Insurance: Mathematics and Economics},
  year    = {2000},
  volume  = {26},
  number  = {1},
  pages   = {101--107}
}

@book{mccullagh1989,
  author    = {McCullagh, P. and Nelder, J. A.},
  title     = {Generalized Linear Models},
  publisher = {Chapman \& Hall},
  address   = {London},
  year      = {1989},
  edition   = {2}
}

@book{Meyers2015,
  author    = {Meyers, G. G.},
  title     = {Stochastic Loss Reserving Using {B}ayesian {MCMC} Models},
  publisher = {Casualty Actuarial Society},
  year      = {2015},
  series    = {CAS Monograph Series},
  number    = {1}
}

@article{renshaw1998,
  author  = {Renshaw, A. E. and Verrall, R. J.},
  title   = {A stochastic model underlying the chain-ladder technique},
  journal = {British Actuarial Journal},
  year    = {1998},
  volume  = {4},
  number  = {4},
  pages   = {903--923}
}

@article{SmythJorgensen2002,
  author  = {Smyth, G. K. and J{\o}rgensen, B.},
  title   = {Fitting Tweedie's Compound Poisson Model to Insurance Claims Data: Dispersion Modelling},
  journal = {ASTIN Bulletin},
  year    = {2002},
  volume  = {32},
  number  = {1},
  pages   = {143--157}
}

@article{srirams2021,
  author  = {Sriram, K. and Shi, P.},
  title   = {Stochastic loss reserving: A new perspective from a {D}irichlet model},
  journal = {Journal of Risk and Insurance},
  year    = {2021},
  volume  = {88},
  number  = {1},
  pages   = {195--230}
}

@article{taylor1983,
  author  = {Taylor, G. and Ashe, F.},
  title   = {Second moments of estimates of outstanding claims},
  journal = {Journal of Econometrics},
  year    = {1983},
  volume  = {23},
  number  = {1},
  pages   = {37--61}
}

@incollection{tweedie1984,
  author    = {Tweedie, M. C. K.},
  title     = {An index which distinguishes between some important exponential families},
  booktitle = {Statistics: Applications and New Directions},
  publisher = {Indian Statistical Institute},
  address   = {Calcutta},
  year      = {1984},
  pages     = {579--604}
}

@book{vdVaart1998,
  author    = {van der Vaart, A. W.},
  title     = {Asymptotic Statistics},
  publisher = {Cambridge University Press},
  year      = {1998}
}

@article{VanOirbeek2026nbcl,
  author        = {Van Oirbeek, R.},
  title         = {The {Negative Binomial Chain-Ladder}: A Full Likelihood Model for Claim Count Reserving},
  journal       = {CAS Forum},
  year          = {2026},
  number        = {1},
  eprint        = {2605.15811},
  archivePrefix = {arXiv},
  primaryClass  = {stat.ME}
}

@unpublished{VanOirbeek2026cnbg,
  author = {Van Oirbeek, Robin},
  title  = {What Paid Triangles Can and Cannot Identify: An Identification Boundary for Loss-Reserve Uncertainty},
  note   = {Working paper},
  year   = {2026}
}

@unpublished{VanOirbeek2026cred,
  author = {Van Oirbeek, R.},
  title  = {Classical Reserving Methods as Credibility Estimators: A Unified Bayesian Framework},
  year   = {2026},
  note   = {Manuscript in preparation}
}

@article{verrall2000,
  author  = {Verrall, R. J.},
  title   = {An investigation into stochastic claims reserving models and the chain-ladder technique},
  journal = {Insurance: Mathematics and Economics},
  year    = {2000},
  volume  = {26},
  number  = {1},
  pages   = {91--99}
}

@book{wuthrich2008,
  author    = {W{\"u}thrich, M. V. and Merz, M.},
  title     = {Stochastic Claims Reserving Methods in Insurance},
  publisher = {Wiley},
  address   = {Chichester},
  year      = {2008}
}

@article{Mack2008,
  author  = {Mack, Thomas},
  title   = {The Prediction Error of {Bornhuetter/Ferguson}},
  journal = {ASTIN Bulletin},
  year    = {2008},
  volume  = {38},
  number  = {1},
  pages   = {87--103}
}

@article{AlaiMerzWuthrich2009,
  author  = {Alai, Daniel H. and Merz, Michael and W{\"u}thrich, Mario V.},
  title   = {Mean Square Error of Prediction in the {Bornhuetter--Ferguson} Claims Reserving Method},
  journal = {Annals of Actuarial Science},
  year    = {2009},
  volume  = {4},
  number  = {1},
  pages   = {7--31}
}

@article{AlaiMerzWuthrich2010,
  author  = {Alai, Daniel H. and Merz, Michael and W{\"u}thrich, Mario V.},
  title   = {Prediction Uncertainty in the {Bornhuetter--Ferguson} Claims Reserving Method: Revisited},
  journal = {Annals of Actuarial Science},
  year    = {2010},
  volume  = {5},
  number  = {1},
  pages   = {7--17}
}

@article{QuargMack2004,
  author  = {Quarg, G. and Mack, T.},
  title   = {{Munich Chain Ladder}},
  journal = {Bl{\"a}tter der DGVFM},
  year    = {2004},
  volume  = {26},
  number  = {4},
  pages   = {597--630}
}

@article{BBMW2006,
  author  = {Buchwalder, M. and B{\"u}hlmann, H. and Merz, M. and W{\"u}thrich, M. V.},
  title   = {The Mean Square Error of Prediction in the Chain Ladder Reserving Method ({M}ack and {M}urphy Revisited)},
  journal = {ASTIN Bulletin},
  year    = {2006},
  volume  = {36},
  number  = {2},
  pages   = {521--542}
}

@article{MackQuargBraun2006,
  author  = {Mack, T. and Quarg, G. and Braun, C.},
  title   = {The Mean Square Error of Prediction in the Chain Ladder Reserving Method: A Comment},
  journal = {ASTIN Bulletin},
  year    = {2006},
  volume  = {36},
  number  = {2},
  pages   = {543--552}
}

@article{Gisler2006,
  author  = {Gisler, A.},
  title   = {The Estimation Error in the Chain-Ladder Reserving Method: A {B}ayesian Approach},
  journal = {ASTIN Bulletin},
  year    = {2006},
  volume  = {36},
  number  = {2},
  pages   = {554--565}
}

@article{GislerWuthrich2008,
  author  = {Gisler, A. and W{\"u}thrich, M. V.},
  title   = {Credibility for the Chain Ladder Reserving Method},
  journal = {ASTIN Bulletin},
  year    = {2008},
  volume  = {38},
  number  = {2},
  pages   = {565--600}
}

@article{Gisler2019,
  author  = {Gisler, A.},
  title   = {The Reserve Uncertainties in the Chain Ladder Model of {M}ack Revisited},
  journal = {ASTIN Bulletin},
  year    = {2019},
  volume  = {49},
  number  = {3},
  pages   = {787--821}
}

@article{MerzWuthrich2008,
  author  = {Merz, M. and W{\"u}thrich, M. V.},
  title   = {Modelling the Claims Development Result for Solvency Purposes},
  journal = {CAS E-Forum},
  year    = {2008},
  note    = {Fall 2008, 542--568}
}

@article{EnglerLindskog2024,
  author  = {Engler, N. and Lindskog, F.},
  title   = {Mack's Estimator Motivated by Large Exposure Asymptotics in a Compound {P}oisson Setting},
  journal = {ASTIN Bulletin},
  year    = {2024},
  volume  = {54},
  number  = {2},
  pages   = {310--326}
}

\appendix
\section{Proofs and derivations}\label{app:proofs}

This appendix collects the distributional results on which the body relies,
together with the moments used in the proof of
Proposition~\ref{prop:c_variance} and the regularised construction of
Section~\ref{sec:regularised}.

\subsection{Gamma--Dirichlet factorisation}

The construction of Section~\ref{sec:amounts} rests on the following
factorisation of a vector of independent Gamma variables into its sum and its
normalised components.

\begin{lemma}[Gamma--Dirichlet factorisation]\label{lem:gamma_dirichlet}
Let $Y_1,\ldots,Y_n \overset{\mathrm{iid}}{\sim} \Gam(\varphi,\lambda)$ and
$S = \sum_k Y_k$. Then $S \sim \Gam(n\varphi,\lambda)$ and
$(Y_1/S,\ldots,Y_n/S) \sim \Dir(\varphi,\ldots,\varphi)$, independently.
\end{lemma}

\begin{proof}
Transform $(y_1,\ldots,y_n) \mapsto (w_1,\ldots,w_{n-1},s)$ with
$w_k = y_k/s$, $s = \sum y_k$. The Jacobian is $s^{n-1}$. The joint density
factors into $\Gam(n\varphi,\lambda)$ for $s$ and
$\Dir(\varphi,\ldots,\varphi)$ for $(w_1,\ldots,w_n)$.
\end{proof}

\subsection{Dirichlet aggregation and subcomposition}

These two operations on a Dirichlet vector are easily conflated and carry
different concentrations. The distinction is what makes the moment estimator
of Section~\ref{sec:c_estimation} consistent.

\begin{lemma}[Aggregation]\label{lem:dirichlet_agg}
If $(W_0,\ldots,W_{J-1}) \sim \Dir(c\pi_0,\ldots,c\pi_{J-1})$ and
$\Wobs = \sum_{j=0}^k W_j$, then $\Wobs \sim \Bet(cF_k, c(1-F_k))$ with
$F_k = \sum_{j=0}^k \pi_j$.
\end{lemma}

\begin{proof}
Aggregating a subset of components of a Dirichlet vector yields a Dirichlet
vector in the summed parameters \citep{aitchison1982}, here
$\Dir(cF_k, c(1-F_k))$ on two components, which is $\Bet(cF_k, c(1-F_k))$. The
total concentration $c$ is unchanged, since $cF_k + c(1-F_k) = c$.
\end{proof}

\begin{lemma}[Subcomposition]\label{lem:dirichlet_sub}
If $(W_0,\ldots,W_{J-1}) \sim \Dir(c\pi_0,\ldots,c\pi_{J-1})$, then the
subcomposition formed by rescaling the first $k+1$ components to their own
total satisfies
\[
\left(\frac{W_0}{\sum_{l\le k} W_l},\ldots,
      \frac{W_k}{\sum_{l\le k} W_l}\right)
\;\sim\; \Dir(c\pi_0,\ldots,c\pi_k),
\]
independently of $\sum_{l \le k} W_l$. In particular each component is
\[
\frac{W_j}{\sum_{l\le k} W_l}
\;\sim\; \Bet\bigl(c\pi_j,\; c(F_k - \pi_j)\bigr)
\;=\; \Bet\bigl(a\,\pi_j^{(k)},\; a(1-\pi_j^{(k)})\bigr),
\qquad
\pi_j^{(k)} = \frac{\pi_j}{F_k},\quad a = cF_k,
\]
i.e.\ it has mean $\pi_j^{(k)}$ but \emph{concentration $a = cF_k$, not $c$}.
\end{lemma}

\begin{proof}
Standard property of the Dirichlet distribution \citep{aitchison1982}: a
subcomposition of a Dirichlet vector is Dirichlet in the corresponding
subvector of parameters, and is independent of the subtotal. The Beta marginal
follows from Lemma~\ref{lem:dirichlet_agg} applied within the subcomposition,
whose total concentration is $\sum_{j \le k} c\pi_j = cF_k$.
\end{proof}

\begin{remark}[Aggregation versus subcomposition]\label{rem:agg_vs_sub}
The two lemmas govern different quantities. The bootstrap of
Algorithm~\ref{alg:bootstrap} uses aggregation: $\Wobs_i$ is the sum of the
observed development shares of the full ultimate, and its concentration is
$c$. The moment estimator of Section~\ref{sec:c_estimation} uses
subcomposition: $W_{ij}^{(k)}$ rescales a cell by a \emph{partial} row total,
and its concentration is $cF_k$. Substituting the first where the second is
required leaves an estimator whose probability limit is $cF_{k^*}$ rather than
$c$, at the median contributing horizon $k^*$.

The magnitude of the resulting understatement depends on the development
pattern and on which cells contribute; it is quantified after
Proposition~\ref{prop:c_estimator}, and it does not diminish as $I$ grows.
\end{remark}

\subsection{Beta central moments used in Proposition~\ref{prop:c_variance}}

The delta-method computation in Proposition~\ref{prop:c_variance} is carried
out on the subcomposition of Lemma~\ref{lem:dirichlet_sub}, whose
concentration is $a = cF_k$ and not $c$. The moments below are therefore
stated in $a$; setting $a = c$ recovers the full-composition case.

For $W \sim \Bet(a\pi, a(1-\pi))$ with $\pi \in (0,1)$, $a > 0$:
\begin{align}
\mu_2(W) &= \frac{\pi(1-\pi)}{a+1}, \label{eq:beta_mu2}\\
\mu_3(W) &= \frac{2\pi(1-\pi)(1-2\pi)}{(a+1)(a+2)}, \label{eq:beta_mu3}\\
\mu_4(W) - \mu_2(W)^2
&= \frac{2\pi(1-\pi)\bigl[(a^2 - 10a - 12)\,\pi(1-\pi) + 3(a+1)\bigr]}
        {(a+1)^2(a+2)(a+3)}. \label{eq:beta_mu4}
\end{align}
These follow from the Beta raw moments
$\E[W^r] = \prod_{s=0}^{r-1}(a\pi+s)/(a+s)$ by direct computation. The
delta-method step uses only these three quantities and the gradient of
$f(m, s^2) = m(1-m)/s^2 - 1$, which is the moment estimator of $a$; dividing
by $\hat{F}_k$ converts it to an estimator of $c$ and scales the asymptotic
variance by $F_k^{-2}$.

\begin{remark}[Where the fourth moment enters twice]\label{rem:mu4_twice}
Expression \eqref{eq:beta_mu4} appears in two distinct roles, which should
be kept apart. In Proposition~\ref{prop:c_variance} it supplies the
asymptotic variance of the moment estimator $\hat{c}_{jk}$, through the
sampling variability of $\hat{s}^2$. In Remark~\ref{rem:validity} the existence of the fourth moment of
$(1-W_i)/W_i$, which is a different random variable, Beta-prime rather than
Beta, determines whether a reported standard error has a finite-sample
guarantee. The two are unrelated: the first
concerns the estimator of the concentration, the second the predictive of the
reserve.
\end{remark}

\subsection{Regularisation in the scarce-information regime}\label{app:bayesian}

The plug-in moment estimator of Section~\ref{sec:c_estimation} requires at
least three rows at some horizon, i.e.\ $n_k = I - k - 1 \geq 3$ for some
$k \geq 1$, and is therefore undefined for $I < 5$. Where it is defined it is
both biased and noisy at realistic $I$ (Remark~\ref{rem:sigma_c_finite}) and
undercovers accordingly. This is not a small-triangle edge case. Macro-level
triangles are $I \times J$ with $I$ typically $7$--$15$ and one observation per
cell; the asymptotics under which a frequentist estimator would identify the
dispersion structure do not engage at these sizes. Regularisation is thus the ordinary operating point of the method and not a
fallback, which is why Algorithm~\ref{alg:regularised} is presented in the
body rather than here,
with its coverage by triangle size in Table~\ref{tab:bayes_coverage} of
Section~\ref{sec:regularised}.

We place a weakly-informative prior $\log c \sim \mathcal{N}(\log 50, 1)$ and
sample the posterior of $c$ by Metropolis-within-Gibbs from the exact
subcomposition likelihood of Lemma~\ref{lem:dirichlet_sub}, holding the
triangle fixed and drawing the predictive
$W_i \sim \Bet(cF_{I-i}, c(1-F_{I-i}))$ under each posterior draw. Partially
developed rows enter exactly: a row observed through development year $k$
contributes its subcomposition, which is $\Dir(c\pi_0,\ldots,c\pi_k)$ with
total concentration $cF_k$. This is Algorithm~\ref{alg:bootstrap} with $c$
integrated over its posterior rather than plugged in, and hence within the
same conditional family.

This appendix reports the prior sensitivity of that construction; the coverage
comparison itself is in Section~\ref{sec:regularised}.

\begin{table}[ht]\centering
\caption{Prior sensitivity of the regularised construction. At $I=5$ coverage
moves with the prior mean by up to $9.0$ points; at $I=10$ the likelihood
dominates and sensitivity is at most $4.2$ points, with no systematic
direction.}
\label{tab:bayes_prior}
\begin{tabular}{@{}lcccccc@{}}
\toprule
& \multicolumn{3}{c}{$I=5$} & \multicolumn{3}{c}{$I=10$} \\
\cmidrule(lr){2-4}\cmidrule(lr){5-7}
$c$ & $\log 30$ & $\log 50$ & $\log 100$ & $\log 30$ & $\log 50$ & $\log 100$ \\
\midrule
 20 & 73.2 & 70.8 & 68.8 & 91.6 & 88.4 & 88.0 \\
 50 & 77.0 & 72.8 & 70.0 & 89.0 & 90.2 & 89.0 \\
100 & 82.4 & 76.0 & 73.4 & 89.8 & 88.0 & 88.6 \\
200 & 85.4 & 83.8 & 76.8 & 95.0 & 90.8 & 92.8 \\
\bottomrule
\end{tabular}
\end{table}

Prior sensitivity is a scarce-data phenomenon: the prior carries weight
exactly where the data are thin, and not otherwise. The residual deficit is informative as well. Even integrating over the posterior of $c$,
coverage at $I = 10$ reaches only $89.2\%$
(Table~\ref{tab:bayes_coverage}), and the posterior mean of $c$ is $61.8$ against a true $50$, i.e.\ a $24\%$
overstatement against the plug-in's $31\%$. Integration reduces the bias without removing it, because the bias is a
property of the likelihood at $n_k \leq 8$ rather than of the estimator
applied to it. The gap to nominal is the information bound of paid-only
reserving: it is not removed by using more of the model, since a single paid
triangle does not contain the information that the frequentist limit requires.

\subsection{Code and reproducibility}\label{app:code}

Complete R implementations, the simulation harness, the per-cell checkpoints
backing every simulated figure in
Sections~\ref{sec:simulations}--\ref{sec:empirical}, and an MD5 manifest of
those checkpoints are available at \url{https://github.com/robin-vo/hgr}.
Every simulated number in this paper is reproducible from that harness by a
single driver script, which runs the full study in under an hour on eight
cores, dominated by the MCMC of this appendix. The harness enforces
per-replication timeouts, atomic checkpointing with restart, and
load-balanced dispatch. The count-hierarchy frailty sweep of Section~\ref{sec:sim_frailty} and
the allocation comparison of Table~\ref{tab:frailty_amounts} use an
in-repository implementation of the ODP bootstrap
\citep{england2002} that was validated against the \texttt{ChainLadder}
implementation on the Taylor--Ashe triangle at $B = 10{,}000$, agreeing on
the mean, standard deviation and five quantiles of the reserve to within
$1.4\%$; the validation output is in the repository. The \texttt{ChainLadder}
implementation is not used for those sweeps because its overdispersed
Poisson process draw is undefined when the estimated dispersion is at or
below one, which is the regime of Poisson counts. Each script writes its
table and the figures quoted in the text, which were transcribed into the
manuscript from that output without rounding beyond what the script
applies. A separate script attaches a measured value to each
asymptotic constant stated in this paper and reports any that its run does not
support; its output is included in the repository.

\end{document}